%% file: main.tex
\documentclass[a4paper]{article}
\usepackage{fullpage}
\usepackage[T1]{fontenc}
\usepackage[utf8]{inputenc}
\usepackage[all=normal,bibliography=tight]{savetrees}

\usepackage{amstext,amsfonts,amsthm,amsmath,amssymb}

\usepackage{graphicx,tikz}
\usetikzlibrary{snakes,shapes}
\usetikzlibrary{calc}
\usetikzlibrary{decorations.markings}
\usetikzlibrary{patterns}
\usepackage{comment}
\usepackage{url}
\usepackage{xspace}
\usepackage[textsize=footnotesize,backgroundcolor=white]{todonotes}
\usepackage[shortlabels]{enumitem}
\setlist{topsep=1ex,itemsep=-1ex,partopsep=0ex,parsep=1ex}
\usepackage[absolute]{textpos}

\definecolor{blue}{rgb}{0.05,0.1,0.4}
\definecolor{brown}{rgb}{0.35,0.35,0.1}
\usepackage[ocgcolorlinks, linkcolor={blue}, citecolor={brown}]{hyperref}

\usepackage[capitalize]{cleveref}

\usepackage{wrapfig}
\usepackage{rotating}
\usepackage{pdflscape}

\def\cqedsymbol{\ifmmode$\lrcorner$\else{\unskip\nobreak\hfil
\penalty50\hskip1em\null\nobreak\hfil$\lrcorner$
\parfillskip=0pt\finalhyphendemerits=0\endgraf}\fi}

\newcommand{\Oh}{\mathcal{O}}
\newcommand{\smalloh}{o}

\newcommand{\T}{\mathcal{T}}
\newcommand{\I}{\mathcal{I}}
\newcommand{\F}{\mathcal{F}}
\newcommand{\X}{\mathcal{X}}
\newcommand{\C}{\mathcal{C}}

\newcommand{\defparproblem}[4]{\par
 \vspace{3mm}
\noindent\fbox{
 \begin{minipage}{0.96\textwidth}
 \begin{tabular*}{\textwidth}{@{\extracolsep{\fill}}lr} #1 & {\bf{Parameter:}} #3 \vspace{1mm} \\ \end{tabular*}
 {\textbf{Input:}} #2%
	\vspace{1mm}\\%
 {\textbf{Question:}} #4%
 \end{minipage}
 }
 \vspace{3mm}
\par
}
\newcommand{\defpartask}[4]{\par
 \vspace{3mm}
\noindent\fbox{
 \begin{minipage}{0.96\textwidth}
 \begin{tabular*}{\textwidth}{@{\extracolsep{\fill}}lr} #1 & {\bf{Parameter:}} #3 \vspace{1mm} \\ \end{tabular*}
 {\textbf{Input:}} #2%
	\vspace{1mm}\\%
 {\textbf{Task:}} #4%
 \end{minipage}
 }
 \vspace{3mm}
\par
}

\newcommand{\compath}{\textsc{ComPath}\xspace}
\newcommand{\comvdpaths}{\textsc{ComVDP}\xspace}
\newcommand{\simplecomvdpaths}{\textsc{SComVDP}\xspace}
\newcommand{\comdeto}{\textsc{ComDetour}\xspace}

\newtheorem{lemma}{Lemma}[section]

\newtheorem{observation}[lemma]{Observation}
\newtheorem{corollary}[lemma]{Corollary}
\newtheorem{theorem}[lemma]{Theorem}
\newtheorem{claim}[lemma]{Claim}
\theoremstyle{definition}
\newtheorem{definition}[lemma]{Definition}
\newtheorem*{reduction*}{Reduction Rule}
\newtheorem{construction}[lemma]{Construction}

\DeclareMathOperator{\dist}{dist}

\newcommand{\Ecol}{\lambda}
\newcommand{\Tcol}{\zeta}

\crefname{lemma}{Lemma}{Lemmas}
\crefname{theorem}{Theorem}{Theorems}

\title{The Complexity of Connectivity Problems in Forbidden-Transition Graphs and Edge-Colored Graphs\thanks{This research is a part of a project that have received funding from the European Research Council (ERC)
under the European Union's Horizon 2020 research and innovation programme
Grant Agreement 714704. Parts of Manuel Sorge's work were performed while visiting TU Vienna, Austria.}}

\author{
         Thomas Bellitto\footnote{Faculty of Mathematics, Informatics and Mechanics, University of Warsaw, Poland, \texttt{tbellitto@mimuw.edu.pl}}
    \and Shaohua Li\footnote{Faculty of Mathematics, Informatics and Mechanics, University of Warsaw, Poland, \texttt{S.Li@mimuw.edu.pl}}
    \and Karolina Okrasa\footnote{Faculty of the Mathematics and Information Science, Warsaw University of Technology and Faculty of Mathematics, Informatics and Mechanics, University of Warsaw, Poland, \texttt{k.okrasa@mini.pw.edu.pl}}
    \and Marcin Pilipczuk\footnote{Faculty of Mathematics, Informatics and Mechanics, University of Warsaw, Poland, \texttt{malcin@mimuw.edu.pl}}
    \and Manuel Sorge\footnote{Faculty of Mathematics, Informatics and Mechanics, University of Warsaw, Poland, \texttt{manuel.sorge@mimuw.edu.pl}}
}

\date{}

\begin{document}

\maketitle

\begin{textblock}{20}(0, 12.0)
\includegraphics[width=40px]{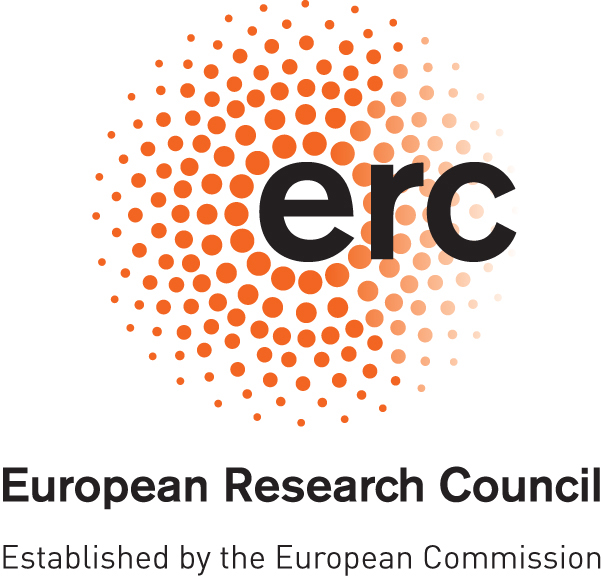}%
\end{textblock}
\begin{textblock}{20}(0, 12.9)
\includegraphics[width=40px]{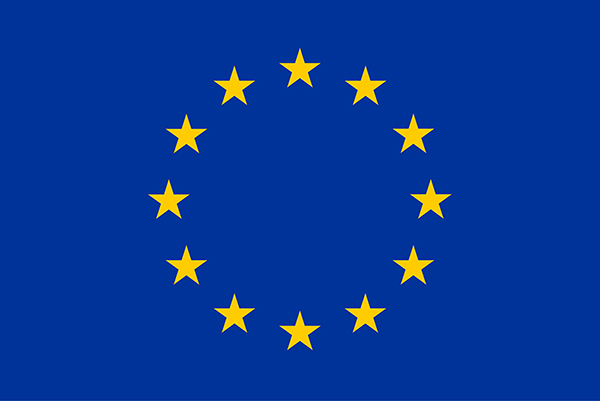}%
\end{textblock}
\begin{abstract}
\input{abstract}

\end{abstract}

\section{Introduction}

\input{intro}
\section{Preliminaries}

\input{prelims}

\section{Detours}\label{sec:detours}
\input{detours}

\section{Graph width parameters}\label{sec:parameters}

In this section we give our results pertaining to graph-width measures.
As outlined in the introduction, finding a compatible \(s\)-\(t\) path of length at most~\(k\) is fixed-parameter tractable with respect to~\(k\) (see \cref{thm:comp-path-fpt-length}).
Because the length of a simple path is upper-bounded by functions of the smallest size of a vertex cover and of the treedepth, tractability for these two parameters also follows.
In \cref{sec:lowerbound} we give a limit to this avenue of proving tractability results for successively stronger parameters:
We prove that detecting compatible \(s\)-\(t\) paths and related problems are W[1]-hard with respect to the size of modulators to constant treewidth.
In \cref{sec:treecutwidth} we show that detecting compatible \(s\)-\(t\) paths is fixed-parameter tractable with respect to the treecut-width.
Then, in \cref{sec:edgecol-treewidth} we turn to edge-colored graphs and give a fixed-parameter algorithm for detecting properly colored Hamiltonian cycles parameterized by the treewidth.

\subsection{Modulator to linear forest}\label{sec:lowerbound}

\input{lower-bound}

\subsection{Treecut-width}\label{sec:treecutwidth}
\input{treecutwidth}

\subsection{Edge-colored graphs and treewidth}\label{sec:edgecol-treewidth}
\input{properly-tw}

\section{Two disjoint shortest paths}\label{sec:2dsp}

\input{2DSPP}

\subsection{Edge-disjoint case}
\input{Edge2DSPP}

\subsection{Vertex-disjoint case}

\input{Vertex2DSPP}
\section{Conclusions}
\input{conclusions}

\bibliographystyle{abbrv}

\bibliography{../ref}

\end{document}

%% file: abstract.tex
The notion of \emph{forbidden-transition graphs} allows for a robust generalization of walks
in graphs. 
In a forbidden-transition graph, every pair of edges incident to a common vertex is 
\emph{permitted} or \emph{forbidden}; a walk is \emph{compatible} if all pairs of consecutive
edges on the walk are permitted. 
Forbidden-transition graphs and related models have found applications in a variety of fields, such as routing in optical telecommunication networks, road networks, and bio-informatics.

\looseness=-1
We initiate the study of fundamental connectivity problems from the point of view of parameterized complexity, including an in-depth study of tractability with regards to various graph-width parameters. Among several results, we prove that finding a simple compatible path between given endpoints in a forbidden-transition graph is $W[1]$-hard when parameterized by the vertex-deletion distance to a linear forest (so it is also hard when parameterized by pathwidth or treewidth).
On the other hand, we show an algebraic trick that yields tractability when parameterized by treewidth of finding a properly colored Hamiltonian cycle in an edge-colored graph;
properly colored walks in edge-colored graphs is one of the most studied special cases
of compatible walks in forbidden-transition graphs.

%% file: intro.tex
Graphs have proved to be an extremely useful tool to model routing problems in a very wide range of applications.
However, we sometimes need to express constraints on the permitted walks that are stronger than what the standard graph model allows for.
For example, in a road network, there can be a crossroad where drivers are not allowed to turn right.
In this case, many walks in the underlying graph without transition restrictions would correspond to routes that a driver is not allowed to use.
To overcome this limitation, Kotzig introduced forbidden-transition graphs in \cite{Kotzig}.
Let $G$ be an undirected graph.
A \emph{transition} in $G$ is an unordered pair of adjacent edges.
Every time a walk in $G$ uses two edges $uv$ and $vw$ consecutively, we say that the walk \emph{uses the transition~$\{uv,vw\}$}.
A \emph{transition system} of $G$ is a set of transitions in~$G$.
A \emph{forbidden-transition graph} is a tuple $(G, T)$ of a graph~$G$ together with a transition system~$T$ of $G$.\footnote{Our notation rather suggests that $(G, T)$ is a \emph{permitted-transition graph} but we use forbidden transitions in keeping with convention in the literature.}
We say that a transition is \emph{permitted} if it is in $T$ and it is \emph{forbidden} otherwise.
We say a walk is \emph{compatible} with $T$ or \emph{$T$-compatible} if all the transitions it uses are permitted, that is, in $T$.
We omit reference to $T$ when it is clear from the context.
For notational clarity, it is sometimes useful to refer to the transitions~$T(v)$ of a specific vertex $v \in V(G)$, that is, $T(v) = \{\{e, f\} \in T \mid e \cap f = \{v\}\}$.

\looseness=-1
Since their introduction, forbidden-transition graphs and related models have found applications in a variety of fields, such as routing in optical telecommunication networks \cite{Ahmed}, road networks \cite{sctm}, and bio-informatics \cite{Dorninger}.
Problems of routing, connectivity, and robustness in those graphs have received a lot of attention but unfortunately, those problems generally turn out to be algorithmically very difficult, even on very restricted subclasses of graphs.
In \cite{Szeider}, Szeider famously proved that even determining the existence of a compatible
   (elementary) path between two given vertices of a forbidden-transition graph is NP-complete.
Similarly, many known results about forbidden-transition graphs are proofs of NP-completeness of problems that are polynomially solvable on standard graphs (e.g.\ \cite{DBLP:journals/tcs/AbouelaoualimDFMMS08}, \cite{DBLP:conf/wg/BellittoB18}, \cite{DBLP:journals/dm/Dvorak09}, \cite{DBLP:journals/dam/GourvesLMM13}, \cite{DBLP:journals/endm/GourvesLMMP09},  \cite{DBLP:conf/tamc/KanteLM13}, \cite{DBLP:conf/wg/KanteMMN15}, \cite{Szeider}).

A very interesting specific case of compatible walks in forbidden-transition graphs are properly colored walks in edge-colored graphs.
Here, a graph is given together with a coloring of its edges and we say that a walk is \emph{properly colored} if it does not use consecutively two edges of the same color.
These graphs have been introduced by Dorninger in \cite{Dorninger} to study chromosome arrangements.
They are a powerful generalization of directed graphs (see \cite{jbjgutin}) and have been studied by many authors since their introduction.
The problem of properly colored Hamiltonian cycles was the first problem studied on edge-colored graphs and this problem and its variants (such as longest elementary cycle or spanning trails among many others) are especially well studied in the literature.
We refer the reader to \cite{GutinKim} or \cite{jbjgutin} for surveys on these problems and to \cite{supereulerian}, \cite{DBLP:journals/dam/Contreras-Balbuena17}, \cite{DBLP:journals/dmtcs/Contreras-Balbuena19}, \cite{ChinesePostman_2017}, \cite{DBLP:journals/dm/LiBXZ17} or \cite{DBLP:journals/gc/LiBZ19} for recent developments.

Because of their expressiveness and wide range of applications, the study of forbidden-transition graphs is a fast-emerging field and has been the subject of growing attention in the past decades but we are still very far from understanding them as well as regular graphs. 
Our aim in this paper is to study the parameterized complexity of some known NP-complete problems, in general forbidden-transition graphs as well as in the specific case of edge-colored graphs. We specifically focus on some problems of great practical interest, such as the existence of an elementary path or the length of a shortest path between given vertices, the problem of Hamiltonian cycles, or linkage problems where we try to connect pairs of vertices by vertex- or edge-disjoint paths. A very rich toolbox already exists to study fixed-parameter tractability in standard graphs (see \cite{cygan2015parameterized} for example) but the generalization of these concepts to forbidden-transition graphs is widely unexplored and raises many challenges that we hope to see get more attention in the future.
\paragraph{Our results.}%
In Section~\ref{sec:detours}, we study the problem of shortest compatible paths between two vertices $s$ and $t$ in a forbidden-transition graph.
Recall that determining whether there exists a compatible path between $s$ and $t$ is known to be NP-complete~\cite{Szeider}.
A simple application of the color-coding technique shows that this problem is fixed-parameter
tractable when parameterized by the length of the path. 
We improve upon this observation by showing 
that the complexity of finding a shortest compatible path from $s$ to $t$ is actually fixed-parameter tractable when parameterized by the length of the detour that the forbidden transitions impose.
In other words, determining whether there exists a compatible path of length at most $d(s,t)+k$ where $d(s,t)$ is the length of the shortest path between $s$ and $t$ in the underlying graph with no forbidden transitions, is fixed-parameter tractable when parameterized by $k$.
Our algorithm follows the main ideas of the algorithm for the \textsc{Exact Detour} problem
by Bez\'{a}kov\'{a} et al.~\cite{BezakovaCDF2019}.

In Section \ref{sec:parameters}, we turn our attention to graph width parameters.
The rich ecosystem of relevant graph-width parameters is depicted on \cref{fig:parameter-hierarchy}; see \cite{marx_immersions_2014,Wollan2015,GanianKS2015} for the corresponding boundedness and unboundedness relations on treecut-width.

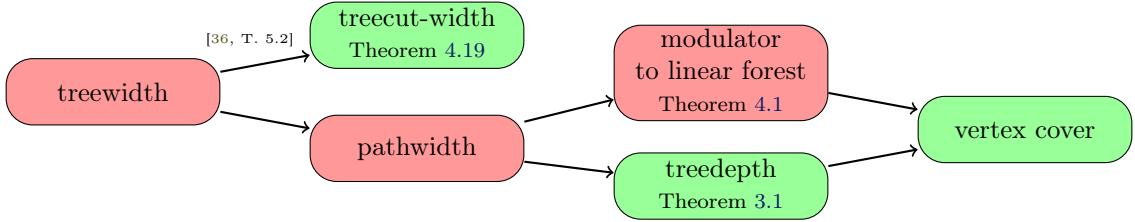
\begin{figure}
  \begin{center}
    \tikzstyle{para}=[draw,rectangle,fill=white,inner sep=0pt,minimum width=80pt,minimum height=25pt,rounded corners=10pt]
    \tikzstyle{ref}=[]
    \tikzstyle{whard}=[fill=red!40]
    \tikzstyle{fpt}=[fill=green!40]
    \tikzstyle{ref}=[]
    \begin{tikzpicture}
\node[para,whard] (tw) at (0,0) {treewidth};
\node[para,whard] (pw) at (4,-0.75) {pathwidth};
\node[para,fpt] (tcw) at (4,0.75) {\begin{tabular}{c}treecut-width \\ {\footnotesize \cref{thm:treecut-width}}\end{tabular}
};
\node[para,whard] (mtlf) at (8,0.25) {\begin{tabular}{c} modulator \\ to linear forest\\ {\footnotesize\cref{thm:lb1}}\end{tabular}};
\node[para,fpt] (td) at (8,-1.25) {\begin{tabular}{c} treedepth \\ {\footnotesize \cref{thm:comp-path-fpt-length}}  \end{tabular}
};
\node[para,fpt] (vc) at (12,-0.5) {vertex cover};
\draw[thick,->] (tw) -- node[ref,above,xshift=-2mm,yshift=1mm] {\tiny\cite[T.\ 5.2]{marx_immersions_2014}} (tcw);
\draw[thick,->] (tw) -- (pw);
\draw[thick,->] (pw) -- (td);
\draw[thick,->] (pw) -- (mtlf);
\draw[thick,->] (mtlf) -- (vc);
\draw[thick,->] (td) -- (vc);
\end{tikzpicture}%
\end{center}%
\caption{A hierarchy of graph-width parameters considered in this work.
  An arrow from $a$ to~$b$ represents the fact that a bound on parameter $b$ imposes a bound on parameter $a$, but there exist families of graphs with bounded $a$ and unbounded $b$.
  We color a parameter~$a$ green if detecting a compatible $s$-$t$ path is fixed-parameter tractable with respect to $a$ and red if it is W[1]-hard.}
\label{fig:parameter-hierarchy}
\end{figure}

First, we focus on the NP-complete problem of determining whether there exists a compatible path between $s$ and $t$ in a forbidden-transition graph. 
Since the problem is fixed-parameter tractable when parameterized by the length of the path (see \cref{sec:detours}),
it is also fixed-parameter tractable when parameterized by the vertex cover number
or the treedepth of the graph, as bounding the vertex cover number or the treedepth
of the graph by $k$ bounds the length of the longest simple path by $2k$ or $2^k-1$, respectively.
Our main result is a negative one: the problem becomes $W[1]$-hard
if one makes one step further to the parameter \emph{modulator to a linear forest}, 
i.e., the number of vertices one has to remove from the graph to turn it into a union of vertex-disjoint paths.
A small tweak of the reduction shows that finding a Hamiltonian cycle is $W[1]$-hard with respect to the size of a modulator to treewidth $2$.
Our reduction in particular implies hardness for the parameters pathwidth and treewidth (for both the compatible path and Hamiltonian cycle problems).

On the other hand, we show that if one considers parameters based on edge cuts
(as opposed to vertex cuts, like in treewidth), one can obtain nontrivial tractability results. 
\emph{Treecut-width} is a width notion based on edge cuts, introduced by Wollan~\cite{Wollan2015},
  and playing the role of treewidth in the world of the immersion relation.
We prove that the problem of finding a compatible $s$-$t$ path is fixed-parameter tractable when parameterized by
the treecut-width of the graph.
More precisely, the problem can be solved in time $k^{\Oh(k^2)}\cdot n^2 + \Oh(n^3) + \Oh\left((4^k\cdot k!)^{\Oh(3k+1)}\right) \cdot n^2$ where $k$ denotes the treecut-width.

In the light of the hardness in general forbidden-transition graphs of detecting $s$-$t$ paths, the most fundamental connectivity problem, we move to the special case of properly colored paths in edge-colored graphs.
As finding a (simple) properly colored path between given endpoints in an edge-colored graph is polynomial-time solvable, we focus on the problem of finding a Hamiltonian cycle.
We introduce an algebraic trick that shows that in edge-colored graphs, finding a properly colored Hamiltonian cycle is fixed-parameter tractable when parameterized by the treewidth of the graph.
More specifically, the problem can be solved in time $2^{\mathcal O(k)}\cdot (|V(G)|+|V(\mathcal T)|+\ell)$ where $k$ is the treewidth, $\mathcal T$ is the tree of the decomposition and $\ell$ is the number of different colors the edges can have.
The crucial property of the result is that $\ell$, the number of colors, is \emph{not} required to be bounded in the parameter and does not appear in the exponential part of the running-time bound.

After discussing graph-width notions, in Section \ref{sec:2dsp}, 
we move to the \textsc{Disjoint Paths} problem. 
In this problem, we are given a directed graph and a sequence $(s_1,t_1), (s_2,t_2), \ldots, (s_r,t_r)$ of terminal pairs;
the goal is to find compatible paths $P_1, P_2, \ldots,P_r$ such that $P_i$ starts in $s_i$ and ends
in $t_i$ and the paths $P_i$ are pairwise edge- or vertex-disjoint.

Observe that the problem quickly becomes hard. 
Even the setting of properly colored paths in edge-colored graphs generalizes directed
graphs\footnote{Consider the reduction that adds an in-neighbor to each $s_i$ and an out-neighbor to each $t_i$, replaces each terminal by the corresponding in- or out-neighbor and then replaces each directed edge $e$ with two undirected edges with two colors according to the direction of~$e$.} and the \textsc{Disjoint Paths} problem for $r=2$ is NP-hard in directed graphs~\cite{FortuneHW80}. Furthermore, in general graphs with transitions the case $r=1$ is NP-hard.
Hence, we focus on the specific case where the path $P_i$ is required to be a shortest
$s_i$-$t_i$ path, even in the unrestricted graph.
In directed graphs, a tractability result for this problem has been obtained
by B\'{e}rczi and Kobayashi~\cite{DBLP:conf/esa/Berczi017} for $r=2$.
This problem is currently a very active topic and new algorithms have been found very recently for several variants in the case $r=2$. 
Polynomial algorithms have been developed by Gottschau et al. \cite{Gottschau} and by Kobayashi and Sako \cite{KobayashiSako}  for undirected graphs with non-negative weighted edges and by Bang-Jensen et al. \cite{DBLP:journals/tcs/Bang-JensenBLY20} in the directed unweighted case where paths do not have to be shortest but have bounded lengths.
The complexity of the problem is still open for $r \geq 3$. 

Thus, in this work we focus on the case $r=2$ in directed forbidden-transition graphs. 
Extending the results of B\'{e}rczi and Kobayashi~\cite{DBLP:conf/esa/Berczi017},
we show that the problem remains polynomial-time solvable both in edge- and vertex-disjoint cases.
The arguments are presented in Section~\ref{sec:2dsp}.

%% file: prelims.tex
For each $n \in \mathbb{N}$ we use $[n]$ to denote $\{1, 2, \ldots, n\}$.
Unless stated otherwise, all graphs are undirected, without self-loops and parallel edges.

Let $G$ be an undirected graph.
By $V(G)$ and $E(G)$ we denote the vertex and edge set of $G$, respectively.
For each $v \in V(G)$ we denote by $E_G(v)$ the set of edges in $G$ that are incident with~$v$ in $G$.
We omit the subscript $G$ if it is clear from the context.
A \emph{walk} in $G$ is a sequence $(v_1,e_1,v_2,e_2,\ldots,e_\ell,v_{\ell+1})$ where $v_i$s are vertices of $G$, $e_i$s are edges of $G$,
  and for every $1 \leq i \leq \ell$, the vertices $v_i$ and $v_{i+1}$ are the two endpoints of the edge $e_i$.
A walk is \emph{closed} if its first vertex is also its last vertex.
The \emph{length} of a walk $W$ equals $\ell$, the number of edges in~$W$.
A \emph{path} is a walk in which no vertex occurs twice, a \emph{cycle} is a closed walk in which no vertex occurs twice except the first and last vertex.
Usually we will denote paths and cycles simply by their sequence of vertices.
By $\dist_G(s,t)$ we mean the length of a simple $s$-$t$ path in $G$ (ignoring any transitions).

For a graph $G$, a \emph{tree decomposition} of $G$ is a pair $(\T, \beta)$
where $\T$ is a tree and $\beta : V(\T) \to 2^{V(G)}$ such that the following holds:
(i) for every $v \in V(G)$, the set $\{t \in V(\T)~|~v \in \beta(t)\}$
induces a nonempty connected subtree of $\T$, and (ii)
  for every $uv \in E(G)$, there exists $t \in V(\T)$ with $u,v \in \beta(t)$.
That is, the function $\beta$ assigns to every node $t \in V(\T)$ a 
subset $\beta(t) \subseteq V(G)$, often called a \emph{bag}. 
It is often convenient to root $\T$ at an arbitrary vertex.
The width of a tree decomposition $(\T,\beta)$ equals $\max_{t \in V(\T)} |\beta(t)|-1$,
and the treewidth of a graph is the minimum possible width of its tree decomposition.

%% file: detours.tex
As Szeider~\cite{Szeider} proved, it is in NP-hard to determine whether a given forbidden-transition graph~$(G, T)$ contains a compatible $s$-$t$ path for two given vertices~$s$ and~$t$.
This of course implies that it is NP-hard to check whether there is a compatible $s$-$t$ path of at most some given length.
In contrast, it is polynomial-time solvable to decide whether there is a compatible $s$-$t$ path which has length at most $\dist_G(s, t)$.
This can for example be seen by using the following strategy.
Construct the line graph~$H$ of $G$.
(That is, $H$ has vertex set $E(G)$ and two vertices in $H$ are adjacent if the corresponding edges in $G$ share an endpoint.)
For each vertex $v \in V(G)$ and each pair~$e, f$ of edges incident to~$v$ such that $e$ and $f$ are not compatible, remove the edge $ef$ from~$H$.
Introduce two new vertices $s'$ and $t'$ into $H$ and make them adjacent to every vertex corresponding to an edge incident in~$G$ with~$s$ or~$t$, respectively.
Finally, check whether $H$ contains an (ordinary) $s'$-$t'$ path of length at most $\dist_G(s, t) + 1$.

In this section we improve on the above observation by showing that checking for compatible $s$-$t$ paths which are marginally longer than $\dist_G(s, t)$ can also be done efficiently.
That is we are going to show the fixed-parameter tractability of the following problem.
\defparproblem{\comdeto}
{An instance $(G,T,s,t,k)$ where $(G, T)$ is a forbidden-transition graph, $s,t \in V(G)$, and \(k \in \mathbb{N}\).}
{$k \in \mathbb{N}$}
{Does there exist a \(T\)-compatible $s$-$t$ path in $G$ of length at most $\dist_G(s,t)+k?$}

For notational convenience, we slightly generalize the notion of an $x$-$y$ path as follows.
For a given graph $G$ and $x,y \in V(G) \cup E(G)$, we say that a path $(v_1, v_2, \ldots, v_\ell)$ in \(G\) is an $x$-$y$ path, if (i) $x=v_1\in V(G)$ or $x=v_1v_2 \in E(G)$ and (ii) $y=v_\ell \in V(G)$ or $y=v_{\ell-1}v_\ell \in E(G)$.

We first show fixed-parameter tractability of the \compath problem, which will be later used as a black box in our algorithm for \comdeto. The algorithm for \compath uses a standard color-coding approach (see \cite{cygan2015parameterized}), slightly modified to track the transitions.
\defpartask{\compath}
{An instance $(G,T,x,y,k)$ where $(G, T)$ is a forbidden-transition graph, $x,y \in V(G) \cup E(G)$, and \(k \in \mathbb{N}\).}
{$k \in \mathbb{N}$}
{Decide whether there exists a $T$-compatible $x$-$y$ path in $G$ of length at most $k$. If so, return the length of a shortest such path.}

The algorithm can be presented using random colorings but this would complicate the analysis later.
We instead use the notion of perfect hash families.
Let $n, k \in \mathbb{N}$ and let $U$ be a set of size~$n$.
A \emph{$k$-perfect hash family} of $U$ is a family $\mathcal{F}$ of functions from $U$ to $[k]$ such that for each subset $S \subseteq U$ of size at most~$k$ there exists a function $f \in \mathcal{F}$ such that $f(S) = [k]$ (that is, $f$ is injective on $S$).
Naor, Schulman, and Srinivasan~\cite{naor_splitters_1995} showed that a $k$-perfect hash family of size $e^k k^{O(\log k)} \log n$ can be computed in $e^k k^{O(\log k)} n \log n$ time; see also \cite[Section 5.6.1]{cygan2015parameterized}.

\begin{theorem}\label{thm:comp-path-fpt-length}
There exists an algorithm solving \compath in time $2^{\Oh(k\log k)} n^{\Oh(1)}$, where $n = |V(G)|$.
\end{theorem} 
\begin{proof}
  Let $(G, T, x, y, k)$ be the input instance.
  First, consider the case where $x, y \in V(G)$.
  The algorithm works as follows.
  We start by coloring the vertices in $G$ using a family of perfect hash functions.
  Compute a $(k - 1)$-perfect hash family $\mathcal{F}$ of $V(G)$ in $e^k k^{O(\log k)} n \log n$ time.
  Iterate over all elements $f \in \mathcal{F}$ and for each such element $f$ proceed as follows.
  Generate a coloring of $V(G)$ by starting with coloring $x$ and $y$ with two unique colors, say $1$ and $k+1$ and then coloring the rest of vertices of $G$ with colors $2,\ldots,k$, according to the values assigned by $f$.
  That is, put the color of a vertex $v \in V(G) \setminus \{x, y\}$ to be $f(v) + 1$.
  Let $c:V(G) \to [k+1]$ be the resulting coloring and for each $i \in [k+1]$ denote by $C_i$ the set $c^{-1}(i)$.
  Observe that it suffices to prove the following claim.
\begin{claim}\label{lem:color-coding}
  Let $x, y \in V(G)$.
  There exists an algorithm to test whether there exists a colorful $T$\nobreakdash-compatible $x$-$y$ path with at most $k+1$ vertices in time $2^k n^{\Oh(1)}$.
\end{claim} 
\begin{proof}
We use the following dynamic programming approach that computes a table~\(D\): for a set $S$ of at least two elements, such that $\{1\} \subseteq S \subseteq [k+1]$, and an edge $uv \in E(G)$ we define a boolean variable $D[S,u,v]$. 
We want it to be equal to TRUE if and only if there exists a colorful compatible $x$-$uv$ path, whose vertices are colored with all colors from $S$. 

For $S=\{1,i\}$ observe that $D[S,u,v]=$ TRUE if and only if $u=x$ and $v \in N(x) \cap C_i$.
Thus, these entries of $D$ can be computed in linear time.
Next, for every~\(S\) with $\{1\} \subseteq S \subseteq [k+1]$ and $|S|\geq 3$ and for every $e=uv \in E(G)$ we compute $D[S,u,v]$  as follows:
\begin{equation*}
D[S,u,v]=\begin{cases}
\bigvee \{D[S \setminus \{c(v)\},w,u]: \{wu,e\} \in T(u), wu \in E(G)\} & \textnormal{ if $c(v) \in S\setminus \{1\}$,} \\
\textnormal{FALSE} & \textnormal{ otherwise.}
\end{cases}
\end{equation*}
Observe that this is a correct way of computing the entries \(D[S, u, v]\).
As to the running time,
the number of possible sets $S$ is bounded by $2^k$ (as 1 is always included in $S$).
Thus, $D[S, u, v]$ can be computed in $2^{k}n^{\Oh(1)}$ time.
We say that $S \subseteq [k+1]$ is \emph{good} if it contains 1 and $k+1$. 
Observe that a compatible $x$-$y$ path of length at most $k$ exists if and only if there exist $uv \in E(G)$ and a good set $S$, such that $D[S,u,v]=$ TRUE (note that in such a case $v=y$). 
Let $\mathcal{S}$ be the set of all triples $(S, u,v)$ such that $S$ is good and $D[S,u,v]=$ TRUE.
If $\mathcal{S}$ it is empty, then clearly there is no $T$-compatible $x$-$y$ path of length at most $k$ in $G$.
Otherwise, we return the value of $|S|-1$ for the smallest $S$ such that $(S,u,v) \in \mathcal{S}$ for some $uv \in E(G)$.
\end{proof}

Analogously, if we are asked for a colorful $T$-compatible $x$-$y$ path when $\{x,y\} \cap E(G) \neq \emptyset$, we give unique colors to the second and/or last but one vertex of a potential path and slightly modify computation of the table $D$ and we look for an optimal solution in $\mathcal{S}$ that respects corresponding the conditions.
\end{proof}

We are ready to prove the fixed-parameter tractability of \comdeto.
Our algorithm for \comdeto is based on an algorithm for computing $s$-$t$ paths of length exactly $\dist(s, t) + k$ in ordinary graphs~\cite{BezakovaCDF2019}.
The basic idea is that in such a path there are at most $k$ segments in which no ``progress'' is made towards reaching~$t$.
These $k$ segments can be determined locally by applying the algorithm for \compath\ from above.
The remaining parts can be computed by ordinary dynamic programming as for computing shortest paths. 
\begin{theorem}
There exists an algorithm solving \comdeto in time $2^{\Oh(k\log k)} n^{\Oh(1)}$.
\end{theorem}
\begin{proof}
  Let $(G,T,s,t,k)$ be the input instance of \comdeto\ and let $d=\dist_G(s,t)$.
  We can clearly assume that $d > k$, as otherwise we can compute whether $(G,T,s,t,k)$ is a yes-instance using the algorithm for \compath with parameter $d+k \leq 2k$.
  For each $i \in [V(G)] \cup \{0\}$
  we define the \emph{$i$-th layer} $X_i:=\{x \in V(G): \dist_G(s,x)=i\}.$ Clearly $\{s\} = X_0$ and $t \in X_d$.
  Note that each compatible $s$-$t$-path of length at most $d+k$ must be contained in the graph induced by $X = \bigcup_{i \in \{0,1,\ldots,d+k\}} X_i$, therefore we can safely assume that $X=V(G)$. 
  For two distinct layers $X_i,X_j$, we say that $X_i$ is \emph{higher} (resp. lower) than $X_j$ if $i>j$ (resp $i<j$).
  We use the following notation:
  An edge \(xy \in E(G)\) is called an \emph{inter-layer} if there exist \(i, j \in [d + k] \cup \{0\}\) such that \(i \neq j\), \(x \in X_i\), and \(y \in X_j\).
  If an edge is not inter-layer, we call it \emph{within-layer}.
  For two vertices $p, q \in V(G)$ such that $\dist(s,p)< \dist(s,q)$ we denote by $G_{(p,q]}$ the subgraph of $G$ induced by $\{p\} \cup \{x \in V(G): \dist(s,p)<\dist(s,x)\leq\dist(s,q)\}.$ 

  The algorithm uses a dynamic-programming procedure that computes a table~$D$.
  Table $D$ is indexed by the inter-layer edges of $G$ and for every inter-layer edge $xy \in E(G)$ the entry $D[xy]$ contains the length $\ell$ of a shortest compatible $xy$-$t$ path in \(G_{(x, d + k]}\) if it exists and if $\dist(s,x)+\ell\leq d+k$; otherwise $D[xy] = \infty$.
  Note that it suffices to compute the entries of the table~\(D\) because we can then look up whether there is a solution in the entries corresponding to the edges incident with~\(s\) (which are all inter-layer edges).
  We first compute the entries for edges with the highest layers and then for edges in successively lower layers.
  Intuitively, when filling~\(D\) for a specific inter-layer edge \(uv\), we can rely on the fact that each solution path using \(uv\) will contain an inter-layer edge \(wx\) in a higher layer and it will reach \(wx\) from \(uv\) in at most \(2k\) steps.
  Thus, when filling the table for \(uv\), we may refer to the correct entry for \(wx\) and compute the path between \(uv\) and \(wx\) using a call to the algorithm for \compath.

  At the beginning of the procedure we initialize the table, putting $D[xy]=\infty$ for every inter-layer edge $xy \in E(G)$.
  Next, we compute entries of $D$ for the last $2k + 1$ layers:
  for every inter-layer edge $xy \in E(G)$ such that $\dist_G(s,x) \geq d-k - 1$ and $\dist_G(s,y) >d-k -1$, solve the \compath instance $(G_{(x, d + k]},T,xy,t,2k + 1)$. 
  If for some $xy$ the result is a path of length~$\ell$ such that $\dist_G(s,x)+\ell\leq d+k$, then we set $D[xy]=\ell$, otherwise we put $D[xy] = \infty$.
  Observe that this will fill \(D[xy]\) with the correct value according to the definition of~\(D\).
  
  Then, we inductively fill in earlier layers by carrying out the following computation steps:
  \begin{enumerate}[(1)]
  \item For every integer~$m$ from $d-k-1$ down to $0$ and for every
    pair of vertices $x,u$ such that $\dist_G(s,x)=m$ and
    $m<\dist_G(s,u)\leq m + k + 1$ we do the following:
    \begin{enumerate}[(a)]
    \item For every pair of edges $e, f \in E(G_{(x,u]})$ such that $e = xy$ (so \(e\) is an inter-layer edge) and $f = vu$ for some vertices $y, v \in V(G_{(x,u]})$, we do the following:
      \begin{enumerate}[(i)]
      \item We solve the \compath instance $(G_{(x,u]},T,e,f,2k)$.
      \item If the answer is negative, we continue with the next pair of candidates for~$e$ and~$f$.
      \item If the answer is positive, let $r$ be the returned path length.
        Observe that $r \in [2k]$.
      \item Let \(p = \min \{D[g] \colon g \text{ is an inter-layer edge and } fg \in T(u)\}\).
        If $\dist(s,x) + r + p \leq d + k$ and $D[e] > r + p$, then we put $D[e] = r + p$.
      \end{enumerate}
    \end{enumerate}
  \end{enumerate}
  For later reference, let $\nu = \min_{v \in N_G(s)} D[sv]$.
  We accept if and only if $\nu \leq \dist_G(s,t) + k$.
  Computing a single entry $D[xy]$ takes time $2^{\Oh(k\log^2k)}n^{\Oh(1)}$, and since there are less than $n^2$ of them, $2^{\Oh(k\log^2k)}n^{\Oh(1)}$ is the complexity of the whole algorithm.

  We now show the correctness of the algorithm.
  Observe first, that each table entry~$D[e]$ only receives a non-infinity value if there is a compatible $e$-$t$ path of length at most~$d + k$; this is ensured by the way the entries are filled in step~(iv).
  Moreover, each non-infinity entry $D[e]$ contains the length of some compatible $e$-$t$ path.
  Thus, the algorithm accepts only if there is a compatible $s$-$t$ path of length at most~$d + k$.
  In particular, if there exists no such path, then the answer is correct.
  Moreover, if there exists a compatible $s$-$t$ path in $G$ of length $\ell^\star \leq d+k$, then $\nu \geq \ell^\star$.

  Now assume that there exists a compatible $s$-$t$ path in $G$ of length~\(\ell^\star\) such that~$\ell^\star \leq d+k$.
  It remains to show that $\nu \leq \ell^\star$.
  We prove the stronger statement that for each inter-layer edge \(xy\) we have that \(D[xy]\) is at most the length, \(\ell\), of a shortest compatible \(xy\)-\(t\) path in \(G_{(x, d + k]}\) that satisfies \(\dist_G(s, x) + \ell \leq d + k\) or \(\infty\) if no such path exists.
  The proof is by induction on \(d - m\) where \(m\) is the layer of \(x\).
  By the above, the statement holds for \(m \geq d - k\).
  Now assume that \(m < d - k\).
  If there is no suitable compatible \(xy\)-\(t\) path, the statement clearly holds.
  Otherwise, let \(P\) be such a path.
  We claim that it suffices to show that on \(P\) there exist consecutive vertices \(v, u, z\) such that the following properties hold.
  \begin{enumerate}[(P1)]
  \item \(m < \dist_G(s, u) \leq m + k + 1\).\label[condition]{en:det-few-layers}
  \item \(\dist_G(s, u) < \dist_G(s, z)\).\label[condition]{en:det-inter-layer}
  \item Let \(P[u, t]\) be the subpath of \(P\) from \(u\) to~\(t\).
    Then \(P[u, t]\) is contained in \(G_{(u, d + k]}\).\label[condition]{en:det-subpath}
  \item There are at most \(2k\) edges on \(P\) between (incl.) \(xy\) and \(vu\).\label[condition]{en:det-short}
  \end{enumerate}
  Let \(e = xy\), \(f = vu\), and \(g = uz\).
  If the above claim is true, then in Step~(1) above we will guess \(x\) and \(u\) by \ref{en:det-few-layers}; in Step~(a) we will guess \(e\) and \(f\) by definition of \(e\); we will find an \(e\)-\(f\) path at most as long as the corresponding subpath of \(P\) in Step~(i) by \ref{en:det-short}; and we will consider \(D[g]\) in the minimum in Step~(iv) by \ref{en:det-inter-layer} and since \(P\) is compatible.
  Furthermore, \(D[g]\) is at most the length, \(\ell_u\), of \(P[u, t]\):
  By \ref{en:det-subpath}, \(P[u, t]\) is a path in \(G_{(u, d + k]}\).
  To see that it satisfies the condition on its length
  let \(\ell' = \ell - \ell_u\).
  Observe that \(\dist_G(s, u) - \dist_G(s, x) \leq \ell'\) and thus we have \(\dist_G(s, u) + \ell_u \leq \dist_G(s, x) + \ell' + \ell_u = \dist_G(s, x) + \ell \leq d + k\).
  Hence indeed, \(P[u, t]\) satisfies \(\dist_G(s, u) + \ell_u \leq d + k\), certifying that \(D[g]\) is at most the length of \(P[u, t]\).
  Thus \(D[e]\) will receive a value that is at most the length of~\(P\) in Step~(iv), as required.

  Before proving the claim, let us observe the following.
  Say that an inter-layer edge \(h\) is a \emph{back} edge if \(P\) traverses the vertex in \(h\) that is in a larger layer before the other vertex in~\(h\).
  Observe that the length of \(P\) is \(d + a + 2b\) where \(a\)~is the number of within-layer edges in~\(P\) and \(b\)~the number of back edges in~\(P\).
  Thus, we have \(a + b \leq k\).
  In particular, there are at most \(k\) layers in which \(P\) contains at least two vertices.
  We will use this fact below.
  
  It remains to prove our claim above.
  Since there are at most \(k\) layers in which \(P\) contains at least two vertices, in the layers \(m + 1, m + 2, \ldots, m + k + 1\) there is least one vertex, \(u\), on \(P\) such that \(u\) is the only vertex of \(P\) in \(u\)'s layer.
  We claim that \(u\) together with the vertex, \(v\), that precedes \(u\) on \(P\) and the vertex, \(z\), that succeeds \(u\) on \(P\) satisfy the properties in the claim.
  Clearly, \ref{en:det-few-layers} is satisfied.
  Since \(u\) is the only vertex of \(P\) on \(u\)'s layer, also \ref{en:det-inter-layer} is satisfied.
  For the same reason, \ref{en:det-subpath} is satisfied.
  Finally, suppose that \ref{en:det-short} does not hold, that is, there are more than \(2k\) edges between \(xy\) and \(vu\) on \(P\).
  Then the length, \(\ell\), of \(P\) is at least \(2k + 1 + \dist_G(u, t)\).
  However, then \(\dist_G(s, x) + \ell \geq \dist_G(s, x) + 2k + 1 + \dist_G(u, t) > d + k\), a contradiction to the fact that \(\dist_G(s, x) + \ell \leq d + k\).
  Thus, the claim holds, meaning that \(\nu \leq \ell^\star\) and the algorithm is correct.
\end{proof}

%% file: lower-bound.tex
Let $G$ be an undirected graph.
A \emph{modulator to a linear forest} of $G$ is a vertex subset $S \subseteq V(G)$ such that $G - S$ is a disjoint union of paths.
The \emph{distance $k$ of $G$ to a linear forest} is the minimum size, $k$, of a modulator to a linear forest.
Note that the distance to a linear forest upper bounds the size of a minimum feedback-vertex set and the treewidth and hence W[1]-hardness for these two parameters is implied by W[1]-hardness for~$k$.
A \emph{modulator to treewidth two} of $G$ is a vertex subset $S \subseteq V(G)$ such that $G - S$ has treewidth at most two.
The \emph{distance of $G$ to treewidth two} is the minimum size of a modulator to treewidth two.
Analogously, the distance to treewidth two upper bounds the treewidth and hence W[1]-hardness for treewidth is implied by W[1]-hardness for the distance to treewidth two.

\looseness=-1
In this section, we first show that finding long paths or cycles is W[1]-hard with respect to the distance~$k$ to a linear forest.
Moreover, assuming the Exponential Time Hypothesis (ETH), no $f(k)\cdot n^{\smalloh(k/\log k)}$-time algorithm can exist.
Informally, the ETH states that \textsc{3-SAT} on $n$-variable formulas cannot be solved in $2^{o(n)}$~time, see~\cite{impagliazzo_which_1998,impagliazzo_complexity_1999}.
We obtain the following.
\begin{theorem}\label[theorem]{thm:lb1}
  Let $(G, T)$ be forbidden-transition graph and $s$, $t$ two vertices in $G$.
  Let $\ell$ be a positive integer and let $k$ be the distance of $G$ to a linear forest.
  For each of the following, it is W[1]-hard with respect to $k$ to decide and, moreover, an $f(k)\cdot n^{\smalloh(k/\log k)}$-time decision algorithm contradicts the ETH:
  \begin{enumerate}[(i)]
  \item whether $G$ contains a compatible $s$-$t$ path,\label{lb:path}
  \item whether $G$ contains a compatible $s$-$t$ path of length at least $\ell$ (or at most \(\ell\)),\label{lb:long-path}
  \item whether $G$ contains a compatible cycle, and\label{lb:cycle}
  \item whether $G$ contains a compatible cycle of length at least $\ell$ (or at most \(\ell\)).\label{lb:long-cycle}
  \end{enumerate}
  \end{theorem}

\newcommand{\col}{\ensuremath{\textsf{col}}}
\newcommand{\subisolong}{\textsc{Partitioned Subgraph Isomorphism}}
\newcommand{\subiso}{\textsc{PSI}}
\newcommand{\multiclique}{\textsc{Multicolored Clique}}

\begin{proof}
We first give a reduction to prove hardness of~\cref{lb:path}.
Observe that~\cref{lb:long-path} follows from~\cref{lb:path}.
We then modify the construction to obtain \cref{lb:cycle} and \cref{lb:long-cycle}.

Our reduction is from the \subisolong~(\subiso) problem.
Herein, we are given two graphs $G$ and $H$, where $V(H) = [n_H]$ for some positive integer $n_H$, and a vertex coloring $\col \colon V(G) \to V(H)$ of the vertices of $G$ with colors that one-to-one correspond to the vertices of~$H$.
Moreover, each vertex of $H$ is incident with at least one edge and for each edge $\{u, v\} \in E(G)$ we have $\col(u) \neq \col(v)$.
We want to decide whether $H$ is isomorphic to a subgraph of~$G$ while respecting the colors, that is, whether there is an injective mapping $\phi \colon V(H) \to V(G)$ such that for all $u \in V(H)$ we have $\col(\phi(u)) = u$ and for all $\{u, v\} \in E(H)$ we have $\{\phi(u), \phi(v)\} \in E(G)$.
In that case, we also say that $\phi$ is a \emph{subgraph isomorphism} from $H$ into~$G$.
In the following we let $m_H = |E(H)|$.
Observe that $n_H \leq 2m_H$ since each vertex of $H$ is incident with at least one edge.
Since \subiso\ contains \multiclique~\cite{fellows_parameterized_2009} as a special case, \subiso\ is W[1]\nobreakdash-hard with respect to~$m_H$.
Moreover, Marx~\cite[Corollary 6.3]{marx_can_2010} observed that an $f(m_H)\cdot n^{\smalloh(m_H/\log m_H)}$-time algorithm for \subiso\ would contradict the ETH.

Our construction works as follows: we first build a path from $s$ to a vertex $t_1$.
This path is the concatenation of $n_H$ subpaths $P^1,\dots,P^{n_H}$ where each subpath is associated with a vertex of $H$.
The subpath $P^i$ contains a vertex for each edge of $G$ incident to a vertex colored $i$.
We then use an extra vertex and an appropriate transition system so that one can choose any vertex $v$ of $G$ with color $i$ and connect the endpoints of $P^i$ with a compatible path that skips the vertices of $P^i$ that denote an edge adjacent to $v$.
This comes down to choosing $\phi(i)=v$.
Finally, we connect $t_1$ to $t$ by a sequence of gadgets each associated with an edge of~$H$.
Choosing a path through a gadget comes down to mapping an edge $uv$ of $H$ to an edge $wx$ of $G$.
Our transition system then requires the path in the gadget to visit the two vertices of $P$ that denote the edge $wx$, which can only be done without repeating vertices if those vertices have been skipped between $s$ and $t_1$.
This means that the endpoints of $wx$ have to be the vertices we chose as $\phi(u)$ and $\phi(v)$.
By ensuring that there is an edge between $\phi(u)$ and $\phi(v)$, we prove that $\phi$ is a subgraph isomorphism.

\tikzset{end/.style={
        fill,circle,scale=0.5
    }
}

\tikzset{gn/.style={
        fill,circle,scale=0.5
    }
}

\tikzset{middle/.style={
        fill,circle,scale=0.4
    }
}

\tikzset{invisible/.style={
        fill,circle,scale=0.05
    }
}

\tikzset{simple/.style={
        black, thick
    }
}
\tikzset{dot/.style={
        black, thick, dotted
    }
}

\tikzset{bl/.style={
        blue
    }
}

\tikzset{gr/.style={
        green
    }
}

\tikzset{re/.style={
        red
    }
}

\tikzset{or/.style={
        olive
    }
}
\tikzset{bracket/.style={
        black, thin
    }
}

\tikzset{short/.style={
        gray, thick
    }
}

\tikzset{none/.style={
    }
}

\begin{figure}[tbp]
  \centering
  \begin{tikzpicture}[scale=0.6]
    
		\node [style=end,label={[label distance=0.1mm]left:{$s$}}] (0) at (-10, -0) {};
		\node [style=end] (1) at (-6, -0) {};
		\node [style=end] (2) at (-4, -0) {};
		\node [style=end] (3) at (0, -0) {};
		\node [style=end] (4) at (2, -0) {};
		\node [style=end] (5) at (6, -0) {};
		\node [style=end] (6) at (8, -0) {};
		\node [style=end,label={[label distance=0.1mm]right:{$t_1$}}] (7) at (12, -0) {};
		\node [style=gn, label={[label distance=0.1mm]above:{$z_3^p$}}] (8) at (-1.75, 3) {};
		\node [style=gn, label={[label distance=0.1mm]above:{$z_2^p$}}] (9) at (1.000001, 3) {};
		\node [style=gn, label={[label distance=0.1mm]above:{$z_1^p$}}] (10) at (3.5, 3) {};
		\node [fill=blue, circle,scale=0.4] (11) at (-2.5, -0) {};
        \node [style=none, label={[label distance=0.1mm]below:{\tiny $x^{i}_{b(e)}$}}] (11a) at (-2.5, 0.3) {};
		\node [fill=green,, circle,scale=0.4] (12) at (-1, -0) {};
        \node [style=none, label={[label distance=0.1mm]below:{\tiny $x^{i}_{b(f)}$}}] (12a) at (-0.8, 0.3) {};
		\node [style=middle,fill=green] (13) at (3.25, -0) {};
        \node [style=none, label={[label distance=0.1mm]below:{\tiny $x^{j}_{a(f)}$}}] (11a) at (3.25, 0.3) {};
		\node [style=middle,fill=blue] (14) at (4.75, -0) {};
        \node [style=none, label={[label distance=0.1mm]below:{\tiny $x^{j}_{a(e)}$}}] (11a) at (4.75, 0.3) {};
		\node [style=end,label={[label distance=0.1mm]left:{$t$}}] (15) at (-11.25, 3) {};
		\node [style=invisible] (16) at (-6, -5) {};
		\node [style=invisible] (17) at (-6, -5.25) {};
		\node [style=invisible] (18) at (-10, -5) {};
		\node [style=invisible] (19) at (-10, -5.25) {};
		\node [style=invisible] (20) at (-4, -5) {};
		\node [style=invisible] (21) at (0, -5) {};
		\node [style=invisible] (22) at (-4, -5.25) {};
		\node [style=invisible] (23) at (0, -5.25) {};
		\node [style=invisible] (24) at (2, -5) {};
		\node [style=invisible] (25) at (6, -5) {};
		\node [style=invisible] (26) at (2, -5.25) {};
		\node [style=invisible] (27) at (6, -5.25) {};
		\node [style=invisible] (28) at (8, -5) {};
		\node [style=invisible] (29) at (12, -5) {};
		\node [style=invisible] (30) at (8, -5.25) {};
		\node [style=invisible] (31) at (12, -5.25) {};
		\node [style=gn, label={[label distance=0.1mm]above:{$z_3^1$}}] (32) at (7.25, 3) {};
		\node [style=gn, label={[label distance=0.1mm]above:{$z_2^1$}}] (33) at (9.5, 3) {};
		\node [style=gn, label={[label distance=0.1mm]above:{$z_1^1$}}] (34) at (11.75, 3) {};
		\node [style=gn, label={[label distance=0.1mm]above:{$z_2^{m_H}$}}] (35) at (-7, 3) {};
		\node [style=gn, label={[label distance=0.1mm]above:{$z_1^{m_H}$}}] (36) at (-5, 3) {};
		\node [style=gn, label={[label distance=0.1mm]above:{$z_3^{m_H}$}}] (37) at (-9, 3) {};
		\node [style=gn, label=below:{$y^i$}] (38) at (-1.75, -3.5) {};
		\node [style=middle,fill=olive] (39) at (-2, -0) {};
        \node [style=none, label={[label distance=0.1mm]above:{\tiny $\text{post}(u)$}}] (39a) at (-2, 0.4) {};
        \node [style=none] (39b) at (-2, 1) {};
        \node [style=none] (39c) at (-2, 0.3) {};
        \draw [->,very thin,gray] (39b) to (39c);
		\node [style=middle,fill=red] (40) at (-1.5, -0) {};
        \node [style=none, label={[label distance=0.1mm]above:{\tiny $\text{pre}(v)$}}] (40a) at (-1.3, -0.1) {};
        \node [style=none] (40b) at (-1.5, 0.6) {};
        \node [style=none] (40c) at (-1.5, 0) {};
        \draw [->,very thin,gray] (40b) to (40c);
		\node [style=middle,fill=olive] (41) at (-3.25, -0) {};
        \node [style=none, label={[label distance=0.1mm]above:{\tiny $\text{pre}(u)$}}] (41a) at (-3.25, -0.3) {};
		\node [style=middle,fill=red] (42) at (-0.5000001, -0) {};
        \node [style=none, label={[label distance=0.1mm]above:{\tiny $\text{post}(v)$}}] (42a) at (-0.1, -0.25) {};
		\node [style=invisible] (43) at (4.5, 3) {};
		\node [style=invisible] (44) at (6.25, 3) {};
		\node [style=invisible] (45) at (-2.5, 3) {};
		\node [style=invisible] (46) at (-4.25, 3) {};
		\node [style=invisible] (47) at (11, 2.75) {};
		\node [style=invisible] (48) at (11.25, 2.5) {};
		\node [style=invisible] (49) at (8.75, 2.75) {};
		\node [style=invisible] (50) at (9, 2.5) {};
		\node [style=invisible] (51) at (9.5, 2.25) {};
		\node [style=invisible] (52) at (10, 2.5) {};
		\node [style=invisible] (53) at (6.75, 2.5) {};
		\node [style=invisible] (54) at (7.25, 2.25) {};
		\node [style=invisible] (55) at (-4.75, 2.5) {};
		\node [style=invisible] (56) at (-4.5, 2.75) {};
		\node [style=invisible] (57) at (-7, 2.5) {};
		\node [style=invisible] (58) at (-6.75, 2.5) {};
		\node [style=invisible] (59) at (-6.5, 2.5) {};
		\node [style=invisible] (60) at (-6.25, 2.5) {};
		\node [style=invisible] (61) at (-8.75, 2.5) {};
		\node [style=invisible] (62) at (-8.5, 2.75) {};
		\node [style=invisible] (63) at (-8, -5.5) {};
		\node [style=invisible, label=below:{$P^1$}] (64) at (-8, -5.75) {};
		\node [style=invisible] (65) at (-2, -5.5) {};
		\node [style=invisible, label=below:{$P^i$}] (66) at (-2, -5.75) {};
		\node [style=invisible] (67) at (4, -5.5) {};
		\node [style=invisible, label=below:{$P^j$}] (68) at (4, -5.75) {};
		\node [style=invisible] (69) at (10, -5.5) {};
		\node [style=invisible, label=below:{$P^{n_H}$}] (70) at (10, -5.75) {};

		\draw [style=simple] (0) to (1);
		\draw [style=simple] (2) to (3);
		\draw [style=simple] (4) to (5);
		\draw [style=dot] (1) to (2);
		\draw [style=dot] (3) to (4);
		\draw [style=dot] (5) to (6);
		\draw [style=simple] (6) to (7);
		\draw [style=bl] (8) to (11);
		\draw [style=bl] (11) to (9);
		\draw [style=bl] (9) to (14);
		\draw [style=bl] (14) to (10);
		\draw [style=gr] (8) to (12);
		\draw [style=gr] (12) to (9);
		\draw [style=gr] (9) to (13);
		\draw [style=bracket] (19) to (18);
		\draw [style=bracket] (19) to (17);
		\draw [style=bracket] (17) to (16);
		\draw [style=bracket] (20) to (22);
		\draw [style=bracket] (22) to (23);
		\draw [style=bracket] (23) to (21);
		\draw [style=bracket] (24) to (26);
		\draw [style=bracket] (26) to (27);
		\draw [style=bracket] (27) to (25);
		\draw [style=bracket] (28) to (30);
		\draw [style=bracket] (30) to (31);
		\draw [style=bracket] (31) to (29);
		\draw [style=gr] (10) to (13);
		\draw [style=re] (40) to (38);
		\draw [style=re] (38) to (42);
		\draw [style=or] (41) to (38);
		\draw [style=or] (38) to (39);
		\draw [style=simple] (34) to [bend left=60] (7);
		\draw [style=simple] (36) to (46);
		\draw [style=simple] (45) to (8);
		\draw [style=simple] (10) to (43);
		\draw [style=simple] (44) to (32);
		\draw [style=simple] (15) to (37);
		\draw [style=dot] (46) to (45);
		\draw [style=dot] (43) to (44);

		\draw [style=short] (34) to (47);
		\draw [style=short] (34) to (48);
		\draw [style=short] (33) to (49);
		\draw [style=short] (33) to (50);
		\draw [style=short] (33) to (51);
		\draw [style=short] (33) to (52);
		\draw [style=short] (32) to (53);
		\draw [style=short] (32) to (54);
		\draw [style=short] (32) to (54);
		\draw [style=short] (37) to (61);
		\draw [style=short] (37) to (62);
		\draw [style=short] (35) to (57);
		\draw [style=short] (35) to (58);
		\draw [style=short] (35) to (59);
		\draw [style=short] (35) to (60);
		\draw [style=short] (36) to (55);
		\draw [style=short] (36) to (56);
		\draw [->] (63) to (64);
		\draw [->] (65) to (66);
		\draw [->] (67) to (68);
		\draw [->] (69) to (70);
\end{tikzpicture}
\caption{An illustration of our construction. We use colors to denote edges that have to be used consecutively because of the set of permitted transitions. For example, the two dark yellow edges correspond to a vertex $u$ of $G$ in the vertex-selection gadget for vertex~$i$ of~$H$ and the four blue edges correspond to an edge $e$ of~$G$ in the edge-selection gadget for the $p$th edge of~$H$.} \label{fig:lowerbound}
\end{figure}
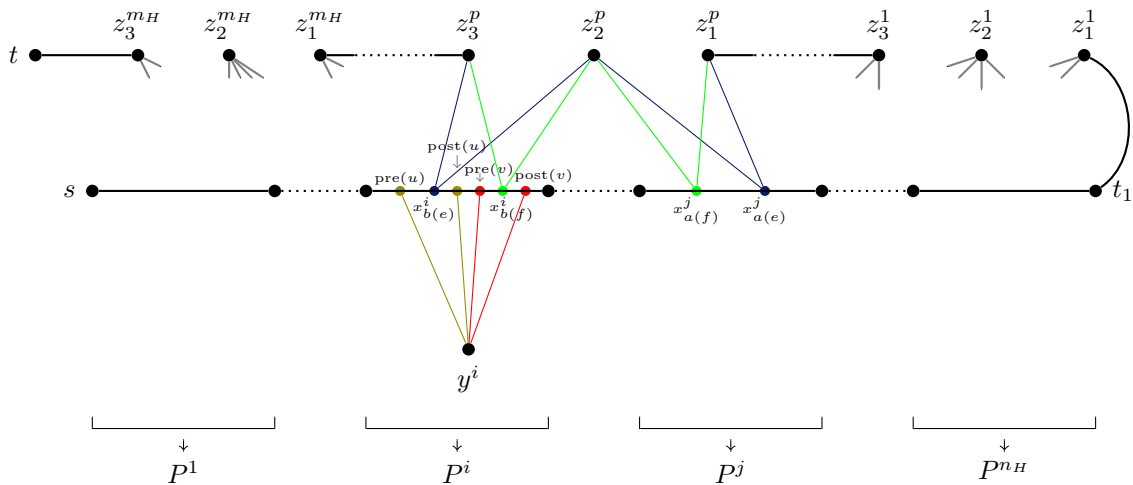

\newcommand{\consg}{\ensuremath{G^\star}}
\newcommand{\segpre}{\ensuremath{\textsf{pre}}}
\newcommand{\segpost}{\ensuremath{\textsf{post}}}
\begin{construction}\label{cons:lb}
  Let $(G, H, \col)$ be an instance of \subiso, where $V(H) = [n_H]$.
  For each $i \in [n_H]$ define $V_i = \{v \in V(G) \mid \col(v) = i\}$.
  For each $i \in [n_H]$ define $E_{i} = \{e \in E(G) \mid \exists u \in V_i \colon u \in e\}$.
  We construct a forbidden-transition graph $(\consg, T)$ as follows, see Figure~\ref{fig:lowerbound} for an illustration.
  We begin with \consg\ being empty.
  We will specify $T$ by giving the permitted-transition sets~$T(v)$ for the individual vertices $v \in V(\consg)$.
  Below, we specify $T(v)$ only for a subset of $V(\consg)$.
  For all the remaining vertices~$v$, we put $T(v) = \binom{E(v)}{2}$ (recall that $E(v)$ is the set of edges in $\consg$ that are incident with~$v$).
  Introduce new vertices $s, t, t_1$ into \consg.
  We construct the vertex-selection gadgets as follows.

  Introduce a path $P$ from $s$ to $t_1$ into \consg; we specify the number of vertices on $P$ indirectly below.
  For each internal vertex $v \in V(P)$ put $T(v) = \{\{\{u, v\}, \{v, w\}\}\}$ where $u$ and $w$ are the neighbors of $v$ in $P$.
  Additional edges and transitions for the vertices on $P$ will be introduced below.
  Partition $P$ into $n_H$ disjoint paths $P^1, \ldots, P^{n_H}$; we specify the number of vertices in each of these paths in the next step.

  \looseness=-1
  For each $i \in [n_H]$, proceed as follows.
  Let $(e^i_a)_{a \in [r_i]}$ be an ordering of $E_i$ such that, for each $v \in V_i$, the edges in $E(v)$ form a segment in $(e^i_a)$ (observe that such an ordering exists since the endpoints of each edge in $E(G)$ have two different colors).
  Set the number of vertices in $P^i$ to $r_i + 4$.
  For each $a \in [r_i]$ denote the $a + 2$-th vertex on $P^i$ by $x^i_a$.
  We say that vertex~$x^i_a$ \emph{corresponds} to the edge~$e^i_a$ of~$G$.
  (We keep the first two and last two vertices of $P^i$ unnamed.)

  Next, introduce a vertex $y^i$ and for each vertex $v \in V_i$ proceed as follows.
  Let $\segpre(v)$ be the vertex in~$P^i$ that directly precedes on $P^i$ the first vertex corresponding to an edge in $E_G(v)$.
  Similarly, let $\segpost(v)$ be the vertex in $P^i$ that directly succeeds on $P^i$ the last vertex corresponding to an edge in~$E_G(v)$.
  For later, it is useful to observe that all vertices on $P^i$ strictly between $\segpre(v)$ and $\segpost(v)$ correspond to edges in~$E_G(v)$.
  Now add the edges $\{\segpre(v), y^i\}$ and $\{y^i, \segpost(v)\}$ to \consg\ and the transition $\{\{\segpre(v), y^i\}, \{y^i, \segpost(v)\}\}$ to~$T(y^i)$.
  Moreover, add the transition $\{\{u, \segpre(v)\}, \{\segpre(v), y^i\}\}$ to $T(\segpre(v))$ where $u$ is the vertex on $P^i$ preceding $\segpre(v)$ (if any), and add the transition $\{\{y^i, \segpost(v)\}, \{\segpost(v), w\}\}$ to $T(\segpost(v))$, where $w$ is the vertex on $P^i$ succeeding $\segpost(v)$.
  This finishes the construction of the vertex-selection gadgets, but further edges and transitions may be introduced later to the vertices of~$P$.

  We now construct the edge-verification gadgets.
  Let $(e_1, \ldots, e_{m_H})$ be an arbitrary ordering of the edges in $E(H)$.
  For each $p \in [m_H]$ proceed as follows.
  Introduce three vertices $z^p_1$, $z^p_2$, and $z^p_3$.
  Let $\{i, j\} = e_p$ where $i > j$.
  For each edge $e \in E_i \cap E_j$ of $G$ proceed as follows.
  Let $a(e)$ be the index of $e$ in the ordering~$(e^i_a)$ defined for vertex~$i$ when constructing the vertex-selection gadget.
  Similarly, let $b(e)$ be the index of $e$ in the ordering $(e^j_a)$ defined for~$j$.
  Introduce the following edges into \consg:
  \[\{z^p_1, x^i_{a(e)}\}, \qquad \{x^i_{a(e)}, z^p_2\}, \qquad \{z^p_2, x^j_{b(e)}\}, \qquad \text{and} \qquad \{x^j_{b(e)}, z^p_3\}.\]
  Furthermore, add the following transitions:
  \begin{itemize}
  \item $\{\{z^p_1, x^i_{a(e)}\}, \{x^i_{a(e)}, z^p_2\}\}$ to $T(x^i_{a(e)})$,
  \item $\{\{z^p_2, x^j_{b(e)}\}, \{x^j_{b(e)}, z^p_3\}\}$ to $T(x^j_{b(e)})$, and
  \item $\{\{x^i_{a(e)}, z^p_2\}, \{z^p_2, x^j_{b(e)}\}\}$ to $T(z^p_2)$.
  \end{itemize}
  \looseness=-1
  To conclude the construction of the edge-verification gadgets, add the following edges: $\{t_1, z^1_1\}$; for each $p \in [m_H - 1]$ the edge $\{z^p_3, z^{p + 1}_1\}$; and $\{z^{m_H}_3, t\}$.
  This concludes the construction of \consg\ and $T(\consg)$ (recall that for vertices $v$ for which we left $T(v)$ unspecified we put $T(v) = \binom{E(v)}{2}$).
\end{construction}

Observe that \cref{cons:lb} can be carried out in polynomial time.
We claim that the distance to linear forest of \consg\ is at most $n_H + 3m_H \leq 5m_H$.
Let $Z = \{z^p_1, z^p_2, z^p_3 \mid p \in [m_H]\}$ and $Y = \{y^i \mid i \in [n_H]\}$.
Note that the only vertices in $\consg - (V(P) \cup \{t\})$ are in $Y \cup Z$.
Moreover, no edges between two vertices on~$P$ have been introduced into \consg.
Thus, $Y \cup Z$ is a modulator to a linear forest and $\consg$ has distance at most $5m_H$ to a linear forest.
If \cref{cons:lb} is correct, by the properties of \subiso\ it thus follows that deciding whether a graph has a compatible $s$-$t$ path is W[1]-hard with respect to the distance, $k$, to a linear forest, and that an $f(k)n^{\smalloh(k/\log k)}$-time decision algorithm contradicts the ETH.
We next show the correctness of \cref{cons:lb}.

\paragraph{Correctness.}
We now show that $(\consg, T)$ contains a compatible $s$-$t$ path if and only if there is a subgraph isomorphism from $H$ into~$G$.

Suppose first that there is a subgraph isomorphism $\phi$ from $H$ into~$G$.
Construct an $s$-$t$ walk $P^\star$ by concatenating the following path segments (observe while reading the construction, that $P^\star$ is compatible):
\begin{enumerate}
\item The subpath on $P$ from $s$ to $\segpre(\phi(1))$.
\item The three vertices $\segpre(\phi(1))$, $y^1$, $\segpost(\phi(1))$.
\item For each $i = 2, 3, \ldots, n_H$ take:
  \begin{enumerate}
  \item The subpath on $P$ from $\segpost(\phi(i - 1))$ to $\segpre(\phi(i))$.
  \item The three vertices $\segpre(\phi(i))$, $y^i$, $\segpost(\phi(i))$.
  \end{enumerate}
\item The subpath on $P$ from $\segpost(\phi(n_H))$ to $t_1$.
\item For each $p = 1, 2, \ldots, m_H$, let $e_p$ be the $p$th edge of
  $H$ according to the ordering of $E(H)$ fixed in \cref{cons:lb}, let $e_p = \{i, j\}$, where $i > j$, let $e = \{\phi(i), \phi(j)\}$, let $a(e)$ be the index of $e$ in the ordering $(e^i_a)$ and $b(e)$ the index of $e$ in the ordering $(e^j_a)$.
  Take the %
  vertices $z^p_1$, $x^i_{a(e)}$, $z^p_2$, $x^j_{b(e)}$, and~$z^p_3$.
\item The edge $\{z^{m_H}_3, t\}$.
\end{enumerate}
\looseness=-1
This concludes the construction of $P^\star$.
Suppose, for a contradiction, that $P^\star$ is not a path, that is, there is a vertex $v$ in \consg\ which is contained twice in $P^\star$.
Since $V(G)$ is partitioned into $V(P)$, $Y$, $Z$, and $\{t\}$ and each vertex of $Y$ and $Z$ occurs only once in the definition of $P^\star$, we have $v \in V(P)$.
Since each segment in the construction of $P^\star$ is a path, the two occurrences must be in different segments.
Observe that all segments of~$P^\star$ in steps~1 to~4 that are contained in $V(P)$ are pairwise disjoint subpaths of $P$.
Furthermore, all vertices in $V(P)$ used in the segments constructed in step~5 are pairwise distinct.
Thus, there is one occurrence of $v$ in steps~1 to~4, and one in step~5.
Moreover, $v$ corresponds to some edge $e$ of $G$.
However, according to the steps~1 to 4, vertex $v$ corresponds to some edge which is not incident to a vertex in $\phi(V(H))$ and, according to step~5, vertex $v$ corresponds to some edge which is incident to a vertex in $\phi(V(H))$, a contradiction.
Thus, indeed, $P^\star$ is a compatible $s$-$t$ path, as required.

Now suppose that $(\consg, T)$ contains a compatible $s$-$t$ path $P^\star$.
Obviously, $P^\star$ starts with a subsegment of~$P$.
By construction of the transitions on vertices on $P$, at each internal vertex of $P$, the path $P^\star$ may either continue on $P$ or go to some vertex of~$Y$.
Moreover, whenever $P^\star$ traverses a vertex of $Y$, it immediately returns to $P$ with the next vertex.
Path $P^\star$ hence begins with a segment which starts at $s$, alternatingly contains a sequence of vertices on $P$ and a vertex of $Y$, and ends at $t_1$.
Let $Y' = Y \cap V(P^\star)$ (we show below that $Y' = Y$).
Observe that, for each vertex $y^i \in Y'$, there exists $v \in V_i$ such that $P^\star$ contains the edges $\{\segpre(v), y^i\}$ and $\{y^i, \segpost(v)\}$, by the transitions defined for~$y^i$.
Define a (partial) function $\phi \colon V(H) \to V(G)$ as follows.
For each $i \in [n_H]$ such that $y_i \in Y'$ put $\phi(i) = v$, where $v$ is as defined above.
For later, put $P^\star_1$ to be the segment of $P^\star$ from $s$ to $t_1$ and put $P^\star_2$ to be the segment of $P^\star$ from $t_1$ to $t$.
Observe that $P^\star_1$ contains precisely all vertices of $P$ except those that correspond to edges in $G$ which are incident to the vertices of~$\phi(Y')$.

To show that $\phi$ is total and that $\phi$ is a subgraph isomorphism from $H$ into~$G$, we now argue that $P^\star_2$ contains $z^p_2$ for each $p \in [m_H]$.
Since $P^\star_2$ is a path, it starts with the edge $\{t_1, z^1_1\}$.
Moreover, by the edges and transitions of the vertices $z^p_1$, $x^i_a$, $z^p_2$, and $z^p_3$ ($p \in [m_H]$, $i \in [n_H]$, $a \in \mathbb{N}$), whenever $P^\star_2$ traverses a vertex~$z^p_1$, $p \in [m_H]$, it next traverses some vertex $x^i_a$, then the vertex $z^p_2$, some vertex $x^j_b$, and the vertex $z^p_3$ for some $i, j \in [n_H]$ where $i > j$.
Moreover, after $z^p_3$, path $P^\star_2$ traverses either $z^{p + 1}_1$ (if $p < m_H$) or $t$ (if $p = m_H$) because the only other vertices that $P^\star_2$ may traverse after $z^p_3$ are vertices $x^j_{a'}$ and, by their transitions, $P^\star_2$ would then have to contain $z^p_2$ a second time.
Concluding, $P^\star_2$ contains $z^p_2$ for each $p \in [m_H]$.

Let $p \in [m_H]$ and let $e_p$ be the $p$th edge of $H$ according to the ordering of $E(H)$ fixed in \cref{cons:lb}.
Let $e_p = \{i, j\}$ with $i > j$.
As argued above $P^\star_2$ contains $z^p_2$.
Let $x^i_a$ and $x^j_b$ be the vertices that $P^\star_2$ traverses before and after $z^p_2$.
By the transitions of $z^p_2$, the vertices $x^i_a$ and $x^j_b$ correspond to the same edge of~$G$.
Denote this edge by~$f_p$.
We now show that the edges $f_p$, $p \in [m_H]$, ensures that $\phi$ is total and a subgraph isomorphism.

First, to see that $\phi$ is total, recall that each vertex $i \in V(H)$ is incident with at least one edge.
Say $i$ is incident with edge~$e_p$.
Let $x^i_a$ be the vertex that corresponds to an edge in $G$ incident with a vertex of color $i$ and that led to the definition of $f_p$, that is, $P^\star_2$ traverses $x^i_a$ before or after~$z^2_p$.
Now recall that $P^\star_1$ contains all vertices of $P$ except those that correspond to the edges incident with vertices in $\phi(Y')$.
Since $P^\star_1$ and $P^\star_2$ are internally vertex-disjoint, $i \in Y'$.
It thus follows that $\phi$ is total.

\looseness=-1
To see that $\phi$ is a subgraph isomorphism, take any edge $e_p \in E(H)$.
Consider the edge $f_p$ and the two vertices $x^i_a$ and $x^j_b$ that led to the definition of~$f_p$, that is, $x^i_a$ and $x^j_b$ are traversed either before or after $z^2_p$.
By the construction of the edges of $z^2_p$, we have $e_p = \{i, j\}$.
We again use the property that $P^\star_1$ contains all vertices of $P$ except those that correspond to the edges incident with vertices in $\phi(Y')$.
Since $x^i_a$ and $x^j_b$ are not in $P^\star_1$, they correspond to an edge incident with both $\phi(i)$ and $\phi(j)$, that is, $f_p = \{\phi(i), \phi(j)\}$.
Thus, indeed $\phi$ is a subgraph isomorphism, as required.
This concludes the proof of \cref{thm:lb1} \cref{lb:path}.
Observe that~\cref{lb:long-path} is implied by \cref{lb:path}.
The remaining parts are proved below.

\paragraph{Cycles.}
We now adapt \cref{cons:lb} to obtain \cref{thm:lb1}~\cref{lb:cycle}.
To this end, we simply add the edge $\{s, t\}$ to~\consg\ (and update the permitted transitions of $s$ and $t$ to allow for combining $\{s, t\}$ with every other edge).
Call the resulting graph $\consg_C$.
Observe that $\consg_C - (Y \cup Z)$ is a path with vertex set $V(P) \cup \{t\}$, and hence $\consg_C$ has distance to a linear forest at most~$5m_H$.

We claim that there is a compatible $s$-$t$ path in \consg\ if and only if there is a compatible cycle in $\consg_C$.
The forward direction is trivial.
For the backward direction, let $C^\star$ be a compatible cycle in $\consg_C$.
We show that $C^\star$ contains $\{s, t\}$.
For a contradiction, assume it does not.
Thus, $C^\star$ is a cycle in \consg.
By the transitions of the vertices in $P$, cycle $C^\star$ does not contain an edge in $P$ nor does it contain a vertex in~$Y$.
Let $\consg_1 = (V(\consg) \setminus Y, E(\consg) \setminus E(P))$ and observe that $C^\star$ is a cycle in $\consg_1$.
Observe that $V(P)$ is an independent set in $\consg_1$.
Thus each cycle (not necessarily compatible) can be written as $z^p_2, x^i_a, z^p_1, x^i_b, z^p_2$ or $z^p_2, x^i_a, z^p_3, x^i_b, z^p_2$ for the corresponding values of $p$, $i$, \(a\), and~$b$.
However, by the transitions of $z^p_2$, none of these cycles is compatible, a contradiction.
Thus, $C^\star$ contains $\{s, t\}$.
Hence, removing $\{s, t\}$ from $C^\star$ gives an $s$-$t$ path in \consg, concluding the proof.
Observe that \cref{lb:long-cycle} follows from~\cref{lb:cycle}.
\end{proof}

We now adapt \cref{cons:lb} to prove that it is W[1]-hard with respect to the distance to treewidth two to check whether there is a  compatible Hamiltonian cycle.

\begin{theorem}\label[theorem]{thm:lb2}
  Let $G$ be a graph and $k'$ its distance to treewidth two.
  It is W[1]-hard with respect to~$k'$ to decide whether $G$ contains a compatible Hamiltonian cycle and, moreover, an $f(k')\cdot n^{\smalloh(k'/\log k')}$-time decision algorithm contradicts the~ETH.
\end{theorem}

\newcommand{\consg}{\ensuremath{G^\star}}

\begin{proof}
To prove this theorem, we use \cref{cons:lb} and add a gadget that allows an $s$-$t$ path in \consg\ to collect all so-far untraversed vertices, wherein we use transitions to not disturb the structure of \consg.
The basic observation that we use is that the path $P^\star$ we have constructed in the correctness proof for detecting $s$-$t$ paths above contains all vertices of $G^\star$ except segments of the path~$P$.
The idea now is to add a path~$Q$ which runs ``parallel'' to $P$ (like a skewed ladder) and which starts after $t$ and ends in~$s$.
Using transitions we allow the solution in each vertex $v$ of $Q$ to either continue to the next vertex of $Q$ or to traverse the vertex parallel to $v$ on $P$ and then immediately return to the next vertex after $v$ on~$Q$.
This allows the solution to traverse all vertices it missed on the traversal from $s$ to~$t$.
Since $Q$ is parallel to $P$, removing $Y \cup Z$ will result in a graph of treewidth two.

The formal construction is as follows.
Construct a forbidden-transition graph $(\consg_1, T_1)$ from $(\consg, T)$ by initially putting $(\consg_1, T_1) = (\consg, T)$.
Let $n = |V(P)| - 2$.
Add a path $Q$ consisting of $n + 1$ vertices to $\consg_1$ and identify the first and last vertex of $Q$ with $s$ and $t$, respectively.
Let $v_1, v_2, \ldots, v_n$ be the internal vertices of~$P$ and $t = u_1, u_2, \ldots, u_{n + 1} = s$ the vertices of~$Q$.
For each $i \in [n]$ proceed as follows.
Add the edges $\{u_i, v_i\}$, and $\{v_i, u_{i + 1}\}$.
Then, update the transition system $T_1$ by adding the transitions $\{\{u_i, v_i\}, \{v_i, u_{i + 1}\}\}$ to $T_1(v_i)$.
This finishes the construction of $\consg_1$ and its transition system (as before, for all vertices with unspecified transition systems we allow all transitions).

Let $\tilde{\consg_1} = \consg_1 - (Y \cup Z \cup \{s, t\})$.
We claim that $\tilde{\consg_1}$ has treewidth two.
Observe that this graph consists only of the vertices in $P$ and $Q$ except for~$s$ and~$t$.
Now observe that, by the definition of the edges between $P$ and $Q$, the following bags give a path decomposition for $\tilde{\consg_1}$ of width two. Note that we specify a bag containing $t = u_1$ for easier notation:
\[\{u_1, u_2, v_1\}, \{u_2, v_1, v_2\}, \ldots, \{u_i, u_{i + 1}, v_i\}, \{u_{i + 1}, v_i, v_{i + 1}\}, \ldots, \{u_{n - 1}, u_n, v_{n - 1}\}, \{u_n, v_{n - 1}, v_n\}\text{.}\]
Thus $Y \cup Z \cup \{s, t\}$ is a modulator of $\consg_1$ to treewidth two, meaning that $\consg_1$ has distance to treewidth two at most $5m_H + 2$, as required.

We claim that $(\consg_1, T_1)$ contains a compatible Hamiltonian cycle if and only if $(\consg, T)$ contains a compatible $s$-$t$ path.
Take a compatible $s$-$t$ path $P^\star$ in $(\consg, T)$.
Observe that in the correctness proof for detecting compatible $s$-$t$ paths we have argued that $P^\star$ contains all vertices in~$Y$ and~$Z$.
In other words, the only vertices of $\consg_1$ that $P^\star$ does not contain are in $P$ and $Q$.
To obtain a compatible Hamiltonian cycle in $(\consg_1, T_1)$ from $P^\star$, simply extend the path after arriving at $t$ by following $Q$ and visiting vertices in $V(\consg_1) \setminus Q$ as needed in order to visit all vertices of~$\consg_1$.
In the other direction, by the updated transitions, a compatible Hamiltonian cycle in $\consg_1$ decomposes into two $s$-$t$ paths, one path containing vertices of $Q$ (and a subset of vertices in $P$), and another path contained in \consg, as claimed.
This concludes the proof of W[1]-hardness of detecting compatible Hamiltonian cycles with respect to the distance to treewidth two and thus concludes the proof of \cref{thm:lb2}.
\end{proof}

%% file: treecutwidth.tex
In this section we are going to prove that the \compath problem is fixed-parameter tractable with respect to treecut-width of $G$; the treecut-width is defined below.
In the \compath problem we get a forbidden-transition graph $(G, T_G)$ and two vertices $s,t \in V(G)$, and we want to decide whether there exists a compatible $s$-$t$ path in~$G$. The notion of treecut-width of a graph was introduced by Wollan \cite{Wollan2015}.
In this section, to keep the notation succinct, we will sometimes refer to a forbidden-transition graph $(G, T_G)$ simply as $G$ and say that $G$ \emph{has} transition system~$T_G$.

\paragraph{Basic definitions and previous results.}
We now define treecut-width and the associated treecut decompositions, and we recap the relevant previous results on computing them.
Consider a graph $G$ (without transitions) and let $v \in V(G)$ be a vertex of degree at most two.
To \emph{suppress} a vertex $v \in V(G)$ means (i) to add an edge between $v$'s two neighbors and (ii) to delete~$v$. 
For a partition $A \cup B$ of the vertex set of $G$ such that both $A$ and $B$ are non-empty, we denote by $E(A,B)$ the cut-set between $A$ and $B$, i.e., the set $\{uv \in E(G): u \in A, v \in B\}$.
Recall that for a vertex $v$ of $G$ by $E_G(v)$ we denote the set of edges incident with $v$ (we omit the subscript if it is clear from the context). We define \emph{shrinking} a (not necessarily connected) set $Q \subseteq V(G)$ into $q$ as an operation which replaces $Q$ in $G$ by a single new vertex~$q$, and adds an edge $qv$ for every edge $uv \in E(G)$ such that $u \in Q$ and $v \notin Q$. Note that this may create parallel edges.

A \emph{treecut decomposition} of a graph $G$ is a pair $(\T, \X)$ such that $\T$ is a rooted tree and $\X=\{X_t\}_{t \in V(\T)}$ is a partition of vertices of $G$ in which we allow sets $X_t$ to be empty.
For a node $t$, let $\T_t$ be the subtree of $\T$ rooted in~$t$. 
Let $Y_t:=\bigcup_{t' \in V(\T_t)}X_{t'}$ and $Z_t:=V(G) \setminus Y_t$. 
For a non-root node $t$, let $E_t:=E(Y_t,Z_t)$, and for a root $r$, we set $E_r:=\emptyset$. We denote the vertices from $Y_t$ by $y_1, y_2, y_3,\ldots$, and the vertices from $Z_t$ by $z_1, z_2, z_3,\ldots$. 

The \emph{torso} of $(\T, \X)$ at a node $t$ is a graph $H_t$, constructed as follows.
If $\T$ consists of a single node~$t$, then the torso at~$t$ is $G$.
Otherwise, let $C_1,\ldots, C_\ell$ be the sets of the vertices of $G$ corresponding to the connected components of $\T \setminus \{t\}$.
We construct the torso by shrinking $C_i$ into $c_i$ for each $i \in [\ell]$.
Note that a torso may have parallel edges.
The \emph{3-center} $\widetilde{H}_t$ of a torso $H_t$ is the graph obtained from $H_t$ by suppressing vertices of degree at most two which belong to the set $V(H_t) \setminus X_t$.
The \emph{width} of a treecut decomposition $(\T, \X)$ is $\max\{|E_t|,|V(\widetilde{H}_t)| \colon t \in V(\T)\}$.
The \emph{treecut-width} of $G$ is the minimum width of a treecut decomposition of $G$.

\begin{theorem}[Kim et al.~\cite{KimOPST2018}]\label{thm:approx-decomp}
There exists an algorithm that, given a graph $G$ and $k \in \mathbb{N}$, either outputs a treecut decomposition of $G$ of width at most $2k$ or correctly reports that $G$ has treecut width larger than $k$ in time $2^{\Oh(k^2\log k)} \cdot |V(G)|^2$.
\end{theorem}

A non-root node $t$ of a treecut decomposition $(\T,\X)$ is \emph{thin} if $|E_t| \leq 2$ and it is \emph{bold} otherwise. Denote by $A_t$ and $B_t$, respectively, the set of all bold and thin children of $t$. A treecut decomposition of~$G$ is \emph{nice} if for every thin node $t \in V(\T)$ we have that $N(Y_t) \cap (\bigcup \{Y_b: b\textnormal{ is a sibling of }t \text{ in } \T\}) = \emptyset$.

\begin{theorem}[Ganian et al.~\cite{GanianKS2015}]\label{thm:nice-decomp}
There exists an algorithm working in time $\Oh(|V(G)|^3)$ which transforms any rooted treecut
decomposition $(\T,\X)$ of $G$ into a nice treecut decomposition of the same graph, without
increasing its width or number of nodes. 
\end{theorem}

The following property of nice decompositions is extremely useful in designing the dynamic programming algorithms.

\begin{theorem}[Ganian et al.~\cite{GanianKS2015}]\label{thm:bold-children}
  Let $t$ be a node of a nice treecut decomposition of width $k$. Then $|A_t| \leq 2k+1$.
\end{theorem}

\paragraph{Auxiliary problems and algorithms.}
To show that \compath\ is fixed-parameter tractable with respect to the treecut-width we provide a dynamic-programming algorithm on the treecut decomposition.
In the individual steps of the dynamic program we will need to solve the more general problem of finding compatible vertex-disjoint paths between given pairs of vertices.
We now introduce this problem and a restricted variant that occurs in a special case of the dynamic program, and we provide a fixed-parameter algorithm for the more restricted variant.

Let $G$ be a graph.
We say that $W$ is \emph{a set of terminal pairs} of $G$, if it consists of pairwise-disjoint two-element subsets of~$V(G)$.
If $W$ is clear from the context, we also simply call the elements of $W$ \emph{terminal pairs} and each of them consists of \emph{terminals}.
Recall that by $T_G$ we denote the transition system of $G$ (again, we omit the subscript if it is clear from the context).
\defparproblem{\textsc{Compatible Vertex-Disjoint Paths} (\comvdpaths)}
{An instance $(G,T_G,W)$ where $(G, T_G)$ is a forbidden transition graph and $W$ is a set of terminal pairs of~$G$.}
{The treecut-width, \(k\), of \(G\).}
{Are there pairwise vertex-disjoint $T_G$-compatible paths in $G$ connecting each pair in~$W$?}

Below we will also say that a set of pairwise disjoint paths as above is a \emph{solution} to the instance $(G, T_G, W)$.
Clearly, an instance $(G,T_G,s,t)$ of \compath is equivalent to an instance $(G,T_G,\{\{s,t\}\})$ of \comvdpaths. 

Our dynamic-programming approach is based on the XP-algorithm for finding edge disjoint paths between given pairs of terminals, described by Ganian and Ordyniak \cite{GanianO2019}. 
We first show fixed-parameter tractability of a simpler variant of \comvdpaths, called \simplecomvdpaths, with additional assumptions on the structure of the input: Essentially, the vertex set is partitioned into a small set of arbitrary structure and a possibly large set of maximum degree two. 
A fixed-parameter algorithm for \simplecomvdpaths will later be used when solving the general case.

\defparproblem{\simplecomvdpaths}
{An instance $(G,T_G,W)$ of \comvdpaths\ and a partition of $V(G)$ into two sets $A,B$ such that $B$ consists of vertices of degree at most~2.}
{$|A| \in \mathbb{N}$}
{Are there pairwise vertex-disjoint $T_G$-compatible paths in $G$ connecting each pair in~$W$?}

Observe that each instance of \simplecomvdpaths has treecut-width at most $|A|$. Indeed, we may construct a treecut decomposition of $G$ as follows:
All vertices of $A$ are contained in a bag $X_r$ and each vertex $b \in B$ forms a separate bag $X_b$. We define $\T$ to be a star, with $r$ as its center and root. It is straightforward to verify that $(\T, \{X_t\}_{t \in\{r\} \cup B})$ is a treecut decomposition of $G$ of width~$|A|$.

We now generalize the notion of suppressing a vertex to forbidden-transition graphs.
Let $(G, T)$ be a forbidden-transition graph and $v$ a vertex of degree at most 2 in~$G$.
If $v$ is of degree 0 or 1, or $T(v)$ does not contain any transition, we delete~$v$.
Otherwise, $T(v)$ consists of a single transition, say~$\{uv,vw\}$.
\begin{itemize}
\item If $uw \notin E(G)$, then we proceed as follows. We add $uw$ to
  $E(G)$.
  In each transition of $T(u)$ we replace $uv$ with $uw$ and in each transition of $T(w)$ we replace $vw$ with $uw$.
  That is, if there is a transition~$t$ in $T(u)$ (resp.\ in $T(w)$) with $t = \{uv, ux\}$ (resp.\ with $t = \{vw, wx\}$) for a vertex $x \in V(G) \setminus \{u, v, w\}$, then we put $T(u) := T(u) \setminus\{\{uv, ux\}\} \cup \{\{uw, ux\}\}$ (resp.\ $T(w) := T(w) \setminus \{\{vw, wx\}\} \cup \{\{uw, wx\}\}$).
\item If $uw \in E(G)$, then for each vertex
  $x \in V(G) \setminus \{u, v, w\}$, if $\{ux, uv\} \in T(u)$ then we add $\{ux, uw\}$ to $T(u)$ and if $\{wx, vw\} \in T(w)$ then we add $\{wx, uw\}$ to $T(w)$.
\end{itemize}
Regardless of whether $uw \in E(G)$, we remove $v$ from $G$ (and remove the corresponding transitions).

Observe that the transitions in a forbidden-transition graph $(G', T')$ obtained from $(G, T)$ by suppressing~$v$ such that $N_G(v)=\{u,w\}$ are defined in such a way that we are allowed to use edge $uw$ in a compatible path $S$ in $(G', T')$ if and only if $S$ is a compatible path also in $(G, T)$ or a path obtained from $S$ by replacing the edge $uw$ by two consecutive edges $uv,vw$ is compatible in $(G, T)$.
This implies that suppressing a non-terminal vertex $v$ is a safe reduction rule, that is, it does not affect the existence of a solution to \comvdpaths or \simplecomvdpaths.

We are now ready to solve \simplecomvdpaths.

\begin{lemma}\label{lem:simplecvdp}
There exists an algorithm solving \simplecomvdpaths in time $k^{\Oh(k)} + O(n^2)$, where $k = |A|$ and $n = |V(G)|$.
\end{lemma}
\begin{proof}
Let $J = (G, T, W)$ be an instance of \simplecomvdpaths. We start with simple sanity checks.
First, observe that if $|W| > n$ then $J$ is clearly a no-instance as we cannot find more than $n$ vertex-disjoint paths in~$G$. 
Similarly, if there exists a vertex which belongs to more than one pair in $W$, then $J$ must be a
no-instance. 
Performing the sanity checks takes $\Oh(n^2)$ time.

Consider a vertex $v \in B$. If $v$ does not belong to any pair in $W$, then we suppress it.
Recall that this preserves the solution.
Therefore, we can assume that for each vertex $v \in B$ there exists a unique vertex $v' \in V(G)$ with $v' \neq v$ such that $\{v, v'\} \in W$.
For every $v \in B$ such that $N(v) \subseteq B$, we check whether $v' \in N(v)$. 
If yes, we can safely remove $v$ and $v'$ from $G$ and $\{v, v'\}$ from $W$. 
Otherwise, we report that $J$ is a no-instance -- since all vertices in $B$ are terminal vertices,
there is no way to connect $v$ and $v'$.  
After this step, each vertex in $B$ must have a neighbor in $A$.
Moreover, observe that there is no $\{v,v'\} \in W$ such that $v, v' \in B$ and $v' \in N(v)$, so each path in a potential solution must intersect $A$.

Next, if $|B| > 2|A|$ we report that $J$ is a no-instance.
Indeed, in order for $J$ to be a yes-instance, since each vertex in $B$ is a terminal, and each path in the solution containing a vertex of~$B$ must also contain a vertex of~$A$, for each two vertices in $B$ there must exist a vertex in~$A$.
Thus, $|V(G)| \leq 3k$, and an algorithm guessing the partition of the vertex set into the desired paths solves the instance in $k^{\Oh(k)}$~time.
The statement now follows.
\end{proof}

Before we proceed to the general algorithm, we introduce a notion which will be useful when handling transitions.
Let $G'$ be an induced subgraph of $G$.
We say that a family $\C:=\{C_1, \ldots, C_\ell\}$ of disjoint subsets of $E\left(V(G'), V(G)\setminus V(G')\right)$ is \emph{terminable} (with respect to $G'$) if for every $i \in [\ell]$ set $C_i$ is either (i) a singleton or (ii) contains exactly two edges $uu'$ and $vv'$ such that $u',v' \in V(G')$ and $u'\neq v'$.
We omit reference to the graph $G$ that contains $G'$ if it is clear from the context.
We now define the operation of terminating a terminable set~$\C$.
Intuitively, this operation returns a modified $G'$ in which all edges in the subsets of $\C$ are added and, if two edges were from the same subset of $\C$, then their endpoints outside of $G'$ are merged into one vertex and the transition over this vertex is made to be allowed.
\begin{definition}\label{def:terminating}
  Let $G$ be a simple graph with transition system $T_G$. Consider an induced subgraph $G'$ of $G$ and let $\C:=\{C_1, \ldots, C_\ell\}$ be terminable with respect to $G'$. We define the operation of \emph{terminating $\C$ in $G'$} which results in a graph $G'_\C$ with transition system $T_{G'_\C}$ as follows.
  The resulting graph $G'_\C$ has vertex set $V(G') \cup V'$, where the elements of $V'$ are $c(C_1),\ldots,c(C_\ell)$. The edge set of $G'_\C$ is defined as follows. \[E(G'_\C):=E(G') \cup \{uc(C_i): u \in V(G')\textnormal{ and there exists an edge in $C_i$ adjacent to $u$ in $G$}\}.\] For simplicity we will sometimes write $c_i$ instead of $c(C_i)$, if it causes no confusion.

The transition system $T_{G'_\C}$ of $G'_\C$ is defined as follows. 
\begin{itemize}
\item If $u \notin V'$, then for $w,z \in N(u)$ a transition $\{wu,uz\}$ belongs to $T_{G'_\C}(u)$ if and only if
\begin{itemize}
\item $w, z \notin V'$ and $\{wu,uz\} \in T_G(u)$, or
\item $w \notin V'$, $z = c_i \in V'$ and $\{wu,e\} \in T_G(u)$, where $e$ is the unique edge from $C_i$ adjacent to~$u$,~or
\item $w=c_i, z = c_j \in V'$, for $i \neq j$, and $\{e,f\} \in T_G(u)$, where $e$ and $f$ are, respectively, the unique edges from $C_i$ and $C_j$ adjacent to $u$.
\end{itemize}
\item If $u =c_i \in V'$, then $T_{G'_C}(u)$ contains all unordered pairs of edges incident with $u$.
\end{itemize} 
Note that since each $C_i$ has at most 2 elements, all vertices from $V'$ have degree at most 2.
\end{definition} 
Observe that after terminating some $\C$ in $G'$, we always obtain a simple graph. Moreover, the following observation is straightforward.
\begin{observation}\label{obs:terminating}
Let $G'$ be an induced subgraph of $G$ and let $\C=\{C_1,\ldots,C_\ell\}$ be a terminable set with respect to $G'$. Let $P= (p_1, p_2, \ldots, p_m)$ be a compatible path in $G$ such that at least one $p_i$ belongs to $V(G')$. Denote by $e_1, e_2, \ldots, e_r$ the edges of $E(V(G'), V(G) \setminus V(G')) \cap P$ in the order in which they appear in~$P$.
\begin{enumerate}
\item If $p_1,p_m \in V(G')$ and for each odd $j \in [r-1]$ we have $\{e_j, e_{j+1}\} \in \C$ then there exists a compatible $p_1$-$p_m$ path in~$G'_\C$. 
\item If $p_1 \in V(G')$ (resp.\ $p_1 \notin V(G')$), and for some $i,i' \in [\ell]$ and even (resp.\ odd) $j \in [r-1]$ we have $C_{i}=\{e_j\}, C_{i'}=\{e_{j+1}\} \in \C$, then there exists a compatible $c_i$-$c_{i'}$ path in~$G'_\C$,
\item If for some $i \in [\ell]$ we have $C_{i}=\{e_1\} \in \C$ (resp.\ $C_{i}=\{e_r\} \in \C$), then there exists a compatible $p_1$-$c_{i}$ path (resp.\ $c_{i}$-$p_m$ path) in~$G'_\C$. 
\end{enumerate}
\end{observation}

\paragraph{Notions for dynamic programming.}
We now introduce the notions used in our dynamic-programming approach and give results that we will later need to prove the correctness.
Let $(G,T_G,W)$ be an instance of \comvdpaths and let $(\T,\X)$ be a treecut decomposition of $G$ of width~$k$.
For a set $X \subseteq V(G)$, by $W[X]$ we denote the subset of terminal pairs from $W$ with both elements in $X$. 
For a node $t \in V(\T)$ let $G_t=G[Y_t]$.
An \emph{unmatched terminal} for $t$ is a vertex $x \in Y_t$ such that there exists a vertex $y \in Z_t$ satisfying $\{x, y\} \in W$.
We let $U_t$ be the set of unmatched terminals for~$t$.
We will use the fact that the number of unmatched terminals is bounded by the width of $(\T, \X)$ in every yes-instance:
\begin{observation}\label{obs:few-unmatched-terminals}
  If $(G,T_G,W)$ is a yes-instance, then for each node $t$ of $\T$ the number $|U_t|$ of unmatched terminals is at most~$k$.
\end{observation}
\begin{proof}
  Consider a solution to the instance and observe that this solution witnesses that there is a flow between $U_t$ and $Z_t$ of value at least $|U_t|$.
  Since $E_t$ is a $Y_t$-$Z_t$ cut in $G$ containing at most $k$ edges and $U_t \subseteq Y_t$, we have $|U_t| \leq |E_t| \leq k$.
\end{proof}
\noindent Since \cref{obs:few-unmatched-terminals} is easily checkable in polynomial time, we will from now on assume that for each node $t$ of $\T$ we have $|U_t| \leq k$.

As mentioned, we are going to describe a dynamic-programming procedure on the treecut decomposition $(\T,\X)$.
Below we introduce a basic notion which will be used to store the information about partial solutions.
\begin{definition}\label{def:record}
  A \emph{record} $R$ for $t \in V(\T)$ is a tuple $(\sigma,\I,\F,\lambda)$ consisting of the following elements.
  \begin{itemize}
  \item Function $\sigma$ is a partition of $E_t$ into sets $I$
    (\emph{internal}), $F$ (\emph{foreign}), $L$ (\emph{leaving}), and $U$ (\emph{unused}), such that for every $v \in V(G)$ we have $|E_G(v) \cap (I \cup F \cup L)| \leq 2$.
    Moreover, for each vertex $v \in V(G)$ with $|E(v) \cap (I \cup F \cup L)| = 2$ either (i) both elements of $E(v) \cap (I \cup F \cup L)$ belong to exactly one of~$I$ or~$F$, or (ii) $v \in Z_t$ and one element of $E(v) \cap (I \cup F \cup L)$ belong to to~$F$ and another one to~$L$.
  \item Set $\I$ is a perfect matching between the elements from $I$, such
    that if $\{y_iz_i,y_jz_j\} \in \I$ then $y_i \neq y_j$ (recall that $y_i, y_j$ denote vertices in $Y_t$ and $z_i, z_j$ vertices in~$Z_t$).
  \item Set $\F$ is a perfect matching between the elements from $F$, such
    that if $\{y_iz_i,y_jz_j\} \in \F$ then $z_i \neq z_j$.
  \item Finally, $\lambda$ is a bijection between $U_t$ and $L$.
\end{itemize}
\end{definition}

Observe that the conditions on $\sigma$ together with the conditions on the matchings $\I$ and $\F$ ensure that $\I \cup \{\{e\} \colon e \in F \cup L\}$ is terminable with respect to the subgraph $G[Y_t]$ of $G$ and that $\F \cup \{\{e\} \colon e \in I \cup L\}$ is terminable with respect to the subgraph $G[Z_t]$ of $G$.

Let $R(t)$ denote the set of all possible records for $t$.
Observe that $|R(t)| \leq 4^k \cdot (k!)^3 = k^{O(k)}$, as there are at most $4^k$ possibilities for $\sigma$, there are at most $k!$ possibilities each for the matchings $\I$ and $\F$, and at most $k!$ possibilities for $\lambda$ because there are at most $k$ unmatched terminals by \cref{obs:few-unmatched-terminals}.

Intuitively, the edges in sets $I, F$, and $L$ will correspond to edges in the solution paths: paths with both terminals in $Y_t$ use edges in $I$, paths with both terminals in $Z_t$ use edges in $F$ and paths with one terminal $y \in Y_t$ and another one $z\in Z_t$ use one edge of $L$ (which is supposed to be the first edge on the solution path from $y$ to $z$ which belongs to $E_t$) and their other edges should belong to~$F$.
The matchings should capture, in case of edges in $I$, which edge is used by path to leave $Y_t$ and then to come back, and in case of edges in $F$, by which edge we enter $Y_t$ and which one is then used to leave.
Finally, $\lambda$ says by which edge we leave $Y_t$ for the first time.
Below we introduce a notion which will help to formalize this intuition.

\begin{definition}\label{def:corr-inst}
  For an instance $(G,T,W)$, a node $t$ of a treecut decomposition $(\T,\X)$ of $G$ and a record $R = (\sigma,\I,\F,\lambda) \in R(t)$, we construct a \emph{corresponding instance} $(G_R,T_R,W_R)$ of \comvdpaths as follows.
  Let $\C=\{C_1,\ldots,C_\ell\}:= \I \cup \{\{e\} : e \in F \cup L\}$.
  Let the graph $G_R$, together with the transition system~$T_R$, be obtained by terminating $\C$ in $G_t$. Denote by $V_R=\{c(C_1),\ldots,c(C_\ell)\}$ the set of vertices $V(G_R) \setminus V(G_t)$. 
  The set $W_R$ contains
  \begin{enumerate}[(i)]
  \item every element of $W[Y_t]$, 
  \item the pair $\{c_i,c_j\}$ for every $c_i, c_j \in V_R$ such that $C_i \cup C_j \in \F$, and
  \item the pair $\{a,c_i\}$ for every $a \in U_t$ and $c_i \in V_R$ such that $C_i=\{\lambda(a)\}$.
  \end{enumerate}
\end{definition}
\noindent Note that the set $\C$ was defined in such a way that the pairs added to $W_R$ are disjoint.

A record $R$ is \emph{valid}, if its corresponding instance $(G_R,T_R,W_R)$ is a yes-instance of \comvdpaths. 
A corresponding instance should capture how the potential solution we construct behaves on $G_t$. 

\begin{definition}\label{def:correspondence}
Let $(\T,\X)$ be a treecut decomposition of $G$, let $t \in V(\T)$ and let  $J=(G,T,W)$ be an instance of \comvdpaths. Assume that there exists a solution $S=\{P_1, P_2, \ldots,P_{|W|}\}$ to $J$. We say that solution $S$ \emph{corresponds} to a record $(\sigma,\I,\F,\lambda)$ for the node $t$, if the following conditions are satisfied for every $a_i$-$b_i$ path $P_i \in S$ such that $E(P_i) \cap E_t \neq \emptyset$. Let $e^i_1, e^i_2, \ldots,e^i_{r_i}$ denote the elements of $P_i \cap E_t$ in the order of their appearance on~$P_i$.
\begin{enumerate}
\item If $a_i, b_i \in Y_t$, then $\{e^i_j, e^i_{j+1}\} \in \I$ for each odd $j \in [r_i-1]$.
\item If $a_i, b_i \in Z_t$, then $\{e^i_j, e^i_{j+1}\} \in \F$ for each odd $j \in [r_i-1]$.
\item If $a_i \in Y_t, b_i \in Z_t$, then $\lambda(a_i)=e^i_1$ (note that in this case $a_i$ is an unmatched terminal, that is, $a_i \in U_t$) and $\{e^i_j, e^i_{j+1}\} \in \F$ for each even $j \in [r_i-1]$.
\item An edge $e \in E_t$ belongs to $U$ if and only if $e \not\in \bigcup_{i \in [|W|]} E(P_i)$.
\end{enumerate}
\end{definition}

The dynamic-programming procedure on the treecut decomposition $(\T, \X)$ of $G$ computes the set $D(t)$ of valid records for a node $t$ of~$\T$. In other words, in $D(t)$ we store the information about these behaviors of a potential solution on $E_t$, which can be extended to $G_t$. Note that $(G,T,W)$ is a yes-instance of \comvdpaths if and only if $D(r)=\{(\emptyset,\emptyset,\emptyset,\emptyset)\}$. 

The last operation which will be used in the algorithm is the simplification of an instance. Intuitively, for a record $R$, the simplified instance is a smaller instance of the \comvdpaths problem which is equivalent to the original instance $J$ assuming that, if $J$ is an yes-instance, then $R$ corresponds to some solution for~$J$. Note that the construction of the simplified instance can be seen as dual to the construction of the corresponding instance.
\begin{definition}\label{def:simplification}
  Let $(\T,\X)$ be a treecut decomposition of $G$, let $t \in V(\T)$ and let  $J=(G,T,W)$ be an instance of \comvdpaths. The operation of \emph{simplification} of the instance $J$ in node $t$ in accordance with record $R = (\sigma,\I,\F,\lambda) \in R(t)$ returns an instance $(G_Q,T_Q,W_Q)$ as follows.
  Let $\C=\{C_1,\ldots,C_\ell\}:= \F \cup \{\{e\} : e \in I \cup L\}$.
  Graph $G_Q$ and its transition system $T_Q$ are obtained by doing the termination of $\C$ with respect to $G[Z_t]$. 
Let $V_Q:=\{c(C_1),\ldots,c(C_\ell)\}$.
The set of terminal pairs $W_Q$ contains
\begin{enumerate}[(i)]
\item every element of $W[Z_t]$, 
\item the pair $\{c_i,c_j\}$ for every $c_i, c_j \in V_Q$ such that $C_i \cup C_j \in \I$, and
\item the pair $\{c_i,b\}$ for every $\{a, b\} \in W$ with $a \in U_t$, $b \in Z_t$, and $c_i \in V_Q$ such that $C_i=\{\lambda(a)\}$.
\end{enumerate}
\end{definition}

Observe that each vertex in $V_Q$ has degree at most 2, and the degree of vertex in $V(G_Q) \setminus V_Q$ is at most its degree in $G$.

The following lemmata reveal how the introduced notions are related to each other. 

\begin{lemma}\label{lem:valid-rec}
Let $(\T,\X)$ be a treecut decomposition of $G$ and let $J=(G,T,W)$ be an instance of \comvdpaths which admits a solution $S$. Then for every node $t \in V(\T)$ there exists a unique record $R \in R(t)$ such that $S$ corresponds to $R$. 

On the other hand, if $S$ corresponds to some record $R$, then $R$ must be valid.
\end{lemma}
\begin{proof}
It is clear that for a fixed $S$ and $t$ there exists an unique $R=(\sigma,\I,\F,\lambda) \in R(t)$ satisfying conditions 1.-4. in the \cref{def:correspondence}. 

The fact that there exists a solution to the corresponding instance of $R$ follows from the construction of the corresponding instance and \cref{obs:terminating}.
\end{proof}

\begin{lemma}\label{lem:simpl-safe}
Consider an instance $J=(G,T,W)$ of \comvdpaths. Let $(\T,\X)$ be a treecut decomposition of~$G$, let $t$ be a fixed node of $\T$, and let $R=(\sigma,\I,\F,\lambda) \in R(t)$. Let $J_Q$ be the result of the simplification of $J$ in $t$ in accordance with~$R$.
\begin{enumerate}
\item If $J_Q$ admits a solution and $R$ is valid, then $J$ admits a solution.
\item If $S$ is a solution to $J$ and $S$ corresponds to $R$, then $J_Q$ admits a solution. 
\end{enumerate}
\end{lemma}
\begin{proof}
Let $J_Q=(G_Q,T_Q,W_Q)$. Recall that $V(G_Q)=V(G[Z_t]) \cup V_Q$, $V_Q=\{c(C_1), c(C_2), \ldots,c(C_\ell)\}$ and $\C=\{C_1, C_2, \ldots,C_\ell\}$ is a set terminated with respect to $G[Z_t]$. Note that in both statements of the lemma, $R$ must be valid -- in the first one by assumption and in the second one it follows from \cref{lem:valid-rec}. Let thus $J_R = (G_R,T_R,W_R)$ be a yes-instance corresponding to $R$, obtained by terminating $\C'=\{C'_1, C'_2, \ldots, C'_{\ell'}\}$. Let $V_R=\{c'(C'_1), c'(C'_2), \ldots,c'(C'_{\ell'})\}$ (we will also write $c'_i$ instead of $c'(C'_i)$) and let $S_R$ be a solution to $J_R$.

For the first statement, assume that $S_Q$ is a solution to $J_Q$. 
For each $\{a,b\} \in W$ we construct a compatible $a$-$b$ path $P^*$ to include in a solution for $J$ as follows.
We iteratively construct a sequence $P^*$ of elements of $V(G_R) \cup V(G_Q)$ with the property that, at each iteration, two consecutive vertices of $P^*$ either form an edge from $E(G_R) \cup E(G_Q)$ or from $E_t$. 
We claim that, at the end of the procedure, $P^*$~is a compatible $a$-$b$ path in $G$.

We first consider the case in which $a,b \in Y_t$.
Observe that there is an $a$-$b$ path $P \in S_R$ in~$G_R$. 
We start with $P^*$ being the consecutive vertices of~$P$.
If there are no vertices from the set $V_R$ in $P$, then $P^*$ is the desired $a$-$b$ path in $G$.
Otherwise, we proceed to Step~1.

Step 1 (Replacing vertices from $V_R$):
Denote by $e_1,e_2, \ldots, e_m$ the edges from $E(Y_t,V_R) \cap E(P)$ in the order in which they appear on $P$. 
Since $a,b \in Y_t$ and by the construction of $J_R$ we must have $\{e_1,e_2\}, \{e_3,e_4\}, \ldots, \{e_{m-1},e_m\} \in \I$.
This implies that for each $i \in \{1,3,\ldots,m-1\}$ there exists a $c(\{e_i\})$\nobreakdash-$c(\{e_{i+1}\})$ path in $S_Q$; call this path~$P_{i, i + 1}$.
We replace each element of $V_R$ adjacent to edges $e_i$ and $e_{i+1}$ in $P^*$ by the interior vertices of~$P_{i,i+1}$.
Observe that the result respects the transitions of~$G$ by the definition of termination of a set.
Observe that in this way we may have added some vertices from~$V_Q$ to the sequence~$P^*$.
If this has happened, we take care of them in Step 2.

Step 2 (Replacing vertices from $V_Q$):
Let $e'_1,e'_2, \ldots, e'_{m'}$ be the edges from $E(Z_t,V_Q) \cap E(P_{i,i+1})$ in the order in which they appear on $P_{i,i+1}$.
By the construction of $J_Q$ we have $\{e'_1,e'_2\},\ldots, \{e'_{m'-1},e'_{m'}\} \in \F$.
Moreover, for each $j \in \{1,3,\ldots,m'-1\}$ there exists a $c'(\{e'_j\})$\nobreakdash-$c'(\{e'_{j+1}\})$ path $P_{j,j+1}$ in~$S_R$.
We replace each element of $V_Q$ adjacent to edges $e'_j$ and $e'_{j+1}$ in $P^*$ by the interior vertices of the $c(\{e'_j\})$\nobreakdash-$c(\{e'_{j+1}\})$ path $P'_{j,j+1}$ from $S_R$.
Observe again that the result respects the transitions of~$G$ by the definition of termination of a set.
Again, in Step 2 we can add to the sequence some vertices from $V_R$.
If this happens, we go back to Step~1.

Since the paths in $S_R$ and $S_Q$ are pairwise disjoint, we never add a vertex to $P^*$ twice. 
Moreover, $|V_R|,|V_Q| \leq k$, so after at most $2k$ iterations, we obtain an $a$-$b$ path $P^*$ which uses only vertices from $Y_t \cup Z_t$, as required. 

The case in which $a,b \in Z_t$ is analogous; the only difference is that the initial path $P^*$ is taken from $S_Q$, and therefore it can contain vertices from $V_Q$ and, in that case, we start with Step~2.
If $a \in Y_t$ and $b \in Z_t$, we start with the $c'(\{\lambda(a)\})$\nobreakdash-$b$ path $P^*$ from $S_Q$ and, before performing Step~1, we replace $c'(\{\lambda(a)\})$ in $P^*$ by the vertices of the $a$-$c(\{\lambda(a)\})$ path in~$S_R$.

It is straightforward to verify that we use every path from $S_R$ and $S_Q$ at most once when we construct paths for all terminal pairs~$\{a,b\}$.
Therefore, since $S_R$ and $S_Q$ are sets of pairwise vertex-disjoint paths, we obtain a set $S$ of pairwise vertex-disjoint paths in $G$ connecting each pair in $W$.

For the second statement, observe that if $S$ is a solution to $J$, then we can derive a construction of every $a$-$b$ path in $G_Q$ from \cref{obs:terminating} (analogously to the proof of \cref{lem:valid-rec}).  
\end{proof}

\paragraph{The algorithm.} We are now ready to show how to proceed with a given instance $J=(G,T,W)$ of the \comvdpaths problem. 
Let $(\T,\X)$ be a treecut decomposition of $G$.

Observe that if $t$ is a leaf of $\T$ and $R \in R(t)$, the corresponding instance $(G_R, T_R, W_R)$ of $R$ is an instance of \simplecomvdpaths. Indeed, since $Y_t = X_t$, $Y_t$ has at most $k$ elements and the vertices in $V_R$ are of degree at most~2. This means that to compute the set $D(t)$ for a leaf $t$, for every element $R$ of $R(t)$ we find a corresponding instance and solve \simplecomvdpaths on it. Since $|R(t)| \leq k^{\Oh(k)}$, we obtain the following.
\begin{lemma} \label{lem:treecut-leaf}
There is an algorithm which takes as input $(G,T_G,W)$ of \comvdpaths, a treecut decomposition
$(\T,\X)$ of $G$ of width $k$ and a leaf $t \in V(\T)$, and computes $D(t)$ in time $k^{\Oh(k)} \cdot n^2$.
\end{lemma}

Next, we proceed to the non-leaf nodes.
A dynamic programming step will consist of three stages. For a non-leaf node $t$, we again construct a corresponding instance, but since its size does not have to be bounded by some function of $k$, we apply some further modifications.
First, we apply a reduction rule for each of the thin children (see below). Then we perform a simplification for each bold child~$t'$ of~$t$ and each $R \in R(t)$. After these, we argue that the graph obtained this way (for a fixed record $R$) is again an instance of \simplecomvdpaths, which can be solved efficiently.

Assume we solve an instance $(G,T_G,W)$ and let $t \in V(\T)$ be a non-leaf node. The safeness of the following reduction rule follows directly from its definition.
\begin{reduction*} Assume that $\T$ is a nice decomposition and let $s \in V(\T)$ be a thin child of~$t$.
  If $D(s)$ is empty, we report a no-instance.
  Otherwise, we proceed with the first option that applies on the following list.
  Herein, by \emph{terminating} a terminable set $\C$, we mean to replace $G$ by the result of the termination of $\C$ with respect to $G[Z_s]$. (In particular, this means to remove $Y_s$ from $G$.)
\begin{enumerate}
\item If $E_s=\{y_iz_i\}$, and:
\begin{itemize}
\item if $((y_iz_i \to L) ,\emptyset,\emptyset, a \mapsto y_iz_i) \in D(s)$, for some $\{a,b\} \in W$ such that $U_s=\{a\}$, then we terminate $\{\{y_iz_i\}\}$ and we replace $\{a,b\}$ in $W$ with $\{c(\{y_iz_i\}),b\}$.
\item if $((y_iz_i \to U),\emptyset,\emptyset, \emptyset) \in D(s)$ and $U_s=\emptyset$, then we remove $Y_s$ from~$G$.
\end{itemize}
\item If $E_s=\{y_iz_i,y_jz_j\}$, $U_s = \emptyset$, and:
\begin{itemize}
\item if $((y_iz_i,y_jz_j \to F),\emptyset,\{y_iz_i,y_jz_j\},\emptyset) \in D(s)$, then we terminate $\C=\{\{y_iz_i,y_jz_j\}\}$.
\item if $((y_iz_i,y_jz_j \to U),\emptyset,\emptyset,\emptyset) \in D(s)$, then we remove $Y_s$ from~$G$.
\item if $((y_iz_i,y_jz_j \to I),\{y_iz_i,y_jz_j\},\emptyset,\emptyset) \in D(s)$, then we terminate $\{\{y_iz_i\},\{y_jz_j\}\}$ and we add $\{c(\{y_iz_i\}),c(\{y_jz_j\})\}$ to~$W$.
\end{itemize}
\item If $E_s=\{y_iz_i,y_jz_j\}$, $U_s = \{a\}$, and $\{a,b\} \in W$, and:
\begin{itemize}
\item if $((y_iz_i \to L,y_jz_j \to U),\emptyset,\emptyset,a \mapsto y_iz_i) \in D(s)$ and $((y_jz_j \to L, y_iz_i \to U),\emptyset,\emptyset,a \mapsto y_jz_j) \in D(s)$, then we terminate $\{\{y_iz_i,y_jz_j\}\}$ and we add $\{c(\{y_iz_i,y_jz_j\}),b\}$ to~$W$.
\item if $(y_iz_i \to L, y_jz_j \to U),\emptyset,\emptyset,a \mapsto y_iz_i) \in D(s)$, then we terminate $\{\{y_iz_i\}\}$ and we add $\{c(\{y_iz_i\}),b\}$ to~$W$.
\end{itemize}
\item If $E_s=\{y_iz_i,y_jz_j\}$, $U_s = \{a_1,a_2\}$, and $\{a_1,b_1\},\{a_2,b_2\} \in W$ and:
\begin{itemize}
\item 
if $((y_iz_i,y_jz_j \to L),\emptyset,\emptyset,a_1 \mapsto y_iz_i,a_2 \mapsto y_jz_j) \in D(s)$ and 
$((y_iz_i,y_jz_j \to L),\emptyset,\emptyset,a_1 \mapsto y_jz_j,a_2 \mapsto y_iz_i) \in D(s)$, then
we terminate $\{\{y_iz_i,y_jz_j\}\}$. Let $c=c(\{y_iz_i,y_jz_j\})$. We add to $G$ a twin $c'$ of $c$, copying the transitions on the incident edges, and we add $\{c,b_1\}$ and $\{c',b_2\}$ to $W$.
\item if $((y_iz_i,y_jz_j \to L),\emptyset,\emptyset,a_1 \mapsto y_iz_i, a_2 \mapsto y_jz_j) \in D(s)$, then we terminate $\{\{y_iz_i\},\{y_jz_j\}\}$ and add $\{c(\{y_iz_i\}),b_1\}$ and $\{c(\{y_jz_j\}),b_2\}$ to $W$.
\end{itemize}
\item In all other cases, we report that $(G,T,W)$ is a no-instance.
\end{enumerate}
\end{reduction*}

Finally, we are ready to prove the following.
\begin{lemma} \label{lem:treecut-nonleaf}
There is an algorithm which takes an instance $J=(G,T,W)$ of \comvdpaths, a nice treecut decomposition
$(\T,\X)$ of $G$ of width $k$ and a non-leaf node $t \in V(\T)$, and computes $D(t)$ in time $k^{\Oh(k)} \cdot n^2$, assuming that for each child $t'$ of $t$ the set $D(t')$ is already computed.
\end{lemma}
\begin{proof}
First, we loop over all possible $R \in R(t)$; recall that $|R(t)| \leq k^{\Oh(k)}$. For a fixed $R$, we compute the corresponding instance $J_R=(G_R,T_R,W_R)$ of \comvdpaths. Recall that $V_R$ is a subset of vertices of $G_R$, which was added to $G_t$ during the construction of $J_R$ and each vertex in $V_R$ has degree at most two.

For the computed instance $J_R$, we apply the above reduction rule for each of the thin children of $t \in V(\T)$.
Note that each vertex $v$ added to $G_R$ during any of these reductions for thin children (denote the set of these vertices by $B_B$) is of degree at most two. Moreover, since the treecut decomposition of our graph is nice, $N(v) \subseteq X_t \cup V_R$. 

Then we loop over all possible functions $\mu$ that map each element $t'$ of $A_t$ to some element of $D(t')$. 
By \cref{thm:bold-children}, there are at most $2k+1$ elements of $A_t$, so there are at most $(k^{\Oh(k)})^{2k + 1} = k^{\Oh(k^2)}$ such functions. 
For the current $R$ and $\mu$, we perform for each $t' \in A_t$ the simplification according to~$\mu(t')$. 
Denote by $\widetilde{J}=(\widetilde{G},\widetilde{T},\widetilde{W})$ the instance obtained from $J_R$ after applying this sequence of simplifications (i.e., $\widetilde{G}$ is induced by vertices from $X_t \cup V_R \cup B_B$ and all $V_Q$'s coming from the simplifications).

We claim that $\widetilde{J}$ is an instance of \simplecomvdpaths. 
Indeed, we observed that after the simplification process, the degree of a vertex which was not added during the procedure is at most its degree in the original graph. So the only vertices which might have degree bigger than 2 are the ones which belong to~$X_t$. We thus can compute whether $\widetilde{J}$ is a yes-instance in $k^{\Oh(k)} + O(n^2)$ time using Lemma~\ref{lem:simplecvdp}.

We claim that $\widetilde{J}$ is a yes-instance if and only if $R$ is valid.
To see that, first assume that $R \in R(t)$ is a valid record.
Thus $J_R$ admits a solution.
By \cref{lem:valid-rec}, this means that there exists a valid record $R_{t'} \in D(t')$ for each child $t'$ of $t$. 
The reduction rule for thin children is safe and thus the solution is preserved.
For the bold children, consider an iteration in which $\mu$ assigns $R_{t'}$ to $t'$, for each $t' \in A_t$. 
The second statement of \cref{lem:simpl-safe} implies that the sequence of simplifications (starting from the instance $J_R$), consecutively in each node $t' \in A_t$, preserves the existence of a solution, so $\widetilde{J}$ is a yes-instance of \simplecomvdpaths.

Conversely, assume that the algorithm adds $R$ to $D(t)$, so the instance $\widetilde{J}$ of \simplecomvdpaths is a yes-instance. Then by the first statement of \cref{lem:simpl-safe}, all the instances obtained in the sequence of simplifications are also yes-instances. Since the reduction rule for thin children is safe, the corresponding instance $J_R$ of $R$ must be also a yes-instance.
Therefore, our algorithm computes correctly the set of valid records for $t$.

The procedure takes time $k^{\Oh(k^2)}$ times the time needed to solve an instance of \simplecomvdpaths; together $k^{\Oh(k^2)} \cdot n^2$. 
\end{proof}

We conclude the section with the following theorem.
\begin{theorem}\label{thm:treecut-width}
There is an algorithm which takes as input $(G,T,W)$ of \comvdpaths and returns an answer in time $k^{\Oh(k^2)}\cdot n^3$.
\end{theorem}
\begin{proof}
  We start by computing a treecut decomposition of width $2k$.
  By \cref{thm:approx-decomp} this can be done in time $k^{\Oh(k^2)}\cdot n^2$.
  Then we use \cref{thm:nice-decomp} to obtain a nice treecut decomposition in time $\Oh(n^3)$.
  We compute the set of valid records for each leaf of $\T$ in time $k^{\Oh(k^2)}\cdot n^2$ (\cref{lem:treecut-leaf}) and then for each non-leaf, by leaf-to-root recursion, in time $k^{\Oh(k^2)} \cdot n^3$ (\cref{lem:treecut-nonleaf}).
  We return a positive answer if and only if $D(r)=\{(\emptyset,\emptyset,\emptyset,\emptyset)\}$, where $r$ is the root of $\T$.
\end{proof}

%% file: properly-tw.tex
Our main result on properly colored paths and cycles in edge-colored graphs of bounded treewidth is as follows:
\begin{theorem}\label{thm:properly-tw}
Given an undirected graph $G$ with an edge coloring $\Ecol : E(G) \to [\ell]$
and a tree decomposition $(\T,\beta)$ of $G$ of width less than $k$,
one can verify if $G$ admits a properly colored Hamiltonian Cycle
in deterministic time $2^{\Oh(k)} \cdot \Oh(|V(G)|+|V(\T)|+\ell)$.
\end{theorem}
The main highlight of Theorem~\ref{thm:properly-tw} is the lack of the dependency on $\ell$
in the exponential part of the running time bound.
For sake of simplicity, we do not analyze in detail the base of the exponent
in the running time bound of the algorithm of Theorem~\ref{thm:properly-tw}.

The structure of the algorithm of Theorem~\ref{thm:properly-tw} follows the
outline of the typical \emph{rank-based} algorithms for connectivity problems
on graphs of bounded treewidth~\cite{BodlaenderCKN15}.
We refer to~\cite[Chapter~11]{cygan2015parameterized} for an exposition of the rank-based approach in the case of \textsc{Steiner Tree} and to~\cite{ZiobroP19} for a different short exposition of the presented approach for
\textsc{Hamiltonian Cycle}.

We start with describing a naive approach and then show how to reduce its complexity.

A \emph{separation} of a graph $G$ is a pair $(A,B)$ of subgraphs of $G$
such that each edge of $G$ belongs to exactly one of the subgraphs $A$ or $B$.
The \emph{order} of the separation $(A,B)$ is $|V(A) \cap V(B)|$. 
A \emph{partial solution} for a separation $(A,B)$ is a subgraph $P$ of $A$
that is a family of vertex-disjoint paths with both endpoints in $V(A) \cap V(B)$ and 
such that
every vertex of $V(A) \setminus V(B)$ is on one of the paths.
Clearly, if $C$ is a Hamiltonian cycle in $G$ and $(A,B)$ is a separation in $G$
with $V(B) \setminus V(A) \neq \emptyset$,
then $C \cap A := (V(C) \cap V(A), E(C) \cap E(A))$ is a partial solution for $(A,B)$. 
If $P$ is a partial solution for $(A,B)$ and $Q$ is a partial solution for $(B,A)$,
then we say that \emph{$P$ and $Q$ fit each other} if $P \cup Q$ is a Hamiltonian cycle in $G$.

The \emph{trace} of a partial solution $P$ for $(A,B)$ is a pair $(f_P, M_P)$ where
\begin{itemize}
\item $f_P : V(A) \cap V(B) \to \{0,1,2\}$ and $f_P(v)$ is the degree of $v$ in $P$;
\item $M_P$ is a matching on the vertex set $f_P^{-1}(1)$, matching endpoints
of the paths of $P$.
\end{itemize}
Note that the set of possible traces for $(A,B)$ is the same as the set of possible traces for $(B,A)$. Two traces $(f_P,M_P)$ and $(f_Q,M_Q)$ \emph{fit each other} if
\begin{itemize}
\item $f_P(v) + f_Q(v) = 2$ for every $v \in V(A) \cap V(B)$; and
\item $M_P \uplus M_Q$ is a single cycle on vertex set $f_P^{-1}(1) = f_Q^{-1}(1)$.
\end{itemize}
The following observation is straightforward.
\begin{lemma}\label{lem:fit1}
If $P$ is a partial solution for $(A,B)$ and $Q$ is the partial solution for $(B,A)$,
   then $P$ fits $Q$ if and only if the trace of $P$ fits the trace of $Q$.
\end{lemma}
Lemma~\ref{lem:fit1} is the base of the naive algorithm for Hamiltonian cycle on graphs
of bounded treewidth.
Given a tree decomposition $(\T,\beta)$, where $\T$ is a rooted tree, for a node $t \in V(\T)$
we use the following notation:
\begin{itemize}
\item $V_t$ is the union of $\beta(s)$ over all descendants $s$ of $t$ (including $t$),
\item $\bar{V}_t := (V(G) \setminus V_t) \cup \beta(t)$,
\item $G_t$ is the subgraph $G[V_t] \setminus E(G[\beta(t)])$,
\item $\bar{G}_t$ is the subgraph $G[\bar{V}_t]$.
\end{itemize}
Note that $(G_t,\bar{G}_t)$ is a separation. 
The algorithm, in the bottom-up fashion, computes for every $t \in V(\T)$ the family
$\mathcal{A}(t)$ of all traces $(f,M)$ for which there exists a partial solution
for $(G_t,\bar{G}_t)$ with trace $(f,M)$. 
If the width of the decomposition is less than $k$, then the number of possible traces
is $2^{\Oh(k \log k)}$, yielding $2^{\Oh(k \log k)} \cdot (|V(G)|+|V(\T)|)$ running time bound of the algorithm.

This naive algorithm can be easily adjusted to accommodate edge colors. Assume
that we are given an edge coloring $\Ecol : E(G) \to [\ell]$ and look for a properly 
colored Hamiltonian cycle. 
In the definition of a partial solution, we require all paths of $P$ to be properly colored.
Furthermore, partial solutions $P$ for $(A,B)$
and $Q$ for $(B,A)$ \emph{fit each other} if their union is a properly colored Hamiltonian cycle,
that is, for every vertex $v \in V(A) \cap V(B)$ that is an endpoint of both a path of $P$
and a path of $Q$, the edge of $P$ incident with $v$ and the edge of $Q$ incident with $v$
are of different colors. 
To accommodate this, the \emph{colored trace} for a partial solution $P$ is a triple $(f_P,M_P,\Tcol_P)$ where $f_P$ and $M_P$ have the meaning as before and $\Tcol_P : f_P^{-1}(1) \to [\ell]$
assigns to every vertex $v$ of degree $1$ in $P$ the color of the unique edge of $P$ incident
with $v$. 
Then, $(f_P, M_P, \Tcol_P)$ and $(f_Q,M_Q,\Tcol_Q)$ fit each other 
if and only if $(f_P,M_P)$ and $(f_Q,M_Q)$ agree as before and also $\Tcol_P(v) \neq \Tcol_Q(v)$
for every $v \in f_P^{-1}(1) = f_Q^{-1}(1)$. 
Again, we have a straightforward analog of Lemma~\ref{lem:fit2}:
\begin{lemma}\label{lem:fit2}
If $P$ is a partial solution for $(A,B)$ and $Q$ is the partial solution for $(B,A)$
in an edge-colored graph $G$,
   then $P$ fits $Q$ if and only if the colored trace of $P$ fits the colored trace of $Q$.
\end{lemma}
Note that there are at most $2^{\Oh(k \log k)} \ell^k$ possible colored traces for a separation
of order at most $k$. 
By following the standard dynamic programming algorithm for Hamiltonian cycle on
graphs of bounded treewidth (see e.g.~\cite{ZiobroP19}), we obtain
\begin{theorem}\label{thm:properly-tw-trivial}
Given an undirected graph $G$ with an edge coloring $\Ecol : E(G) \to [\ell]$
and a tree decomposition $(\T,\beta)$ of $G$ of width less than $k$,
one can verify if $G$ admits a properly colored Hamiltonian Cycle
in deterministic time $2^{\Oh(k(\log k + \log \ell))} \cdot \Oh(|V(G)|+|V(\T)|+\ell)$.
\end{theorem}

We now introduce the rank-based approach.
Fix a separation $(A,B)$ and fix a function $f: V(A) \cap V(B) \to \{0,1,2\}$. 
Let $Z := f^{-1}(1)$; for every trace $(f,M)$, $M$ is a matching on vertex set $Z$.
A \emph{cut} of $Z$ is an unordered pair $\{Z_1,Z_2\}$ such that $Z_1$ and $Z_2$ form a partition
of $Z$; a cut \emph{agrees} with a matching $M$ on $Z$ if no edge of $M$ connects the two sides
of a cut. Note that there are $2^{|Z|-1}$ cuts.

For two matchings $M_1$ and $M_2$ on $Z$, let $\mathbf{M}[M_1,M_2]$ be $1$ if
$M_1 \uplus M_2$ is a single cycle and $0$ otherwise.
For a matching $M$ and a cut $C$, let $\mathbf{C}[M,C]$ be $1$ if $M$ agrees with $C$
and $0$ otherwise. The crux of the rank-based approach lies in the following identity
(see e.g.\ Lemma~11.9 of~\cite{cygan2015parameterized}):
\begin{lemma}\label{lem:rank-based}
If $\mathbf{M}$ and $\mathbf{C}$ are treated as matrices over $\mathbb{F}_2$, then
$$\mathbf{M} = \mathbf{C} \cdot \mathbf{C}^T.$$
\end{lemma}
Let $\mathcal{A}$ be a family of traces for $(A,B)$.
We say that $\mathcal{A}' \subseteq \mathcal{A}$ \emph{represents} $\mathcal{A}$
if for every trace $\tau$, if there exists a trace $\tau_1 \in \mathcal{A}$ fitting $\tau$,
then there exists a trace $\tau_2 \in \mathcal{A}'$ fitting $\tau$.

Assume that all elements of $\mathcal{A}$ have $f$ as the first coordinate.
Let $\mathbf{M}[\mathcal{A},\cdot]$ be the submatrix of $\mathbf{M}$ induced by the rows
of the matchings $\{M~|~(f,M) \in \mathcal{A}\}$. 
Then, if $\mathcal{A}' \subseteq \mathcal{A}$ is such that
$\mathbf{M}[\mathcal{A}',\cdot]$ spans the same subspace as $\mathbf{M}[\mathcal{A},\cdot]$,
then $\mathcal{A}'$ represents $\mathcal{A}$. 
By the factorization of Lemma~\ref{lem:rank-based}, we infer that:
\begin{lemma}\label{lem:rank-based2}
Assume that $\mathcal{A}$ is a family of traces for $(A,B)$ with the same first coordinate
$f$. If $\mathcal{A}' \subseteq \mathcal{A}$ is such that 
$\mathbf{C}[\mathcal{A}',\cdot]$ spans the same subspace as $\mathbf{C}[\mathcal{A},\cdot]$,
then $\mathcal{A}'$ represents $\mathcal{A}$. 
\end{lemma}
Consequently, we can improve the naive algorithm for the \textsc{Hamiltonian Cycle} problem
as follows.
At node $t \in V(\T)$, replace $\mathcal{A}(t)$ with a subset $\mathcal{A}'(t)$
representing $\mathcal{A}(t)$ as folows: for every possible $f : \beta(t) \to \{0,1,2\}$,
restrict $\mathcal{A}(t)$ to $\mathcal{A}(t,f)$ consisting of traces with the first coordinate $f$, 
compute $\mathcal{A}'(t,f) \subseteq \mathcal{A}(t,f)$ representing $\mathcal{A}(t,f)$
using Lemma~\ref{lem:rank-based2} and a Gaussian elimination
on $\mathbf{C}[\mathcal{A}(t,f),\cdot]$ over $\mathbb{F}_2$,
and declare $\mathcal{A}'(t) = \bigcup_{f} \mathcal{A}'(t,f)$ to be a set
representing $\mathcal{A}(t)$. 
Note that $\mathcal{A}'(t,f)$ is of size at most $2^{|f^{-1}(1)|-1}$ as the number
of columns of $\mathbf{C}$ is $2^{|f^{-1}(1)|-1}$; hence $\mathcal{A}'(t)$ is
of size at most
$$\sum_{k=0}^{|\beta(t)|} \binom{|\beta(t)|}{k} 2^{k-1} 2^{|\beta(t)|-k} \leq 4^{|\beta(t)|} \leq 4^k.$$

By now-standard methods (see, e.g., the exposition of~\cite{ZiobroP19}), one can compute
in a bottom-up fashion representatives $\mathcal{A}'(t)$ for $t \in V(\T)$;
the computation at note $t$ takes into account the representatives
at children of $t$ and takes time $2^{\Oh(k)}$ per child. 

By following the same outline,
to prove Theorem~\ref{thm:properly-tw}, it suffices to show the following:
\begin{lemma}\label{lem:properly-tw}
Let $G$ be a graph with edge coloring $\Ecol : E(G) \to [\ell]$,
$(A,B)$ be a separation of order $k$,
and $\mathcal{A}$ be a family of colored traces for $(A,B)$.
Then, there exists a polynomial-time algorithm
that, given $\mathcal{A}$ and integers $k$ and $\ell$, 
finds a subset $\mathcal{A}' \subseteq \mathcal{A}$
that represents $\mathcal{A}$ and is of size at most $6^k$.
\end{lemma}
By partitioning $\mathcal{A}$ according to the first coordinate, it suffices to prove the following:
\begin{lemma}\label{lem:properly-tw2}
Let $G$ be a graph with edge coloring $\Ecol : E(G) \to [\ell]$,
$(A,B)$ be a separation of order $k$,
$f : V(A) \cap V(B) \to \{0,1,2\}$,
and $\mathcal{A}$ be a family of colored traces for $(A,B)$ with the first coordinate $f$.
Then, there exists an algorithm
that, given $\mathcal{A}$, $f$, and integers $k$ and $\ell$, 
  in time polynomial in the input size and $2^k$,
finds a subset $\mathcal{A}' \subseteq \mathcal{A}$
that represents $\mathcal{A}$ and is of size at most $4^k$.
\end{lemma}

The rest of this section is devoted to the proof of Lemma~\ref{lem:properly-tw2}.

Let $a$ be such that $2^a > \ell$ and let $\mathbb{F} = \mathbb{F}_{2^a}$ be the field
of characteristic $2$ with $2^a$ elements. Operations on $\mathbb{F}$ can be done in time
polynomial in $a = \Oh(\log \ell)$. Furthermore, $\mathbb{F}_2 \subseteq \mathbb{F}$.
We replace the range of colors $[\ell]$ with non-zero elements of $\mathbb{F}$:
for every color $i \in [\ell]$ we pick a distinct non-zero element $a_i \in \mathbb{F}$.
Henceforth, we assume that $\Ecol$ and all functions $\Tcol_P$ for $(f,M_P,\Tcol_P) \in \mathcal{A}$
have range $\{a_i~|~i \in [\ell]\} \subseteq \mathbb{F} \setminus \{0\}$.

Let $Z = f^{-1}(1)$. 
For two functions $\Tcol_P,\Tcol_Q : Z \to \mathbb{F}$, 
define
$$\pi(\Tcol_P, \Tcol_Q) = \prod_{v \in Z}(\Tcol_P(v) - \Tcol_Q(v)).$$
Note that $\pi$ can be treated as a $2|Z|$-variate multilinear polynomial of degree $|Z|$
with variables $(\Tcol_P(v))_{v \in Z}$ and $(\Tcol_Q(v))_{v \in Z}$.
Furthermore, we have that
\begin{equation}\label{eq:ptw:1}
\pi(\Tcol_P,\Tcol_Q) = 0 \quad \Leftrightarrow \quad \exists_{v \in Z} \Tcol_P(v) = \Tcol_Q(v).
\end{equation}
For a function $\Tcol_P : Z \to \mathbb{F}$, define a $|Z|$-variate multilinear
polynomial
$$\pi_{\Tcol_P}((x_v)_{v \in Z}) = \pi(\Tcol_P, \{v \mapsto x_v~|~v \in Z\}).$$
Consider a matrix $\mathbf{D}$ with rows indexed by possible
colored traces $(f,M_P,\Tcol_P)$ for $(A,B)$ and columns by all $2^{|Z|}$ multilinear monomials 
on variables $(x_v)_{v \in Z}$. 
The row $\mathbf{D}[(f,M_P,\Tcol_P), \cdot]$ contains the coefficients of $\pi_{\Tcol_P}$. 

Let $\mathbf{C}'$ be a matrix with rows indexed by possible colored traces
$(f,M_P,\Tcol_P)$ for $(A,B)$ and columns by all cuts of $Z$.
The row $\mathbf{C}'[(f,M_P,\Tcol_P),\cdot]$ equals $\mathbf{C}[(f,M_P),\cdot]$,
where every element is treated as an element of $\mathbb{F}$.
Finally, let $\mathbf{E}$ be a matrix with rows indexed by possible
colored traces $\tau$ for $(A,B)$
and $\mathbf{E}[\tau,\cdot]$ is the tensor product
of $\mathbf{C}'[\tau,\cdot]$ and $\mathbf{D}[\tau,\cdot]$.
Note that $\mathbf{E}$ has $2^{|Z|-1} \cdot 2^{|Z|} = 2^{2|Z|-1}$ columns.

Lemma~\ref{lem:properly-tw2} follows from the following lemma by applying Gaussian elimination
on the rows of $\mathbf{E}$ corresponding to the elements of $\mathcal{A}$.
\begin{lemma}\label{lem:properly-tw3}
Let $\mathcal{A}' \subseteq \mathcal{A}$ be such that
$\mathbf{E}[\mathcal{A}',\cdot]$ spans the same subspace as $\mathbf{E}[\mathcal{A},\cdot]$.
Then, $\mathcal{A}'$ represents $\mathcal{A}$.
\end{lemma}
\begin{proof}
Let $\tau_Q = (f_Q,M_Q,\Tcol_Q)$ be a colored trace for $(B,A)$
and let $\tau_P = (f,M_P,\Tcol_P) \in \mathcal{A}$ be a trace fitting $\tau_Q$. 
Note that $f_Q^{-1}(1) = Z$.

Let $v_1$ be a vector over $\mathbb{F}$ with elements indexed by all cuts of $Z$
with value $1$ if the corresponding cut agrees with $M_Q$ and $0$ otherwise. 
By Lemma~\ref{lem:rank-based}, $\mathbf{C}[(f,M),\cdot] \cdot v_1 \neq 0$
if and only if the trace $(f,M)$ fits $(f_Q,M_Q)$. 

Let $v_2$ be a vector over $\mathbb{F}$ with elements indexed by all $2^{|Z|}$ multilinear monomials
over variables $(x_v)_{v \in Z}$;
the value of $v_2$ at monomial $\prod_{v \in I} x_v$ for $I \subseteq Z$ equals
$\prod_{v \in I} \Tcol_Q(v)$. 
For a colored trace $(f,M_R,\Tcol_R)$, by~\eqref{eq:ptw:1}, we have
that $\Tcol_R(v) \neq \Tcol_Q(v)$ for every $v \in Z$ if and only if
$\mathbf{D}[(f,M_R,\Tcol_R),\cdot] \cdot v_2 \neq 0$. 

Consequently, for a colored trace $\tau = (f,M_R,\Tcol_R)$, we have that $\tau$ fits $\tau_Q$
if and only if $\mathbf{C}'[\tau,\cdot] \cdot v_1 \neq 0$ and $\mathbf{D}[\tau,\cdot] \cdot v_2 \neq 0$. 
The latter is equivalent to $\mathbf{E}[\tau,\cdot] \cdot (v_1 \otimes v_2) \neq 0$, 
where $v_1 \otimes v_2$ is the tensor product of $v_1$ and $v_2$. 

Since $\tau_P$ fits $\tau_Q$, 
      $\mathbf{E}[\tau_P,\cdot] \cdot (v_1 \otimes v_2) \neq 0$.
Since $\mathbf{E}[\mathcal{A}',\cdot]$ spans the same subspace as $\mathbf{E}[\mathcal{A},\cdot]$,
there exist elements $\tau_1,\tau_2,\ldots,\tau_r \in \mathcal{A}'$ and coefficients $\lambda_1,\ldots,\lambda_r$, such that
$$\mathbf{E}[\tau_P,\cdot] = \sum_{i=1}^r \lambda_i \mathbf{E}[\tau_i,\cdot].$$
Since $\mathbf{E}[\tau_P,\cdot] \cdot (v_1 \otimes v_2) \neq 0$, there exists $1 \leq i \leq r$ such that
$\lambda_i \neq 0$ and $\mathbf{E}[\tau_P,\cdot] \cdot (v_1 \otimes v_2) \neq 0$. 
Hence, $\tau_i$ fits $\tau_Q$ and we have $\tau_i \in \mathcal{A}'$. This finishes the proof of the lemma.
\end{proof}
Recall that $\mathbf{E}$ has $2^{2|Z|-1} \leq 4^{|Z|}$ columns.
Hence, with Gaussian elimination one can find $\mathcal{A}'$ as in Lemma~\ref{lem:properly-tw3} of size at most $4^{|Z|}$.
This finishes the proof of Theorem~\ref{thm:properly-tw}.

%% file: 2DSPP.tex
\begin{quote}
\textsc{Directed Two Disjoint Shortest Paths Problem (2-DSPP) with transition restrictions}\\
\textbf{Input:} A directed graph $G=(V,E)$ with transition system $T$, a length function $w: E\rightarrow\mathbb R_{\ge 0}$ and two pairs of vertices ($s_1,t_1$),($s_2,t_2$) in $G$.\\
\textbf{Target:} Find two disjoint (vertex-disjoint or edge-disjoint) paths $P_1$ and $P_2$ in $G$ such that for both $i=1,2$, path $P_i$ is a shortest path (even in the graph $G$ with no transition restrictions) from $s_i$ to $t_i$ and $P_i$ is also $T$-compatible.
\end{quote}

Our algorithm is an adaption of the algorithm for \textsc{Directed Two Disjoint Shortest Paths Problem} (assuming that every dicycle in $G$ has positive length) of B{\'{e}}rczi and Kobayashi~\cite{DBLP:conf/esa/Berczi017}.
Roughly speaking, we show that transition restrictions are not a barrier for using the same strategy.
Note that in this section we consider directed graphs with parallel edges.
Transitions and transition systems for directed graphs are defined in the natural way analogous to the undirected case.

We follow the notations from the paper of B{\'{e}}rczi and Kobayashi~\cite{DBLP:conf/esa/Berczi017} for convenience. We define $E_i$ to be the set of edges that appear in some shortest path (without transition restrictions) from $s_i$ to $t_i$ for $i=1,2$. By this definition, an $s_i$-$t_i$ path is a shortest $T$-compatible $s_i$-$t_i$ path if and only if it consists of edges of $E_i$ and is also $T$-compatible for $i=1,2$. Thus the edge-disjoint (vertex-disjoint) \textsc{2-DSPP with transition restrictions} is equivalent to finding two edge-disjoint (vertex-disjoint) $T$-compatible paths $P_1$ and $P_2$ such that $P_i$ is from $s_i$ to $t_i$, $E(P_i)\subseteq E_i$ and $P_i$ satisfies the transition restrictions for $i=1,2$. Each set $E_i$ can be computed in polynomial time using the method from the paper of B{\'{e}}rczi and Kobayashi~\cite{DBLP:conf/esa/Berczi017}. First, we compute the distance $d_i(v)$ from $s_i$ to $v$ for $i=1,2$, using Dijkstra's algorithm. Let ${\cal E}_i=\{uv\mid d_i(v)-d_i(u)=w(uv)\}$. Then $E_i=\{uv\in{\cal E}_i\mid \text{there exists a path from }v \text{ to }t_i\text{ in }{\cal E}_i\}$.

\begin{theorem}
If the length of every directed cycle is positive, both edge-disjoint and vertex-disjoint variants of \textsc{2-DSPP with transition restrictions} can be solved in polynomial time.
\end{theorem}
\begin{corollary}
If the length of every edge is positive, both edge-disjoint and vertex-disjoint variants of \textsc{2-DSPP with transition restrictions} can be solved in polynomial time.
\end{corollary}

For a set $F$ of directed edges, let $\overline{F}$ be the set of edges obtained by reversing all edges of $F$, that is, $\overline{F}=\{vu\mid uv\in F\}$. For a directed edge $e=uv$, let $\overline{e}=vu$ denote the edge obtained by reversing $e$. 
For two paths $P$ and $Q$ with consecutive edges $e^p_1,e^p_2,\ldots,e^p_{|P|}$ and, respectively, $e^q_1,e^q_2,\ldots,e^q_{|Q|}$ such that $\text{head}(e^p_{|P|})=\text{tail}(e^q_1)$, by $P \cdot Q$ we denote the concatenation of paths $P$ and $Q$, i.e., $P \cdot Q = e^p_1,e^p_2,\ldots,e^p_{|P|},e^q_1,e^q_2,\ldots,e^q_{|Q|}$. 
Note that if $P$ and $Q$ are vertex-disjoint except for $\text{head}(e^p_{|P|})=\text{tail}(e^q_1)$, then $P \cdot Q$ is a path, too.

%% file: Edge2DSPP.tex
We show that the edge-disjoint case of \textsc{2-DSPP with transition restrictions} can be solved in polynomial time.  We use the method of B{\'{e}}rczi and Kobayashi~\cite{DBLP:conf/esa/Berczi017}, which reduces the problem of \textsc{Edge Disjoint 2-DSPP} to finding a path in a graph $\cal G$ constructed from the input graph $G$. Based on that, we just need to delete edges of $\cal G$ which correspond to forbidden transitions of $G$ with respect to $T$ and it suffices to find the path in the remaining subgraph of $\cal G$.

We repeat the procedure of B{\'{e}}rczi and Kobayashi~\cite{DBLP:conf/esa/Berczi017} briefly here for consistency. Let $G$ be a graph (without transition systems $T$) such that the length of every dicycle in $G$ is positive. First, we compute $E_i$ for $i=1,2$. Then we create four new vertices $s_1',s_2',t_1',t_2'$, create four edges $s_1's_1$, $s_2's_2$, $t_1t_1'$, $t_2t_2'$ of length $0$ respectively, and add $s_i's_i$, $t_it_i'$ to $E_i$ for $i=1,2$. Let $E_0=E_1\cap E_2, E_1^*=E_1\setminus E_0, E_2^*=E_2\setminus E_0$.  We remove all edges of $E(G)\setminus (E_1\cup E_2)$, contract all edges of $E_0$ and reverse all edges of $E_2^*$. Finally we get a new graph $G^*=(V^*,E^*=E_1^*\cup \overline{E_2^*})$. Let $V_0\subseteq V$ be the set of vertices that are newly created after contracting $E_0$. For $v\in V_0$, we use $G_v$ to denote the subgraph of $G-(E(G)\setminus(E_1\cup E_2))$ induced by the vertices corresponding to $v$ before contracting. For an edge $e\in E^*$, let $f(e)\in E(G)$ be the edge corresponding to $e$ before the contracting and reversing operations.

The following two lemmas show that $G_v$ is acyclic for every $v\in V_0$ and $G^*$ is acyclic.
\begin{lemma}\cite{DBLP:conf/esa/Berczi017}\label{acyclic-V}
The edge set $E_i$ forms no dicycle in $G$ for $i=1,2$.
\end{lemma}

\begin{lemma}\cite{DBLP:conf/esa/Berczi017}\label{contract-V}
In the graph $G$, suppose that $C$ is a dicycle in $E_1\cup \overline{E_2}$. Then $E_1\cap E(C)\subseteq E_2$ and $E_2\cap \overline{E(C)}\subseteq E_1$.
\end{lemma}

Then we define a new digraph $\cal G$ whose vertex set is $W=E_1^*\times \overline{E_2^*}$. There is a directed edge from $(e_1,e_2)$ to $(e_1',e_2')$ if one of three cases holds.
\begin{enumerate}[label=(\roman*)]
\item $e_1=e_1'$ and $\text{head}_{G^*}(e_2)=\text{tail}_{G^*}(e_2')=v$. There is no path from $\text{head}_{G^*}(e_1)$ to $v$ in $G^*$. Moreover, if $v\in V_0$, then $G_v$ contains a path from $\text{tail}_{G}(e_2')$ to $\text{head}_{G}(e_2)$.

\item $e_2=e_2'$ and $\text{head}_{G^*}(e_1)=\text{tail}_{G^*}(e_1')=v$. There is no path from $\text{head}_{G^*}(e_2)$ to $v$ in $G^*$. Moreover, if $v\in V_0$, then $G_v$ contains a path from $\text{head}_{G}(e_1)$ to $\text{tail}_{G}(e_1')$.

\item $\text{head}_{G^*}(e_2)=\text{tail}_{G^*}(e_2')=\text{head}_{G^*}(e_1)=\text{tail}_{G^*}(e_1')=v$. If $v\in V_0$, then $G_v$ contains two edge-disjoint paths from $\text{head}_{G}(e_1)$ to $\text{tail}_{G}(e_1')$ and from $\text{tail}_{G}(e_2')$ to $\text{head}_{G}(e_2)$ respectively.
\end{enumerate}

Finally the following lemma reduces the edge-disjoint version of \textsc{2-DSPP} to finding a path in $\cal G$ from $(s_1's_1,t_2't_2)$ to $(t_1t_1',s_2s_2')$. Note that $s_i,t_i\in V(G)$ might be the endpoints of edges of $E_0$ for $i=1,2$. In this case, although we might contract the edges incident to $s_i,t_i\in V(G)$ and replace these vertices with new vertices. Therefore, we slightly abuse the notation and use $s_i$ and $t_i$ to denote the vertex adjacent to $s_i'$ and $t_i'$ respectively in $G^*$ for $i=1,2$, for the sake of simplicity.

\begin{lemma}\cite{DBLP:conf/esa/Berczi017}\label{mainE2DSPP}
There is a directed path in $\cal G$ from $(s_1's_1,t_2't_2)$ to $(t_1t_1',s_2s_2')$ if and only if $G$ has two edge-disjoint paths $P_1$ and $P_2$ such that $P_i$ is from $s_i'$ to $t_i'$ and $P_i\subseteq E_i$ for $i=1,2$.
\end{lemma}

To solve the edge-disjoint version of \textsc{Directed Two Disjoint Shortest Paths Problem (2-DSPP) with transition restrictions}, we will show that it suffices to delete the edges in $\cal G$ which correspond to forbidden transitions of $G$ and find a path in the remaining graph of $\cal G$ from $(s_1's_1,t_2't_2)$ to $(t_1t_1',s_2s_2')$.
For every edge in $\cal G$, we check whether it corresponds to forbidden transitions according to the following three cases and delete the edge if it corresponds to forbidden transitions.
Suppose the edge is from some vertex $(e_1,e_2)\in W$ to another vertex $(e_1',e_2')\in W$.
\begin{itemize}
\item[$\bullet$] The edge is of type (i), i.e., $e_1=e_1'$ and $\text{head}_{G^*}(e_2)=\text{tail}_{G^*}(e_2')=v$.  If $v\in V_0$, let $G^s$ be the subgraph of $G$ consisting of all edges of $G_v$ together with $f(e_2)$ and $f(e_2')$. In this case, if there is no $T$-compatible paths in $G^s$ from $\text{tail}_G(f(e_2'))$ to $\text{head}_G(f(e_2))$, then remove the edge from $\cal G$. If $v\notin V_0$ and $\{\overline{e_2'}, \overline{e_2}\} \notin T_G(v)$, then remove the edge from $\cal G$.

\item[$\bullet$] The edge is of type (ii), i.e., $e_2=e_2'$ and $\text{head}_{G^*}(e_1)=\text{tail}_{G^*}(e_1')=v$. If $v\in V_0$, let $G^s$ be the subgraph of $G$ consisting of all edges of $G_v$ together with $f(e_1)$ and $f(e_1')$. In this case, if $v\in V_0$ and there is no $T$-compatible path in $G^s$ from $\text{tail}_G(f(e_1))$ to $\text{head}_G(f(e_1'))$, then remove the edge from $\cal G$. If $v\notin V_0$ and $\{e_1, e_1'\} \notin T_G(v)$, then remove the edge from $\cal G$.
\item[$\bullet$] The edge is of type (iii), i.e., $\text{head}_{G^*}(e_2)=\text{tail}_{G^*}(e_2')=\text{head}_{G^*}(e_1)=\text{tail}_{G^*}(e_1')=v$. If $v\in V_0$, let $G^s$ be the subgraph of $G$ consisting of all edges of $G_v$ together with $f(e_1), f(e_1'), f(e_2)$ and $f(e_2')$. In this case, if $G^s$ does not contain two $T$-compatible edge-disjoint paths such that one path is from $\text{tail}_G(f(e_1))$ to $\text{head}_G(f(e_1'))$ and the other path is from $\text{tail}_G(f(e_2'))$ to $\text{head}_G(f(e_2))$, then remove the edge from $\cal G$. If $v\notin V_0$ and $\{e_1, e_1'\} \notin T_G(v)$ or if $v\notin V_0$ and $\{\overline{e_2'}, \overline{e_2}\} \notin T_G(v)$, then remove the edge from $\cal G$.
\end{itemize}

We need to check whether there exists a $T$-compatible path between two given vertices in a (direced) forbidden-transition graph.
Szeider shows a dichotomy of NP-complete and linear-time solvable for the problem of finding a $T$-compatible path between two given vertices of an (undirected) graph~\cite{Szeider}.
In contrast, the following lemma shows that in a directed acyclic graph, we can find a $T$-compatible path between two given vertices in polynomial time.

\begin{lemma} \label{pathDAG}
In a directed acyclic graph $G=(V,E)$ with transition system $T_{G}$, we can compute if there is a directed $T$-compatible path $P$ from $s$ to $t$ for $s,t\in V(G)$ in polynomial time.
\end{lemma}
\begin{proof}
We construct a directed graph $\tilde{G}$ as follows. First create two vertices $s_{0},t_{0}$. Then for every edge $e\in E(G)$, create a vertex $v_e$. For any two edges $e,e'\in E(G)$, create an edge $v_ev_{e'}$ if $ee'\in E(T_{G}(v))$ for some $v\in V(G)$. Finally, create edges $s_0v_e$ for every $e\in E(G)$ such that $\text{tail}_G(e)=s$ and create edges $v_{e'}t_0$ for every $e'\in E(G)$  such that $\text{head}_G(e')=t$. We claim that we can find a directed path $P'$ from $s_0$ to $t_0$ in $\tilde{G}$ if and only if there is a directed $T$-compatible path $P$ from $s$ to $t$ in $G$. For the ``if'' direction, suppose that there is such a path $P= e_1,e_2,\ldots,e_\ell$ in $G$, where $e_1,\ldots,e_\ell$ are the consecutive edges of $P$. Then we can obviously get the path $P'=s_0v_{e_1},v_{e_1}v_{e_2},\ldots,v_{e_{\ell}}t_0$ by the definition of $\tilde{G}$. For the ``only if'' direction, suppose that there is a directed path $P'=s_0v_{e_{i_1}},v_{e_{i_1}}v_{e_{i_2}},\ldots,v_{e_{i_{\ell}}}t_0$ in $\tilde{G}$. Then $P=e_{i_1},e_{i_2},\ldots,e_{i_{\ell}}$ is a directed $T$-compatible walk from $s$ to $t$ in $G$. Since $G$ is acyclic, $P$ is also a path. This completes the proof of the claim. We can build the graph $\tilde{G}$ in $O(|E|^{2})$-time and find an $s_0t_0$ path in $\tilde{G}$ using DFS in $O(|E|^{2})$ time. Thus the lemma holds.
\end{proof}

For $v\in V_0$, by Lemma~\ref{acyclic-V}, there is no dicycle in $G_v$. Moreover, observe that we cannot have a vertex in $V(G) \setminus V(G_v)$ adjacent to more than one edge from $E(G^s) \setminus E(G_v)$, so $G^s$ is also acyclic. So we can decide whether or not to remove the edges of type (i) or (ii) from $\cal G$ in polynomial time according to lemma~\ref{pathDAG}. For the edges of type (iii), we need to compute if there are two edge-disjoint $T$-compatible paths in a directed acyclic graph. We show that it can be done in polynomial time and the algorithm is an adaption of the algorithm of finding two vertex-disjoint paths in DAG by Perl and Shiloach~\cite{DBLP:journals/jacm/PerlS78}.
\begin{lemma}
In a directed acyclic graph $G=(V,E)$ with transition system $T_{G}$, we can solve the edge-disjoint version of \textsc{2-DSPP with transition restrictions} in polynomial time.
\end{lemma}
\begin{proof}
  First we modify the graph $G$ as follows.
  We create four vertices $s_1',s_2',t_1',t_2'$ and update $V(G)$ as $V(G)\leftarrow V(G)\cup\{s_1',s_2',t_1',t_2'\}$.
  We create four edges $\{s_1's_1,s_2's_2,t_1t_1',t_2t_2'\}$ and update $E(G)$ as $E(G)\leftarrow E(G)\cup\{s_1's_1,s_2's_2,t_1t_1',t_2t_2'\}$.
  Also, for $i=1,2$, we update $T_G(s_i)$ as 
  \[T_G(s_i) \leftarrow T_G(s_i) \cup\{\{e, e'\}\mid e=s_i's_i\text{ and tail}_G(e')=s_i\},\]
  and we update $T_G(t_i)$ as 
  \[T_G(t_i) \leftarrow T_G(t_i) \cup\{\{e, e'\} \mid e'=t_it_i'\text{ and head}_G(e)=t_i\}.\]
  For every vertex $v\in V(G)$, define the level $\ell(v)$ as the length of a longest directed path in $G$ starting from~$v$.
  Since $G$ is acyclic, this can be computed by repeatedly removing a vertex of~$G$.
  Then we create a graph $\tilde{G}$ as follows.
  Let the vertex set of $\tilde{G}$ be $V(\tilde{G})=\{(e_1,e_2)\mid e_1,e_2\in E(G)\text{ and }e_1\neq e_2\}$.
  For every $(e_1,e_2),(e_1',e_2')\in V(\tilde{G})$, create an edge from $(e_1,e_2)$ to $(e_1',e_2')$ if one of the following cases holds:
\begin{enumerate}[label=(\arabic*)]
    \item $e_1=e_1'$, $\ell(\text{head}_G(e_2))\geq \ell(\text{head}_G(e_1))$, $\{e_2, e_2'\} \in T_G(\text{head}_G(e_2))$.
    \item $e_2=e_2'$, $\ell(\text{head}_G(e_1))\geq \ell(\text{head}_G(e_2))$, $\{e_1, e_1'\} \in T_G(\text{head}_G(e_1))$.
    \item $e_1=e_1'=t_1t_1'$, $\ell(\text{head}_G(e_2))<\ell(t_1')$, $\{e_2, e_2'\} \in T_G(\text{head}_G(e_2))$.
    \item $e_2=e_2'=t_2t_2'$, $\ell(\text{head}_G(e_1))<\ell(t_2')$, $\{e_1, e_1'\} \in T_G(\text{head}_G(e_1))$.
\end{enumerate}
We claim that there are two $T$-compatible edge-disjoint paths $P_1$ and $P_2$ in $G$ such that $P_i$ is from $s_i'$ to $t_i'$ for $i=1,2$ if and only if there is a path $P$ from $(s_1's_1,s_2's_2)$ to $(t_1t_1',t_2t_2')$ in $\tilde{G}$.

\medskip

(``only if'' direction): Let $P_1=e_1^0,e_1^1,\ldots,e_1^{p+1}$ and $e_1^0=s_1's_1,e_1^{p+1}=t_1t_1'$. Let $P_2=e_2^0,e_2^1,\ldots,e_2^{q+1}$ and $e_2^0=s_2's_2,e_2^{q+1}=t_2t_2'$. For any $i\in \{0,1,\ldots,p+1\},j\in\{0,1,\ldots,q+1\}$ such that $(i,j) \neq (p+1,q+1)$, one of the following four cases must hold.
\begin{itemize}
    \item[$\bullet$] $i\leq p$ and $j\leq q$, $\ell(\text{head}_G(e_1^i))\leq \ell(\text{head}_G(e_2^j))$ and there is an edge in $\tilde{G}$ from $(e_1^i,e_2^j)$ to $(e_1^i,e_2^{j+1})$.
    \item[$\bullet$] $i\leq p$ and $j\leq q$, $\ell(\text{head}_G(e_1^i))\geq \ell(\text{head}_G(e_2^j))$ and there is an edge in $\tilde{G}$ from $(e_1^i,e_2^j)$ to $(e_1^{i+1},e_2^j)$.
    \item[$\bullet$] $i=p+1$ and $j\leq q$, $\ell(\text{head}_G(e_2^j))<\ell(t_1')$ and there is an edge in $\tilde{G}$ from $(e_1^{p+1},e_2^j)$ to $(e_1^{p+1},e_2^{j+1})$.
    \item[$\bullet$] $j=q+1$ and $i\leq p$, $\ell(\text{head}_G(e_1^i))<\ell(t_2')$ and there is an edge in $\tilde{G}$ from $(e_1^i,e_2^{q+1})$ to $(e_1^{i+1},e_2^{q+1})$.
\end{itemize}
As a result, there is a path $P$ from $(s_1's_1,s_2's_2)$ to $(t_1t_1',t_2t_2')$ in $\tilde{G}$. This finishes the proof for ``only if'' direction.

\medskip

(``if'' direction): Suppose that there exists a path $P$ from $(s_1's_1,s_2's_2)$ to $(t_1t_1',t_2t_2')$ in $\tilde{G}$. Let $P=(e_1^0,e_2^0), (e_1^1,e_2^1), \ldots, (e_1^r,e_2^r)$ such that $(s_1's_1,s_2's_2)=(e_1^0,e_2^0)$ and $(e_1^r,e_2^r)=(t_1t_1',t_2t_2')$. We construct two edge-disjoint $T$-compatible paths $P_1,P_2$ as follows. First we initialize $P_1=e_1^0,P_2=e_2^0$. 
Then for $i=0,\ldots,r-1$, we update $P_1$ and $P_2$ according to the following cases:
\begin{itemize}
    \item[$\bullet$] Suppose that the edge from $(e_1^i,e_2^i)$ to $(e_1^{i+1},e_2^{i+1})$ is of type (1). Then $P_2\leftarrow P_2\cdot e_2^{i+1}$.
    \item[$\bullet$] Suppose that the edge from $(e_1^i,e_2^i)$ to $(e_1^{i+1},e_2^{i+1})$ is of type (2). Then $P_1\leftarrow P_1\cdot e_1^{i+1}$.
    \item[$\bullet$] Suppose that the edge from $(e_1^i,e_2^i)$ to $(e_1^{i+1},e_2^{i+1})$ is of type (3). Then $P_2\leftarrow P_2\cdot e_2^{i+1}$.
    \item[$\bullet$] Suppose that the edge from $(e_1^i,e_2^i)$ to $(e_1^{i+1},e_2^{i+1})$ is of type (4). Then $P_1\leftarrow P_1\cdot e_1^{i+1}$.
\end{itemize}
By the definition of edges of $\tilde{G}$, we get that $P_1$ and $P_2$ are two $T$-compatible edge-disjoint paths in $G$ such that $P_i$ is from $s_i'$ to $t_i'$ for $i=1,2$. We can construct a graph $\tilde{G}$ in $O(|E|^{3})$ time and find a path from $(s_1's_1,s_2's_2)$ to $(t_1t_1',t_2t_2')$ in $O(|E|^{3})$ time. Thus the lemma holds.
\end{proof}

Thus we can also decide whether or not to remove an edge of type (iii) from $\cal G$ in polynomial time and let $\hat{\cal G}$ be the remaining subgraph of $\cal G$. The following lemma shows that we can reduce edge-disjoint version of \textsc{2-DSPP with transition restrictions} to finding a path from $(s_1's_1,t_2't_2)$ to $(t_1t_1',s_2s_2')$ in $\hat{\cal G}$.

\begin{lemma}\label{mainE2DSPPT}
There is a directed path in $\hat{\cal G}$ from $(s_1's_1,t_2't_2)$ to $(t_1t_1',s_2s_2')$ if and only if $G$ has two edge-disjoint $T$-compatible paths $P_1$ and $P_2$ such that $P_i$ is from $s_i'$ to $t_i'$ and $P_i\subseteq E_i$ for $i=1,2$.
\end{lemma}
\begin{proof}
(``if'' direction) Suppose that $G$ has two edge-disjoint $T$-compatible paths $P_1$ and $P_2$ such that $P_i$ is from $s_i'$ to $t_i'$ and $P_i\subseteq E_i$ for $i=1,2$. $E(P_1)\setminus E_0$ forms a directed path $P_1^*$ in $G^*$ from $s_1'$ to $t_1'$. $\overline{E(P_2)\setminus E_0}$ forms a directed path $P_2^*$ in $G^*$ from $t_2'$ to $s_2'$.
Let $P_1^*=e_1^0,e_1^1,...,e_1^{p+1}$ and $e_1^0=s_1's_1,e_1^{p+1}=t_1t_1'$. Let $P_2^*=e_2^0,e_2^1,...,e_2^{q+1}$ and $e_2^0=t_2't_2,e_2^{q+1}=s_2s_2'$.  It follows that $e_1^i\in E_1^*$ for $i=0,1,...,p+1$ and $e_2^j\in E_2^*$ for $j=0,1,...,q+1$. By the proof of Lemma~\ref{mainE2DSPP} (interested readers could refer to the proof of Lemma $8$ in~\cite{DBLP:conf/esa/Berczi017}), there is a directed path $P$ in $\cal G$ from $(s_1's_1,t_2't_2)$ to $(t_1t_1',s_2s_2')$ such that every edge of $P$ is of one of the three types: (i) from $(e_1^i,e_2^j)$ to $(e_1^i,e_2^{j+1})$ ($i\in \{0,...,p+1\},j\in \{0,...,q\}$); (ii) from $(e_1^i,e_2^j)$ to $(e_1^{i+1},e_2^j)$ ($i\in \{0,...,p\},j\in \{0,...,q+1\}$); (iii) from $(e_1^i,e_2^j)$ to $(e_1^{i+1},e_2^{j+1})$ (($i\in \{0,...,p\},j\in \{0,...,q\}$)). Since $P_1$ and $P_2$ are $T$-compatible, by the rules we construct $\hat{\cal G}$, we can see that all edges of $P$ in $\cal G$ remains in $\hat{\cal G}$. This completes the proof for ``if direction''.

\medskip

(``only if'' direction) Suppose that there is a directed path $P$ from $(e_1^0,e_2^0)=(s_1's_1,t_2't_2)$ to $(e_1^r,e_2^r)=(t_1t_1',s_2s_2')$ in $\hat{\cal G}$ that goes through $(e_1^0,e_2^0),(e_1^1,e_2^1),\ldots,(e_1^r,e_2^r)$ consecutively. Since $\hat{\cal G}$ is a subgraph of $\cal G$, by Lemma~\ref{mainE2DSPP}, there exists two edge-disjoint paths $P_1$ and $P_2$ in $G$ such that $P_i$ is from $s_i'$ to $t_i'$ and $P_i\subseteq E_i$ for $i=1,2$ in $G$. Moreover, again from the proof of \cref{mainE2DSPP}, it follows that $e_1^i\in E(P_1)$ and $\overline{e_2^i}\in E(P_2)$. By the rule we construct $\hat{\cal G}$, for an edge from $(e_1^i,e_2^i)$ to $(e_1^{i+1},e_2^{i+1})$ ($i\in \{0,...,r-1\}$), there is a $T$-compatible subpath of $P_1$ from $\text{tail}_G(f(e_1^i))$ to $\text{head}_G(f(e_1^{i+1}))$ if $e_1^i\neq e_1^{i+1}$ or there is a $T$-compatible subpath of $P_2$ from $\text{tail}_G(f(e_2^{i+1}))$ to $\text{head}_G(f(e_2^i))$ if $e_2^i\neq e_2^{i+1}$. It follows that $P_1$ and $P_2$ are also $T$-compatible. This finishes the proof for ``only if'' direction.
\end{proof}

Since $\hat{\cal G}$ is a subgraph of $\cal G$ and $\cal G$ contains at most $|E|^{2}$ vertices, we can detect a path in $\hat{\cal G}$ in polynomial time. Thus Lemma~\ref{mainE2DSPPT} shows that we can solve edge-disjoint version of \textsc{2-DSPP with transition restrictions} in polynomial time  assuming that every cycle in the input graph has positive length.

%% file: Vertex2DSPP.tex
When computing vertex-disjoint version of \textsc{2-DSPP} in the paper of B{\'{e}}rczi and Kobayashi~\cite{DBLP:conf/esa/Berczi017}, they create a new digraph $G_2$ as follows: for every vertex $v\in V$ create two vertices $v^{+}$ and $v^{-}$. Create an edge $v^{-}v^{+}$ with $w(v^{-}v^{+})=0$. Create an edge $u^{+}v^{-}$ if there is an edge $uv$ in $G$ and let $w(u^{+}v^{-})=w(uv)$. Thus \textsc{vertex-disjoint 2-DSPP} in $G$ is reduced to edge-disjoint variant of \textsc{2-DSPP} in $G_2$. However, this method does not work in the forbidden-transitions setting because part of the information of transitions will be lost after creating the new graph $G_2$.

In order to keep the information of transitions, we first modify $G$ as follows.
We compute the set $E_1$ and $E_2$ of $G$. Remove all edges of $E(G)\setminus (E_1\cup E_2)$ from $E(G)$ and all isolated vertices from $V(G)$. When removing the edges or vertices we update the transition system accordingly. Then create four new vertices $s_1',s_2',t_1',t_2'$ and four edges $s_1's_1$, $s_2's_2$, $t_1t_1'$, $t_2t_2'$ all with length $0$. Add $s_i's_i$ and $t_it_i'$ to $E_i$ for $i=1,2$. Thus a shortest path from $s_i$ to $t_i$ corresponds to a shortest path from $s_i'$ to $t_i'$ starting with the edge $s_i's_i$ and ending with the edge $t_it_i'$. We update $T_G(s_i)$ by adding $\{\{e, e'\}\mid e=s_i's_i\text{ and tail}_G(e')=s_i\}$ to it for $i=1,2$. Let $T_G(t_i)=\{\{e, e'\} \mid \text{head}_G(e)=t_i\text{ and }e'=t_it_i'\}$ for $i=1,2$.

Then we create a graph $G'$ as follows. For every vertex $v\in V(G)\setminus \{s_1',s_2',t_1',t_2'\}$, create two vertices $v^{+}$ and $v^{-}$.  We also create four vertices $s_1',s_2',t_1',t_2'$ in $G'$ and create four edges $s_1's_1^-,s_2's_2^-,t_1^+t_1',t_2^+t_2'$ in $G'$ all with length $0$. 
For every vertex $v\in V(G)\setminus \{s_1',s_2',t_1',t_2'\}$, let $in_1(v),\ldots,in_{r_v}(v)$ be the incoming edges of $v$. Then create $r_v$ parallel edges $e_1(v),\ldots,e_{r_v}(v)$ with $\text{tail}_{G'}(e_j(v))=v^-$ and $\text{head}_{G'}(e_j(v))=v^+$ in $G'$ for $j=1,\ldots,r_v$ such that each of the edges is of length $0$. If there is an edge $uv=in_p(v)$ in $G$ for some $p\in [r_v]$ and $u,v\notin\{s_1',s_2',t_1',t_2'\}$, create an edge $in_p(v^-)=u^{+}v^{-}$ in $G'$ and let $w(u^{+}v^{-})=w(uv)$.
Next, we define the transition system for $G'$ as follows. $T_{G'}(v^{-})=\{\{in_j(v^-), e_j(v)\} \mid  j\in [r_{v}]\}$. For every $e,e'\in (E_1\cup E_2)\setminus \{t_1t_1',t_2t_2'\}\subseteq E(G)$ such that $e=uv=in_p(v),e'=vw$ (let $\hat{e}=v^+w^-$), if $\{e, e'\} \in T_{G}(v)$, then $\{e_p(v), \hat{e}\} \in T_{G'}(v^{+})$. In particular, let $e_{i}=t_i^+t_i'$ for $i=1,2$. If $e=ut_i=in_q(t_i)\in E(G)$ for some $q\in [r_{t_i}]$, then $\{e_{q}(t_i), e_{i}\} \in T_{G'}(t_i^{+})$.

We also need to compute the set of edges $E_i'$ that exist in some shortest path (without transitions) from $s_i'$ to $t_i'$ for $i=1,2$. By this definition, obviously $s_i's_i^-,s_i^-s_i^+,t_i^-t_i^+,t_i^+t_i'\in E_i'$ for $i=1,2$.

\begin{lemma}
For $u,v\in V(G)\setminus \{s_1',s_2',t_1',t_2'\}$, $uv\in E_i$ if and only if $u^+v^-\in E_i'$ for $i=1,2$. Moreover, if some incoming edge of $v^-$ belongs to $E_i'$, then all of the parallel edges $v^{-}v^{+}$ belong to $E_i$ for $i=1,2$.
\end{lemma}
\begin{proof}
Suppose that $P_1=s_1,w,...,u,v,...,t_1$ is a shortest path from $s_1$ to $t_1$ in $G$. We claim that $P_1'=s_1',s_1^-,s_1^+,w^-,w^+,...,u^-,u^+,v^-,v^+,...,t_1^-,t_1^+,t_1'$ is a shortest path from $s_1'$ to $t_1'$ in $G'$. For contradiction, suppose the claim is not true. Then we can find a path ${P_0'}=s_1',s_1^-,s_1^+,w_1^-,w_1^+,...,w_\ell^-,w_\ell^+,t_1^-,t_1^+,t_1'$ in $G'$ such that $w(P_0')<w(P_1')=w(P_1)$. Then there is a path $P_0=s_1,w_1,...,w_{\ell},t_1$ in $G$ such that $w(P_0)=w(P_0')<w(P_1)$, contradicting that $P_1$ is a shortest path from $s_1$ to $t_1$.

Suppose that $P_1'=s_1',s_1^-,s_1^+,w^-,w^+,...,u^-,u+,v^-,v^+,...,t_1^-,t_1^+,t_1'$ is a shortest path from $s_1'$ to $t_1'$ in $G'$. We claim that $P_1=s_1,w,...,u,v,...,t_1$ is a shortest path from $s_1$ to $t_1$ in $G$. For contradiction, suppose that the claim is not true. Then there exists a path $P_0=s_1w_1...w_{\ell}t_1$ in $G$ such that $w(P_0)<w(P_1)=w(P_1')$. Thus there is a path ${P_0'}=s_1',s_1^-,s_1^+,w_1^-,w_1^+,...,w_\ell^-,w_\ell^+,t_1^-,t_1^+,t_1'$ in $G'$ such that $w(P_0')=w(P_0)<w(P_1')$, contradicting that $P_1'$ is a shortest path from $s_1'$ to $t_1'$ in $G'$.

Similarly we can show that $P_2=s_2,w,...,u,v,...,t_2$ is a shortest path from $s_2$ to $t_2$ in $G$ if and only if $P_2'=s_2',s_2^-,s_2^+,w^-,w^+,...,u^-,u^+,v^-,v^+,...,t_2^-,t_2^+,t_2'$ is a shortest path from $s_2'$ to $t_2'$ in $G'$. It follows that for $u,v\in V(G)\setminus \{s_1',s_2',t_1',t_2'\}$, $uv\in E_i$ if and only if $u^+v^-\in E_i'$ for $i=1,2$.

For $i=1,2$, as $w(v^{-}v^{+})=0$, we have that $d_i(v^{+})=d_i(v^{-})+w(v^{-}v^{+})$. 
Since some ingoing edge of $v^-$ belongs to $E_i'$, there is a $v^-t_i'$ path in $E_i'$. It follows that there is also a $v^+t_i'$ path in $E_i'$. By the definition of $E_i'$, all of the parallel edges $v^{-}v^{+}$ belong to $E_i'$.
\end{proof}

It's not hard to verify that Lemma~\ref{acyclic-V} and Lemma~\ref{contract-V} also apply to $G'$, but we will also state them here for clarity.

\begin{lemma}\cite{DBLP:conf/esa/Berczi017}\label{acyclic}
The edge set $E_i'$ forms no dicycle in $G'$ for $i=1,2$.
\end{lemma}

\begin{lemma}\cite{DBLP:conf/esa/Berczi017}\label{contract}
In the graph $G'$, suppose that $C$ is a dicycle in $E_1'\cup \overline{E_2'}$. Then $E_1'\cap E(C)\subseteq E_2'$ and $E_2'\cap \overline{E(C)}\subseteq E_1'$.
\end{lemma}

Let $E_0'=E_1'\cap E_2', E_1^*=E_1'\setminus E_0', E_2^*=E_2'\setminus E_0'$.  We contract all edges of $E_0'$ and get a graph $G''=(V'',E'')$. For an edge $e\in E''$, let $f(e)\in E(G')$ denote the edge corresponding to $e$ before the contracting operations.  We need to compute the new transition system of $G''$ as follows. Let $V_0'\subseteq V''$ be the set of vertices that are newly created after contracting $E_0'$. For $v\in V_0'$, we use $G'_v$ to denote the subgraph of $G'-(E(G')\setminus(E_1'\cup E_2'))$ induced by the vertices corresponding to $v$ before contracting. For every $u\in V(G'')\setminus V_0'$, if $f(e)f(e')\in T_{G'}(u)$ then $\{e, e'\} \in T_{G''}(u)$.  Let $v\in V_0'$ and $\text{head}_{G''}(e)=\text{tail}_{G''}(e')=v$. If there is a $T$-compatible path in the subgraph of $G'$ consisting of all edges of $G'_v$ together with $f(e)$ and $f(e')$ from $\text{tail}_{G'}(f(e))$ to $\text{head}_{G'}(f(e'))$, then $\{e, e'\} \in T_{G''}(v)$.  By Lemma~\ref{acyclic}, there is no dicycle in $G'_v$. Moreover, the subgraph of $G'$ consisting of all edges of $G'_v$ together with $f(e)$ and $f(e')$ is also acyclic. So we can compute $T_{G''}(v)$ for every $v\in V_0'$ in polynomial time according to Lemma~\ref{pathDAG}.  Since $E_1^*\cap E_2^*=\emptyset$, then we can reverse all edges of $E_2^*$ (the lengths of edges unchanged) with $E_1^*$ unchanged.  We get a new graph $G^*=(V^*,E^*)$, such that $V^*=V''$ and $E^*=E_1^*\cup \overline{E_2^*}$.

Then we also need to compute the new transition systems of $G^*$. If $e,g\in E_1^*$ and $\{e, g\}\in T_{G''}(v)$ for some $v\in V''$, then $\{e, g\} \in T_{G^{*}}(v)$. If $e,g\in E_2^*$ and $\{e, g\}\in T_{G''}(v)$ for some $v\in V''$, then $\{\bar{g}, \bar{e}\}\in T_{G^{*}}(v)$. Here we use $\bar{e},\bar{g}\in \overline{E_2^*}$ to denote the reverse of $e,g$ respectively.

\begin{claim}
After reversing the edges of $E_2^{*}$, there is no dicycle in $G^*$.
\end{claim}
\begin{proof}
Suppose for contradiction that there is a dicycle $C$ in $G^*$. By Lemma~\ref{acyclic}, $E(C)\not\subseteq E_1^*,E(C)\not\subseteq \overline{E_2^*}$. It follows that $E(C)\cap E_1^*\neq\emptyset$ and $E(C)\cap \overline{E_2^*}\neq\emptyset$. Then by Lemma~\ref{contract}, $E(C)$ should have been contracted in $G''$, contradicting that $C$ is a dicycle in $G^*$.
\end{proof}

We define a new digraph $\cal G$ as follows. Let $W=E_1^*\times \overline{E_2^*}$ be its vertex set. For $(e_1,e_2),(e_1',e_2')\in W$, there is a directed edge from $(e_1,e_2)$ to $(e_1',e_2')$ if one of three cases hold.
\begin{enumerate}[label=(\roman*)]
\item $e_1=e_1'$, $\text{head}_{G^*}(e_2)=\text{tail}_{G^*}(e_2')=v$ and $\{e_2, e_2'\}\in T_{G^*}(v)$. There is no path from $\text{head}_{G^*}(e_1)$ to $v$ in $G^*$.

\item $e_2=e_2'$, $\text{head}_{G^*}(e_1)=\text{tail}_{G^*}(e_1')=v$ and $\{e_1, e_1'\} \in T_{G^*}(v)$. There is no path from $\text{head}_{G^*}(e_2)$ to $v$ in $G^*$.

\item $\text{head}_{G^*}(e_2)=\text{tail}_{G^*}(e_2')=\text{head}_{G^*}(e_1)=\text{tail}_{G^*}(e_1')=v$ and $\{e_1, e_1'\}, \{e_2, e_2'\}\in T_{G^*}(v)$. Furthermore, if $v\in V_0$, let $G^{s}$ be the subgraph of $G'$ consisting of all edges of $G'_v$ together with $f(e_1), f(e_1'), f(\overline{e_2})$ and $f(\overline{e_2'})$. Then $G^{s}$ contains two $T$-compatible vertex-disjoint paths such that one path is from $\text{tail}_{G'}(f(e_1))$ to $\text{head}_{G'}(f(e_1'))$ and the other path is from $\text{tail}_{G'}(f(\overline{e_2'}))$ to $\text{head}_{G'}(f(\overline{e_2}))$.
\end{enumerate}

In the third case above, we claim that $v$ must belong to $V_0'$. Suppose for contradiction that $v\notin V_0'$. Clearly, $v \notin \{s'_1,s'_2,t'_1,t'_2\}$, as in must be both, head and tail of some edges. So there are two remaining cases. The first case is that $v=u^-$ for some $u\in V(G)$. Then all outgoing edges of $u^-$ in $G''$ are parallel edges, that is, $\text{head}_{G''}(\overline{e_2})=\text{head}_{G''}(e_1')$. Then $e_1'$ and $e_2$ form a cycle in $G^*$, contradicting that $G^*$ is acyclic. The second case is that $v=u^+$ for some $u\in V(G)$. Then all ingoing edges of $u^+$ in $G''$ are parallel edges, that is, $\text{tail}_{G''}(\overline{e_2'})=\text{tail}_{G''}(e_1)$. Then $e_1$ and $e_2'$ form a cycle in $G^*$, contradicting that $G^*$ is acyclic. Thus $v$ must belong to $V_0'$. Then we need to solve the vertex-disjoint version of \textsc{2-DSPP with transition restrictions} in the acyclic graph $G'_v\cup \{e_1,e_1',\overline{e_2},\overline{e_2'}\}$. The following lemma shows that we can do it in polynomial time. The algorithm is an adaption of the algorithm of finding two vertex-disjoint paths in DAG given by Perl and Shiloach~\cite{DBLP:journals/jacm/PerlS78}.

\begin{lemma}
In a directed acyclic graph $G=(V,E)$ with transition system $T_{G}$, we can solve the vertex-disjoint version of \textsc{2-DSPP with transition restrictions} in polynomial time.
\end{lemma}
\begin{proof}
First we modify the graph $G$ as follows. We create four vertices $s_1',s_2',t_1',t_2'$ and update $V(G)$ as $V(G)\leftarrow V(G)\cup\{s_1',s_2',t_1',t_2'\}$. We create four edges $\{s_1's_1,s_2's_2,t_1t_1',t_2t_2'\}$ and update $E(G)$ as $E(G)\leftarrow E(G)\cup\{s_1's_1,s_2's_2,t_1t_1',t_2t_2'\}$. Also, for $i=1,2$, we update $T_G(s_i)$ as 
\[T_G(s_i) \leftarrow T_G(s_i) \cup\{\{e, e'\} \mid e=s_i's_i\text{ and tail}_G(e')=s_i\},\] 
and we update $T_G(t_i)$ as 
\[T_G(t_i)\leftarrow T_G(t_i) \cup \{\{e, e'\} \mid e'=t_it_i'\text{ and head}_{G}(e)=t_i\}.\] 
For every vertex $v\in V(G)$, define the level $\ell(v)$ as the length of a longest directed path in $G$ starting from $v$. This can be computed by repeatedly removing a vertex of $G$.  Then we create a graph $\tilde{G}$ as follows. Let the vertex set of $\tilde{G}$ be $V(\tilde{G})=\{(e_1,e_2)\mid e_1,e_2\in E(G)\text{ and }e_1\neq e_2\}$. For every $(e_1,e_2),(e_1',e_2')\in V(\tilde{G})$, create an edge from $(e_1,e_2)$ to $(e_1',e_2')$ if one of the following cases holds:
\begin{enumerate}[label=(\arabic*)]
    \item $e_1=e_1'$, $\ell(\text{head}_{G}(e_2))\geq \ell(\text{head}_{G}(e_1))$, $\{e_2, e_2'\} \in T_G(\text{head}_{G}(e_2))$, $\text{head}_{G}(e_2')\neq \text{tail}_{G}(e_1)$ and $\text{head}_{G}(e_2')\neq \text{head}_{G}(e_1)$.
    \item $e_2=e_2'$, $\ell(\text{head}_{G}(e_1))\geq \ell(\text{head}_{G}(e_2))$, $\{e_1, e_1'\}\in T_G(\text{head}_{G}(e_1))$, $\text{head}_{G}(e_1')\neq \text{tail}_{G}(e_2)$ and $\text{head}_{G}(e_1')\neq \text{head}_{G}(e_2)$.
    \item $e_1=e_1'=t_1t_1'$, $\ell(\text{head}_{G}(e_2))<\ell(t_1')$, $\{e_2, e_2'\}\in T_G(\text{head}_{G}(e_2))$.
    \item $e_2=e_2'=t_2t_2'$, $\ell(\text{head}_{G}(e_1))<\ell(t_2')$, $\{e_1, e_1'\}\in T_G(\text{head}_{G}(e_1))$.
\end{enumerate}
We claim that there are two $T$-compatible vertex-disjoint paths $P_1$ and $P_2$ in $G$ such that $P_i$ is from $s_i'$ to $t_i'$ for $i=1,2$ if and only if there is a path $P$ from $(s_1's_1,s_2's_2)$ to $(t_1t_1',t_2t_2')$ in $\tilde{G}$.

(``only if'' direction): Let $P_1=e_1^0,e_1^1,\ldots,e_1^{p+1}$ and $e_1^0=s_1's_1,e_1^{p+1}=t_1t_1'$. Let $P_2=e_2^0,e_2^1,\ldots,e_2^{q+1}$ and $e_2^0=s_2's_2,e_2^{q+1}=t_2t_2'$. For any $i\in \{0,1,\ldots,p+1\},j\in\{0,1,\ldots,q+1\}$, such that $(i,j) \neq (p+1,q+1)$, one of the following four cases must hold.
\begin{itemize}
    \item[$\bullet$] $i\leq p$ and $j\leq q$, $\ell(\text{head}_G(e_1^i))\leq \ell(\text{head}_G(e_2^j))$, then there is an edge in $\tilde{G}$ from $(e_1^i,e_2^j)$ to $(e_1^i,e_2^{j+1})$.
    \item[$\bullet$] $i\leq p$ and $j\leq q$, $\ell(\text{head}_G(e_1^i))\geq \ell(\text{head}_G(e_2^j))$, then there is an edge in $\tilde{G}$ from $(e_1^i,e_2^j)$ to $(e_1^{i+1},e_2^j)$.
    \item[$\bullet$] $i=p+1$ and $j\leq q$, $\ell(\text{head}_G(e_2^j))<\ell(t_1')$, then there is an edge in $\tilde{G}$ from $(e_1^{p+1},e_2^j)$ to $(e_1^{p+1},e_2^{j+1})$.
    \item[$\bullet$] $j=q+1$ and $i\leq p$, $\ell(\text{head}_G(e_1^i))<\ell(t_2')$, then there is an edge in $\tilde{G}$ from $(e_1^i,e_2^{q+1})$ to $(e_1^{i+1},e_2^{q+1})$.
\end{itemize}
As a result, there is a path $P$ from $(s_1's_1,s_2's_2)$ to $(t_1t_1',t_2t_2')$ in $\tilde{G}$. This finishes the proof for ``only if'' direction.\\

(``if'' direction): Suppose that there exists a path $P$ from $(s_1's_1,s_2's_2)$ to $(t_1t_1',t_2t_2')$ in $\tilde{G}$. Let $P=(e_1^0,e_2^0),(e_1^1,e_2^1),\ldots,(e_1^r,e_2^r)$, such that $(e_1^0,e_2^0)=(s_1's_1,s_2's_2)$ and $(e_1^r,e_2^r)=(t_1t_1',t_2t_2')$. We construct two vertex-disjoint $T$-compatible paths $P_1,P_2$ as follows. First we initialize $P_1=e_1^0,P_2=e_2^0$. 
Then for $i=0,\ldots,r-1$, we update $P_1$ and $P_2$ according to the following cases:
\begin{itemize}
    \item[$\bullet$] Suppose that the edge from $(e_1^i,e_2^i)$ to $(e_1^{i+1},e_2^{i+1})$ is of type (1). Then $P_2\leftarrow P_2\cdot e_2^{i+1}$.
    \item[$\bullet$] Suppose that the edge from $(e_1^i,e_2^i)$ to $(e_1^{i+1},e_2^{i+1})$ is of type (2). Then $P_1\leftarrow P_1\cdot e_1^{i+1}$.
    \item[$\bullet$] Suppose that the edge from $(e_1^i,e_2^i)$ to $(e_1^{i+1},e_2^{i+1})$ is of type (3). Then $P_2\leftarrow P_2\cdot e_2^{i+1}$.
    \item[$\bullet$] Suppose that the edge from $(e_1^i,e_2^i)$ to $(e_1^{i+1},e_2^{i+1})$ is of type (4). Then $P_1\leftarrow P_1\cdot e_1^{i+1}$.
\end{itemize}
By the definition of edges of $\tilde{G}$, we get that $P_1$ and $P_2$ are two $T$-compatible vertex-disjoint paths in $G$ such that $P_i$ is from $s_i'$ to $t_i'$ for $i=1,2$. We can construct a graph $\tilde{G}$ in $O(|E|^{3})$ time and find a path from $(s_1's_1,s_2's_2)$ to $(t_1t_1',t_2t_2')$ in $O(|E|^{3})$ time. Thus the lemma holds.
\end{proof}

By the results above, we can construct $\cal G$ in polynomial time. Now we show that we can solve the vertex-disjoint version of \textsc{2-DSPP with transition restrictions} in $G$ by finding a path in $\cal G$ from $(s_1's_1^-,t_2't_2^+)$ to $(t_1^+t_1',s_2^-s_2')$. Note that $s_i^-,t_i^+\in V(G')$ might be the endpoints of edges of $E_0'$ for $i=1,2$. In this case, although we might contract the edges incident to $s_i^-,t_i^+\in V(G')$ and replace these vertices with new vertices, we slightly abuse $s_i^-,t_i^+$ to denote the vertex adjacent to $s_i',t_i'$ respectively in $G^*$ for $i=1,2$ for the sake of simplicity.

\begin{lemma} \label{mainV2DSPP}
There is a directed path in $\cal G$ from $(s_1's_1^-,t_2't_2^+)$ to $(t_1^+t_1',s_2^-s_2')$ if and only if $G'$ has two vertex-disjoint $T$-compatible paths $P_1$ and $P_2$ such that $P_i$ is from $s_i'$ to $t_i'$ and $P_i\subseteq E_i'$ for $i=1,2$.
\end{lemma}
\begin{proof}
(``if'' direction) Suppose that $G'$ has two vertex-disjoint $T$-compatible paths $P_1$ and $P_2$ such that $P_i$ is from $s_i'$ to $t_i'$ and $P_i\subseteq E_i'$ for $i=1,2$. Recall that we contract the edges of $E_0'$ in $G'$ and reverse the edges of $E_2^*$ in $G''$ to get $G^*$. So by the definition of transition systems of $G''$ and $G^*$, the set $E(P_1)\setminus E_0'$ forms a directed $T$-compatible path $P_1^*$ in $G^*$ from $s_1'$ to $t_1'$, and the set $\overline{E(P_2)\setminus E_0'}$ forms a directed $T$-compatible path $P_2^*$ in $G^*$ from $t_2'$ to $s_2'$. Let $P_1^*=e_1^0,e_1^1,\ldots,e_1^{p+1}$ and $e_1^0=s_1's_1^-,e_1^{p+1}=t_1^+t_1'$. Let $P_2^*=e_2^0,e_2^1,\ldots,e_2^{q+1}$ and $e_2^0=t_2't_2^+,e_2^{q+1}=s_2^-s_2'$. It follows that $e_1^i\in E_1^*$ for $i=0,1,\ldots,p+1$ and $e_2^j\in E_2^*$ for $j=0,1,\ldots,q+1$. Since $G^*$ is acyclic, for any $i=0,1,\ldots,p+1$ and for any $j=0,1,\ldots,q+1$, at least one of the following three cases holds.
\begin{enumerate}[label=(\arabic*)]
    \item  There is no directed path from $\text{head}_{G^*}(e_1^i)$ to $\text{head}_{G^*}(e_2^j)$ in $G^*$.
    \item  There is no directed path from $\text{head}_{G^*}(e_2^j)$ to $\text{head}_{G^*}(e_1^i)$ in $G^*$.
    \item  $\text{head}_{G^*}(e_1^i)=\text{head}_{G^*}(e_2^j)$.
\end{enumerate}
By the definition of $\cal G$, the following statements hold.
\begin{itemize}
    \item[$\bullet$] If (1) holds and $j\neq q+1$, then $\cal G$ has an edge from $(e_1^i,e_2^j)$ to $(e_1^i,e_2^{j+1})$.
    \item[$\bullet$] If (2) holds and $i\neq p+1$, then $\cal G$ has an edge from $(e_1^i,e_2^j)$ to $(e_1^{i+1},e_2^{j})$.
    \item[$\bullet$] If (3) holds, then $\cal G$ has an edge from $(e_1^i,e_2^j)$ to $(e_1^{i+1},e_2^{j+1})$.
\end{itemize}
We can see that if $i=p+1$ then (1) holds and if $j=q+1$ then (2) holds. As a result, there is an edge from $(e_1^i,e_2^j)$ to $(e_1^{i+1},e_2^j)$, $(e_1^i,e_2^{j+1})$ or $(e_1^{i+1},e_2^{j+1})$ in $\cal G$ if $(i,j)\neq (p+1,q+1)$. It follows that starting from $(e_1^i,e_2^j)$ with $i=0,j=0$, we can find a directed path ending at $(e_1^{p+1},e_2^{q+1})$ through increasing $i$ by $1$, increasing $j$ by $1$ or increasing both $i$ and $j$ by $1$ iteratively. This concludes the proof for ``if direction''.\\

(``only if'' direction) Suppose that there is a directed path from $(e_1^0,e_2^0)=(s_1's_1^-,t_2't_2^+)$ to $(e_1^r,e_2^r)=(t_1^+t_1',s_2^-s_2')$ in $\cal G$ that goes through $(e_1^0,e_2^0),(e_1^1,e_2^1),...,(e_1^r,e_2^r)$ consecutively. We construct two $T$-compatible paths $P_1,P_2$ in $G'$ as follows. 
First we initialize $P_1=e_1^0,P_2=\overline{e_2^0}$. Then for $i=0,...,r-1$, we update $P_1$ and $P_2$ according to the following three cases:
\begin{itemize}
    \item[$\bullet$] Suppose that the edge from $(e_1^i,e_2^i)$ to $(e_1^{i+1},e_2^{i+1})$ is of type (i), namely $e_1^i=e_1^{i+1}$, $\text{head}_{G^*}(e_2^i)=\text{tail}_{G^*}(e_2^{i+1})=v$ and $\{e_2^i, e_2^{i+1}\} \in T_{G^*}(v)$. There is no path from $\text{head}_{G^*}(e_1^i)$ to $v$ in $G^*$. If $v\in V_0'$, let $Q$ be the $T$-compatible path in the subgraph of $G'$ consisting of all edges of $G'_v$ together with $f(\overline{e_2^i})$ and $f(\overline{e_2^{i+1}})$ from $\text{tail}_{G'}(f(\overline{e_2^{i+1}}))$ to $\text{head}_{G'}(f(\overline{e_2^{i}}))$. Then $P_2\leftarrow f(\overline{e_2^{i+1}})\cdot Q\setminus\{f(\overline{e_2^i}),f(\overline{e_2^{i+1}})\}\cdot P_2$. Otherwise, if $v\notin V_0'$, $P_2\leftarrow f(\overline{e_2^{i+1}})\cdot P_2$.
    \item[$\bullet$] Suppose that the edge from $(e_1^i,e_2^i)$ to $(e_1^{i+1},e_2^{i+1})$ is of type (ii), namely $e_2^i=e_2^{i+1}$, $\text{head}_{G^*}(e_1^i)=\text{tail}_{G^*}(e_1^{i+1})=v$ and $\{e_1^i, e_1^{i+1}\} \in T_{G^*}(v)$. There is no path from $\text{head}_{G^*}(e_2^i)$ to $v$ in $G^*$.  If $v\in V_0'$, let $Q$ be the $T$-compatible path in the subgraph of $G'$ consisting of all edges of $G'_v$ together with $f(e_1^i)$ and $f(e_1^{i+1})$ from $\text{tail}_{G'}(f(e_1^{i}))$ to $\text{head}_{G'}(f(e_1^{i+1}))$. Then $P_1\leftarrow P_1\cdot Q\setminus\{f(e_1^i),f(e_1^{i+1})\}\cdot f(e_1^{i+1})$. Otherwise, if $v\notin V_0'$, $P_1\leftarrow P_1\cdot f(e_1^{i+1})$.
    \item[$\bullet$] Suppose that the edge from $(e_1^i,e_2^i)$ to $(e_1^{i+1},e_2^{i+1})$ is of type (iii), namely $\text{head}_{G^*}(e_1^i)=\text{tail}_{G^*}(e_1^{i+1})=\text{head}_{G^*}(e_2^i)=\text{tail}_{G^*}(e_2^{i+1})=v$ and $\{e_1^i, e_1^{i+1}\} \in T_{G^*}(v), \{e_2^i, e_2^{i+1}\} \in T_{G^*}(v)$. If $v\in V'_0$, let $G^{s}$ be the subgraph of $G'$ consisting of all edges of $G'_v$ together with $f(e_1), f(e_1'), f(\overline{e_2})$ and $f(\overline{e_2'})$. There are two $T$-compatible vertex-disjoint paths in $G^{s}$, namely $Q_1$ from $\text{tail}_{G'}(f(e_1^i))$ to $\text{head}_{G'}(f(e_1^{i+1}))$ and $Q_2$ from $\text{tail}_{G'}(f(\overline{e_2^{i+1}}))$ to $\text{head}_{G'}(f(\overline{e_2^{i}}))$. Then $P_1\leftarrow P_1\cdot Q_1\setminus\{f(e_1^i),f(e_1^{i+1})\}\cdot f(e_1^{i+1})$ and $P_2\leftarrow f(\overline{e_2^{i+1}})\cdot Q_2\setminus\{f(\overline{e_2^i}),f(\overline{e_2^{i+1}})\}\cdot P_2$. Otherwise, if $v\notin V'_0$, then $P_1\leftarrow P_1\cdot f(e_1^{i+1})$ and $P_2\leftarrow f(\overline{e_2^{i+1}})\cdot P_2$.
\end{itemize}
As a result, we construct two vertex-disjoint $T$-compatible paths $P_1$ and $P_2$ such that $P_i$ is from $s_i'$ to $t_i'$ and $P_i\subseteq E_i'$ for $i=1,2$. This finishes the proof for ``only if'' direction.
\end{proof}

Since $\cal G$ contains $O(|E|^{3})$ edges, we can detect a path in $\cal G$ in polynomial time. Thus Lemma~\ref{mainV2DSPP} shows that the vertex version of \textsc{2-DSPP with transition restrictions} can be solved in polynomial time assuming that every cycle in the input graph has positive length. 

%% file: conclusions.tex
We would like to conclude our work with posing one open problem that eluded us during this research. 
The NP-hardness reduction of Szeider~\cite{Szeider} for finding a (simple) path between two vertices
of a forbidden-transition graph can be easily modified to prove that it is also NP-hard
to find a (simple) compatible cycle in a forbidden-transition graph. 
In contrast, in edge-colored graphs finding any properly colored cycle is polynomial-time solvable~\cite{GrossmanHaggkvist,Yeo97}. But what about finding
a \emph{long} properly colored cycle?
More precisely, given an edge-colored graph $G$ and an integer $k$ we
ask whether $G$ admits a simple properly colored cycle of length at least $k$.
Is this problem fixed-parameter tractable when parameterized by $k$?
As the notion of properly colored walks in edge-colored graphs generalizes walks in directed
graphs, the problem in question is more general than finding a cycle of length at least $k$
in a directed graph.